\documentclass[twoside]{article}
\newcommand{\mytitle}{Root Statistics of Random Polynomials with
  Bounded Mahler Measure} 
\newcommand{\keywords}{Pfaffian point process, Mahler measure, random
  polynomial, eigenvalue statistics, skew-orthogonal polynomials,
  matrix kernel.}

\usepackage{amsmath,amssymb,
  amsthm,amsfonts,amscd,mathrsfs,pifont,upgreek,cancel}
\usepackage{eucal}  % Euler fonts
\usepackage[normalem]{ulem}
\usepackage{verbatim}
\usepackage{ifpdf}
\usepackage[colorlinks=true, citecolor=blue,backref=page]{hyperref}
\usepackage{subfig}

\ifpdf
        \usepackage[pdftex]{graphicx}
\else
        \usepackage[dvips]{graphicx}
\fi
\newtheorem{thm}{Theorem}[section]
\newtheorem{cor}[thm]{Corollary}

\newtheorem{lemma}[thm]{Lemma}
\newtheorem{prop}[thm]{Proposition}

\newtheorem{thm*}{Theorem}[]
\newtheorem{cor*}[thm*]{Corollary}
\newtheorem{claim*}[thm*]{Claim}
\newtheorem{lemma*}[thm*]{Lemma}
\newtheorem{prop*}[thm*]{Proposition}
\newtheorem{conj*}[thm*]{Conjecture}

\theoremstyle{definition}

\newtheorem{problem*}{Problem}[section]

\newtheorem{question*}{Question}[section]

\newtheorem{defn*}{Definition}

\theoremstyle{remark}
\newtheorem*{rem}{Remark}

\newcommand{\qq}[1]{\qquad \mbox{#1} \qquad}

\newcommand{\BB}[1]{\ensuremath{\mathbb{#1}}}

\newcommand{\HH}{\ensuremath{\BB{H}}}
\newcommand{\N}{\ensuremath{\BB{N}}}

\newcommand{\R}{\ensuremath{\BB{R}}}
\newcommand{\Z}{\ensuremath{\BB{Z}}}
\newcommand{\C}{\ensuremath{\BB{C}}}
\newcommand{\T}{\ensuremath{\BB{T}}}
\newcommand{\D}{\ensuremath{\BB{D}}}
\newcommand{\Om}{\ensuremath{\BB{O}}}

\newcommand{\bs}{\ensuremath{\boldsymbol}}

\newcommand{\wt}{\ensuremath{\widetilde}}

\newcommand{\Arg}{\ensuremath{\mathrm{Arg}}}
\newcommand{\upd}{\ensuremath{\mathrm{d}}}

\newcommand{\la}{\ensuremath{\langle}}
\newcommand{\ra}{\ensuremath{\rangle}}

\newcommand{\transpose}{\ensuremath{\mathsf{T}}}

\newcommand{\ip}[1]{\mathrm{Im}(#1)}
\newcommand{\rp}[1]{\mathrm{Re}(#1)}

\newcommand{\wh}[1]{\widehat{#1}}

\newcommand{\G}[1]{\Gamma\left( #1 \right)}

\DeclareMathOperator{\sgn}{sgn}
\DeclareMathOperator{\vol}{vol}

\DeclareMathOperator{\Tr}{Tr}

\DeclareMathOperator{\Pf}{Pf}

\numberwithin{equation}{section} 

\numberwithin{equation}{section}

\begin{document}
\title{\bfseries\sffamily \mytitle}  
\author{{\sc Christopher D.~Sinclair} and {\sc Maxim L.~Yattselev}}
\maketitle

\begin{abstract}
  The Mahler measure of a polynomial is a measure of complexity formed
  by taking the modulus of the leading coefficient times the modulus
  of the product of its roots outside the unit circle.  The roots of a
  real degree $N$ polynomial chosen uniformly from the set of
  polynomials of Mahler measure at most 1 yields a Pfaffian point
  process on the complex plane.  When $N$ is large, with probability
  tending to 1, the roots tend to the unit circle, and we investigate
  the asymptotics of the scaled kernel in a neighborhood of a point on
  the unit circle.  When this point is away from the real axis (on
  which there is a positive probability of finding a root) the scaled
  process degenerates to a determinantal point process with the same
  local statistics ({\em i.e.} scalar kernel) as the limiting process
  formed from the roots of complex polynomials chosen uniformly from
  the set of polynomials of Mahler measure at most 1.  Three new matrix
  kernels appear in a neighborhood of $\pm 1$ which encode information
  about the correlations between real roots, between complex roots and
  between real and complex roots.  Away from the unit circle, the
  kernels converge to new limiting kernels, which imply among other
  things that the expected number of roots in any open subset of $\C$
  disjoint from the unit circle converges to a positive number.  We
  also give ensembles with identical statistics drawn from
  two-dimensional electrostatics with potential theoretic weights, and
  normal matrices chosen with regard to their topological entropy as
  actions on Euclidean space.
\end{abstract}


\vspace{.25cm}


\vspace{.5cm}

\tableofcontents

\section{Introduction}

The study of roots of random polynomials is an old subject extending
back at least as far as the early 1930s.  Several early results
revolve around estimating, as a function of degree $N$, the number of
real roots of polynomials with variously proscribed integer or real
coefficients.  These begin with Bloch and P\'olya, who gave bounds for
the maximum number of real roots of polynomials with coefficients in
$\{-1, 0, 1\}$ (their lower bound for this maximum being
$O(N^{1/4}/\sqrt{\log N})$) and showed that the expected number of
real roots does not exceed is $O(\sqrt{N})$ \cite{MR1576817}. Shortly
thereafter, Littlewood and Offord proposed the same sort of questions,
but for polynomials with independent standard normal coefficients, and
for coefficients chosen uniformly from $[-1,1]$ or $\{-1, 1\}$.  They
proved that the expected number of real roots in these cases is
eventually bounded by $25 (\log N)^2 + 12 \log N$, but were unable to
determine if this bound was of the right order \cite{MR1574980,
  PSP:2031108}.  In the 1940s, Kac determined not only the correct
asymptotic ($2 \log N/\pi$) for the expected number of real roots in
the i.i.d.~normal case, but in
fact gave, for each fixed $N$, an explicit function when integrated
over an interval gives the expected number of real roots in that
interval \cite{MR0007812}. (Such a function is called an {\em
  intensity} or {\em correlation function}, and will be central in the
present work).  Kac later extended his results for coefficients which
are (arbitrary) i.i.d. continuous random variables with unit variance
\cite{MR0030713} (in particular the asymptotic estimate remains
unchanged in this situation).  Since then, many results on the
expected number of real roots of random polynomials have been
presented for various meanings of the word `random'; of particular
note are \cite{MR0073870}, \cite{MR1231689} and \cite{MR1915821}.

An obvious question beyond `how many roots of a random polynomial are
real?' is `where do we expect to find the roots of a random
polynomial?'.  Certainly, knowing that some expected number of roots
are real gives some information about where we expect to find them.
Moreover, Kac's formula for the intensity, when applicable, gives
detailed information about where the real roots are expected to
reside.  In the 1990s Shepp and Vanderbei extended Kac's intensity
result to the complex roots of random polynomials with i.i.d.~real
coefficients by producing a complimentary intensity supported on $\C
\setminus \R$ \cite{MR1308023}.  This together with Kac's intensity
specifies the spatial density of roots of random Gaussian complex
polynomials, and in particular shows that such roots have a tendency,
when $N$ is large, to clump near the unit circle.  Some part of this
observation was made much earlier---the early 1950s---by Erd\H{o}s and
Tur\'an, who prove (in their own paraphrasing) 
\begin{quote}
$\ldots$ that the roots of a polynomial are
uniformly distributed in the different angles with vertex at the
origin if the coefficients ``in the middle'' are not too large
compared with the extreme ones \cite{MR0033372}.
\end{quote}
Strictly speaking, the result of Erd\H{o}s and Tur\'an is not a result
about random polynomials, but rather gives an upper bound, as a
function of the coefficients of the polynomials, for the difference
between the number of roots in an angular segment of a polynomial from
the number assuming radial equidistribution, .  This can be translated
into a result about random polynomials given information about the
distribution of coefficients.

Erd\H{o}s and Tur\'an's result presaged the fact that for many types
of random polynomials, the zeros have a tendency to be close to
uniformly distributed on the unit circle, at least when the degree is
large.  One way of quantifying this accumulation is to form a
probability measure from a random polynomial by placing equal point
mass at each of its roots.  This measure is sometimes called the {\em
  empirical} measure, and given a sequence of polynomials of increasing
degree we can ask whether or not the resulting sequence of empirical
measure weak-$\ast$ converges (or perhaps in some other manner) to
uniform measure on the unit circle (or perhaps some other measure).
Given a random sequence of such polynomials we can then investigate in
what probabilistic manner (almost surely, in probability, etc.) this
convergence occurs, if it occurs at all.  Another way of encoding
convergence of roots to the unit circle (or some other region) can be
given by convergence of intensity/correlation functions (assuming such
functions exist).  For Gaussian polynomials this convergence of
intensity functions appears in the work of Shepp and Vanderbei, and
was later extended to i.i.d.~coefficients from other stable
distributions (and distributions in their domain of attraction) by
Ibragimov and Zeitouni in \cite{MR1390040}.  For more general
conditions which imply convergence of roots to the unit circle, see
the recent work of Hughes and Nikeghbali \cite{MR2422348}.

The random polynomials we consider here will {\em not} have
i.i.d.~coefficients---a situation first considered rigorously in
generality by Hammersley\cite{MR0084888}.  Our polynomials will be
selected uniformly from a certain compact subset of $\R^{N+1}$ (or
$\R^N$) as identified with coefficient vectors of degree $N$
polynomials (or monic degree $N$ polynomials).  Specifically we will
be concerned with polynomials chosen uniformly from the set with
Mahler measure at most 1.  Definitions will follow, but for now we
define the Mahler measure of a polynomial to be the absolute value of
the leading coefficient times the modulus of the product of the roots
outside the unit circle\footnote{It goes without saying that Mahler
  measure is not a measure in the sense of integration theory.
  Perhaps a better name would be `Mahler height' but that ship has
  already sailed.}.  The set of coefficient vectors of degree $N$
polynomials with Mahler measure at most 1, which we denote for the
moment by $B^{(N)}$, is a compact subset of $\R^{N+1}$, and we will
primarily be concerned with the roots of polynomials and monic
polynomials chosen uniformly from this region, especially in the limit
as $N \rightarrow \infty$.  We remark that Mahler measure is
homogeneous, and thus the set of polynomials with Mahler measure
bounded by $T > 0$ is a dilation (or contraction) of $B^{(N)}$.  That
is, whether one chooses uniformly from the set of Mahler measure at
most 1 or $T$, the distribution of roots is the same.  (This latter
fact is {\em not} true for monic polynomials).

The choice of this region is not arbitrary; Mahler measure is an
important {\em height}, or measure of complexity, of polynomials, and
appears frequently in the study of integer polynomials and algebraic
numbers.  Of particular note in this regard (a result of Kronecker,
though not phrased in this manner) is that the set of integer
polynomials with Mahler measure equal to 1 is exactly equal to the
product of monomials and cyclotomic polynomials; that is an integer
polynomial with Mahler measure 1 has all roots in $\T \cup \{0\}$
where $\T$ is the unit circle \cite{kron1957}.  An unresolved problem,
posed by D.H.~Lehmer in 1933, is to determine whether or not 1 is a
non-isolated point in the range of Mahler measure restricted to
integer polynomials\footnote{The current reigning champion
  non-cyclotomic polynomial with smallest Mahler measure is $z^{10} +
  z^9 - z^7 - z^6 - z^5 - z^4 - z^3 + z + 1$ and has Mahler measure
  $\approx 1.18$.  Remarkably, this polynomial was discovered by
  Lehmer in 1933, and has survived the advent of computers.}
\cite{MR1503118}.  On one hand, this is a question about how the sets
$B^{(N)}$ are positioned relative to the integer lattices $\Z^N$---in
particular if we denote by $r_N > 1$ the smallest number such that the
dilated star body $r_N B^{(N)}$ contains an integer polynomial with a
non-cyclotomic factor, Lehmer's question reduces to whether or not
$r_N \rightarrow 1$.  On the other hand, since Mahler measure is a
function of the roots of a polynomial, Lehmer's problem can be
translated as a question about how quickly the roots of a sequence of
non-cyclotomic polynomials can approach the unit circle.

One motivation for studying the roots of polynomials chosen uniformly
from $B^{(N)}$ is that such results might suggest analogous results
for integer polynomials with small Mahler measure.  
Lehmer's problem has been resolved for various classes of polynomials,
for instance a sharp lower bound for the Mahler measure of
non-cyclotomic polynomials with all real roots has been given by
Schinzel \cite{MR0360515}, and a sharp lower bound for non-cyclotomic,
non-reciprocal polynomial was given by Smyth \cite{MR0289451}; both of
these results appeared in the early 1970s, and reflect the fact that
polynomials with small Mahler measure have roots which are in some
manner constrained.  Along these lines, a result of Langevin from the
1980s says that, the roots of a sequence of irreducible integer
polynomials with unbounded degree and bounded Mahler measure cannot
avoid any open set in $\C$ which contains a point on the unit circle
\cite{MR812558}.  In fact, the result is stronger: any sequence of
irreducible integer polynomials with unbounded degree whose roots
avoid such a set have Mahler measure which grows exponentially with
the degree.  A more recent result of Bilu, and one which is
similar in spirit to that of Erd\H{o}s and Tur\'an, states that the
empirical measures of any sequence of irreducible, integer polynomials
with increasing degree and Mahler measure tending to 1 converges
weak-$\ast$ to uniform measure on the unit circle \cite{MR1470340}.

Another motivation for studying the roots of polynomials chosen
uniformly from $B^{(N)}$ comes from random matrix theory.  Indeed, the
results presented here will be familiar, in form at least, to results
about the eigenvalues of random matrices.  It is beyond the scope of
the current work to present a survey of random matrix theory (see
however the collection \cite{oxford} for a glimpse into the current
state of random matrix theory and its applications).  We will,
however, underscore two major themes here.  The first is, for certain
very well-behaved ensembles\footnote{{\em Ensemble} is physics
  parlance for a probability space.} of random matrices, the
intensity/correlation functions are expressible in terms of
determinants or Pfaffians of matrices formed from a single object---a
{\em kernel}.  Pioneering work on such {\em determinantal} and {\em
  Pfaffian} ensembles was done in the 1960s and 1970s by Mehta
\cite{MR0112645, MR0112895, MR0277221}, Gaudin \cite{MR0112895,
  Gaudin1961447} and Dyson \cite{MR0143558, MR0278668}.  For these
ensembles, this kernel depends on the size of the matrices, and the
limiting eigenvalue statistics can be determined from the limit of
this kernel in various scaling regimes.  The scaled limits of these
kernels yield {\em limiting, local} eigenvalue statistics on scales
where the eigenvalues have some prescribed spacing or density.  We
will find that our ensemble of random polynomials has such a kernel
(of the Pfaffian form) and will present the scaling limits of this
kernel here.

A second major theme in random matrix theory is that of {\em
  universality}.  Loosely stated, universality says that the limiting,
local statistics of eigenvalues of ensembles of random matrices fall
into one of a handful of {\em universality classes} \footnote{We are
  being woefully imprecise here.  For instance, Lubinsky
  \cite{Lub09} showed that each reproducing kernel of a de Branges
  space gives rise to a universality class that may arise in the bulk
  (that is, in the limiting support of eigenvalues) for some unitary
  ensemble; however, he also showed \cite{Lub12} that in measure in
  the bulk it is always the universality class of the ``sine
  kernel''.} based on large scale structure (for instance symmetries
on the entries which may geometrically constrain the position of
eigenvalues) but largely independent of the actual distribution on the
matrices.  (See Kuijlaars' essay \cite{MR2932626} for a more precise
definition of universality).  Thus, a universality class is akin to
the basin of attraction for stable distributions.

Universality is important in combination with the identification of
ensembles whose limiting local statistics are well understood.  Such
ensembles, for instance those which have kernels whose scaling limits
are explicitly described, play the role of stable distributions, in
the sense that they provide prototypes for their universality class.
We will present an ensemble of random matrices whose eigenvalue
statistics are identical to those of the roots of polynomials chosen
at random from $B^{(N)}$.  Thus, the results here have repercussion
beyond just the study of the statistics of roots of polynomials with
bounded Mahler measure, but also as a prototypical ensemble in a newly
discovered universality class.

\subsection{Mahler measure}

The Mahler measure of a polynomial $f(z) \in \C[z]$ is given by
\[
M(f) = \exp\left\{
\int_0^1 \log|f(e^{2 \pi i \theta})| \, \upd\theta
\right\}.
\]
By Jensen's formula, if $f$ factors as $f(z) = a \prod_{n=1}^N (z - \alpha_n)$, then
\begin{equation}
\label{eq:25}
%f(z) = a \prod_{n=1}^N (z - \alpha_n) \qq{then} 
M(f) = |a|\prod_{n=1}^N \max\big\{1, |\alpha_n| \big\}.
\end{equation}
Mahler measure is not a measure in the sense of measure theory, but an example of a height---a measure of
complexity---of polynomials, and is of primary interest when
restricted to polynomials with integer (or other arithmetically
defined) coefficients.  

In \cite{MR1868596}, Vaaler and Chern compute the volume (Lebesgue
measure) of the set of real degree $N$ polynomials (as identified with
a subset of $\R^{N+1}$ via their coefficient vectors).  This volume
arises in the main term for the asymptotic estimate for the number of
integer polynomials of degree $N$ with Mahler measure bounded by $T$
as $T \rightarrow \infty$.  Amazingly, Chern and Vaaler's calculation
showed that this volume was a rational number with a simple product
description for each $N$.  A similar simple expression was found for the
related volume of monic polynomials.

The first author in \cite{sinclair-2005} gave a Pfaffian formulation
for Chern and Vaaler's product formulation, which was later shown to
be related to a normalization constant for an ensemble of random
matrices in \cite{sinclair-2007}.  

The purpose of this article is to explore the statistics of zeros of
polynomials chosen at random from Chern and Vaaler's volumes,
especially in the limit as $N \rightarrow \infty$ in an appropriate
scaling regime.  In fact, we will look at a natural one-parameter
family of ensembles of random polynomials which interpolate between
the volumes of monic and non-monic polynomials considered by Chern and
Vaaler.

We will also introduce ensembles of random matrices and a
two-dimensional electrostatic model which have the same statistics as
our ensembles of random polynomials.  

\subsection{Volumes of Star Bodies}

Mahler measure is an example of a distance function in the sense of
the geometry of numbers, and therefore, when restricted to coefficient
vectors of degree $N$ polynomials satisfies all the axioms of a vector
norm except the triangle inequality.  Specifically, $M$ is continuous,
positive definite and homogeneous: $M(c f) = |c| M(f)$.  We will
generalize the situation slightly by introducing a parameter $\lambda
\geq 0$ and define the $\lambda$-{\em homogeneous} Mahler measure by
\begin{equation}
\label{eq:1}
M_{\lambda}\left(a \prod_{n=1}^N (z - \alpha_n)\right) = |a|^{\lambda}
\prod_{n=1}^N \max\big\{1, |\alpha_n| \big\}.
\end{equation}
Generalizing in this manner, $M_{\lambda}$ is no longer continuous as
a function of coefficient vectors (except when
$\lambda=1$) however, as we shall see, the parameter $\lambda$ will appear
naturally in our subsequent calculations.  

The `unit-ball' of $M$  is not convex, but rather is a symmetric star
body.  We define this set as
\[
B_{\lambda}  = \left\{ \boldsymbol a \in \R^{N+1} : M_{\lambda}\left(\sum_{n=0}^N a_n
      z^n\right) \leq 1 \right\},
\]
and we define the star body of radius $T > 0$ by $B_{\lambda}(T) =
T^{1/\lambda} B_{\lambda}$.  We also define the related sets of monic
polynomials given by $\wt B = \wt B(1)$ where
\[
\wt B(T) = \left\{ \boldsymbol b \in \R^N : M_{\lambda}\left(z^N +
  \sum_{n=0}^{N-1} b_n z^n\right) \leq T \right\}.
\]
Note that, when restricted to monic polynomials, $M_{\lambda} = M$,
and $\wt B$ corresponds to the set of monic degree $N$ real
polynomials with all roots in the closed unit disk.  We remark that
$B_{\lambda}$ and $\wt B$ are rather complicated geometrically.  For
instance, note that both $(z-1)^N$ and $z^N - 1$ lie in $\wt B$
while $(z^{N-10} - 1)(z^{10} + z^9 - z^7 - z^6 -z^5-z^4 -z^3 + z + 1)$ 
is {\em not} in $\wt B$.  That is, $\wt B$ contains both points of large
and small 2-norm ($2^{N/2}$ and $\sqrt{2}$, respectively), but there also
exist integer points of relatively small norm which are not in $\wt B$.  

For the convenience of the next calculation, we shall write
$M_{\lambda}(c,\boldsymbol b)$ and $\wt M(\boldsymbol b)$ for
\[
M_{\lambda}(c,\boldsymbol b) = M_{\lambda}\left(cz^N+\sum_{n=0}^{N-1} b_n z^n\right)
\]
and
\[
\wt M(\boldsymbol b) = M_{\lambda}(1, \boldsymbol b) =
M_{\lambda}\left(z^N + \sum_{n=0}^{N-1} b_n z^n\right). 
\]
The second of these functions, we shall call the {\em monic} Mahler
measure (and is obviously independent of $\lambda$).  

The following manipulations are elementary, but of central importance
in Chern and Vaaler's determination of the volume of $B_{\lambda}$
and $\wt B$, and in our analysis of the statistics of roots
of polynomials chosen uniformly from these sets.  By volume, we shall of
course mean Lebesgue measure.  
\begin{align*}
\vol B_{\lambda} &= \int_{-\infty}^{\infty} \vol \big\{ \boldsymbol b: M_{\lambda}(c, \boldsymbol b) \leq 1 \big\} \, \upd c \\
&= \int_{-\infty}^{\infty} \vol \big \{c \boldsymbol b: M_{\lambda}(c, c \boldsymbol b) \leq 1 \big \} \, \upd c \\
& = 2 \int_0^{\infty} c^N \vol\left\{\boldsymbol b: \wt M(\boldsymbol b) \leq c^{-\lambda} \right\} \, \upd c,
\end{align*}
where at this last step we have used the $\lambda$-homogeneity of
$M_{\lambda}$.  The change of variables $\xi = c^{-\lambda}$ gives 
\[
\vol B_{\lambda} = \frac{2}{\lambda} \int_0^{\infty} \xi^{-(N+1)/\lambda}
\vol\left\{\boldsymbol b: \wt M(\boldsymbol b) \leq \xi \right\}  \, \frac{\upd\xi}{\xi}.
\]
If we denote the distribution function of $\wt M$ by $f_N(\xi) := \vol
\wt B(\xi) $, then 
%\[
%f_N(\xi) = \vol \wt B(M; \xi) 
%\]
%then
\begin{equation}
\label{eq:23}
\vol B_{\lambda} = \frac{2}{\lambda} \int_0^{\infty} \xi^{-s-1} f_N(\xi)  \,
\upd\xi, \qquad s = \frac{N + 1}{\lambda}.
\end{equation}
The astute reader will notice the appearance of the Mellin transform
of $f_N$ appearing in this expression for the volume of $B_{\lambda}$.  
Interpreting the integral in (\ref{eq:23}) as a Lebesgue-Stieltjes
integral we can apply integration by parts to find
\[
\vol B_{\lambda} = \frac{2}{N+1} F(s), \qquad F(s) = \int_{\R^N} \wt M(\boldsymbol b)^{-s} \, \upd\mu_\R^N(\boldsymbol b),
\]
where $\mu_\R$ is Lebesgue measure on $\R$ and $\mu_{\R}^N$ the
resulting product measure on $\R^N$. Clearly, also the volume of the
set of monic polynomials with Mahler measure at most 1 is given by
\[
\vol \wt B = \lim_{s \rightarrow \infty} F(s). 
\]

By a random polynomial we shall mean a polynomial chosen with respect
to the density $M(\boldsymbol b)^{-s}/F(s)$, where we view $s$ as a
parameter.  The goal of this paper is to explore the statistics of
roots of random polynomials in the limit as $N \rightarrow \infty$ and
$(N+1)/s \rightarrow \lambda$.  This is equivalent to studying the
statistics of zeros of polynomials chosen uniformly from $B_{\lambda}$
as the degree goes to $\infty$.  

\begin{figure}[!ht]
\centering
\includegraphics[scale=.38]{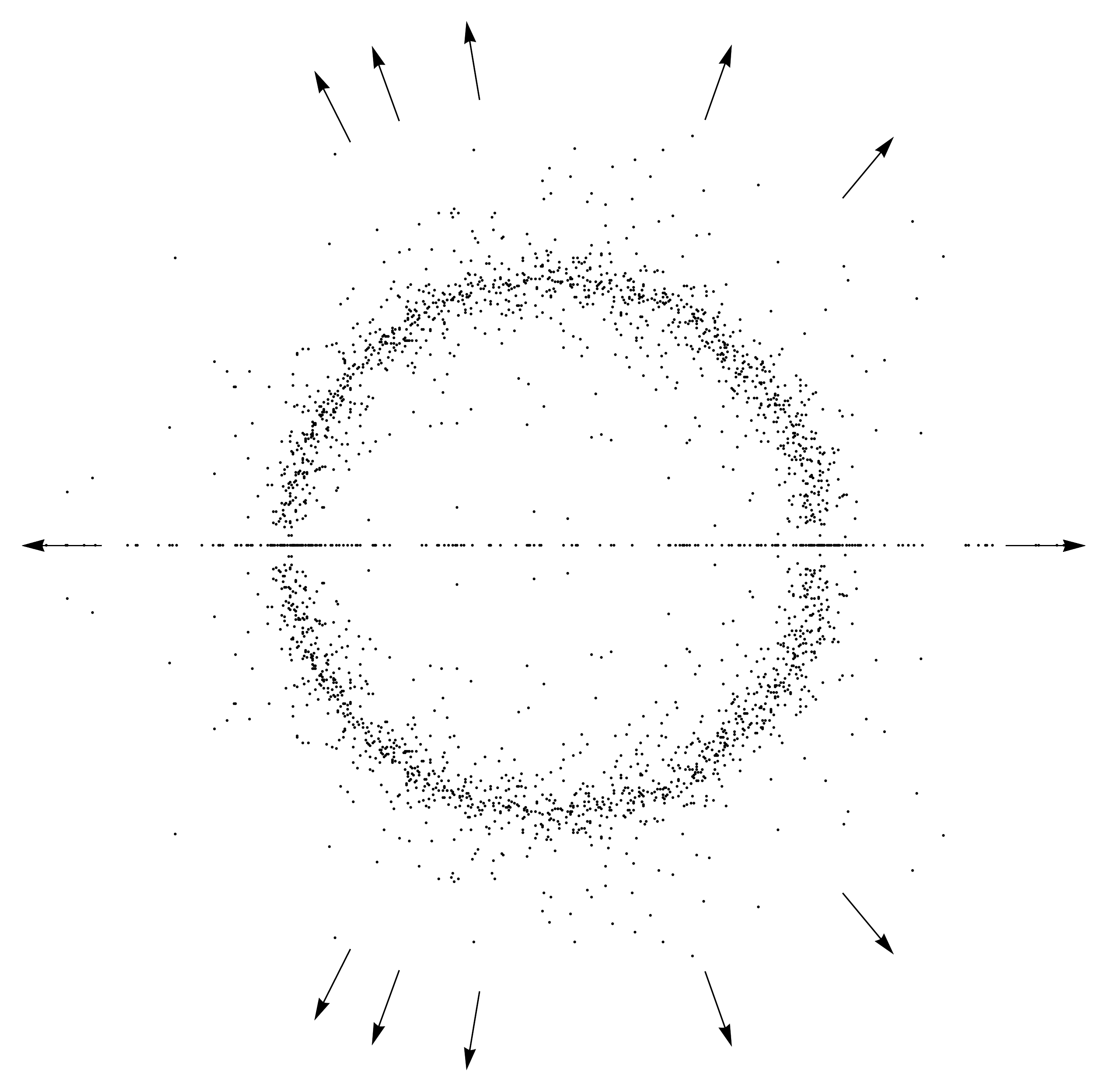}
\caption{A simultaneous plot of the roots of 100 random
  polynomials of degree 28 corresponding to $\lambda = 1$.  Note that, since we are
  sampling uniformly from a region in $\R^N$, the complicated
  geometric nature of $\widetilde B$ makes it difficult to accurately sample
  a truly (pseudo) random polynomial.  This example was achieved by
  doing (for each polynomial) a ball-walk of 10,000 steps of length
  $.01$ starting from $x^{28}$.  The arrows indicate directions of
  outlying roots.}
\label{fig:3}
\end{figure}

\subsection{The Joint Density of Roots}
\label{sec:joint-density-roots}

Since the coefficients of our random polynomials are real, the roots
are either real or come in complex conjugate pairs; that is we may
identify the set of roots of degree $N$ polynomials with
\[
\bigcup_{(L, M) \atop L + 2M = N} \R^L \times \C^M.
\]
In the context of our random polynomials $L$ and $M$ are integer
valued random variables representing the number of real and complex
conjugate pairs of roots respectively.  The density on coefficients
induces a different density on each component $\R^L \times
\C^M$---that is the joint density on roots is naturally described as a
set of conditional densities based on the number of real roots.  

The map $E_{L,M}: \R^L \times
\C^M \rightarrow \R^N$ specified by $(\bs \upalpha, \bs \upbeta)
\mapsto \boldsymbol b$ where
\[
\prod_{l=1}^L \big(x - \alpha_l\big) \prod_{m=1}^M \big(x - \beta_m\big)\big(x - \overline
\beta_m\big) = x^N + \sum_{n=1}^N b_n x^{N-n} 
\]
is the map from roots to coefficients, and a generic $\boldsymbol b \in
\R^N$ has $2^M M! L!$ preimages under this map.  The conditional joint
density of roots is determined from the Jacobian of $E_{L,M}$, which
was computed in \cite{sinclair-2005}.  Specifically, the conditional
joint density of roots of random polynomials with exactly $L$ real
roots is given by $P_{L,M} : \R^L \times \C^M \rightarrow 
[0,\infty)$ where
\begin{equation}
\label{eq:6}
P_{L,M}(\bs \upalpha, \bs \upbeta) = \frac{2^M}{Z_{L,M}} \prod_{l=1}^L \max\big\{1,
|\alpha_l|\big\}^{-s} \prod_{m=1}^M \max\big\{1, |\beta_m|\big\}^{-2s} \left|
  \Delta(\bs \upalpha, \bs \upbeta) \right|,
\end{equation}
and $\Delta(\bs \upalpha, \bs \upbeta)$ is the Vandermonde
determinant in the variables $\alpha_1, \ldots, \alpha_L$, $\beta_1, 
\overline \beta_1, \ldots, \beta_M, \overline \beta_M$, and $Z_{L,M}$
is the conditional partition function\footnote{{\em Partition
    function} is physics parlance for ``normalization constant.''},
\[
Z_{L,M}(s) = 2^M \int_{\R^L} \int_{\C^M} \prod_{l=1}^L \max\big\{1,
|\alpha_l|\big\}^{-s} \prod_{m=1}^M \max\big\{1, |\beta_m|\big\}^{-2s} \left|
  \Delta(\bs \upalpha, \bs \upbeta) \right| \, \upd\mu_{\R}^L(\bs \upalpha) \,
  \upd\mu_{\C}^M(\bs \upbeta)
\]
($\mu_{\R}, \mu_{\R}^L, \mu_{\C}, \mu_{\C}^M$ are Lebesgue measure on
$\R, \R^L, \C$ and $\C^M$ respectively).

The total partition function is then given by
\begin{equation}
\label{eq:2}
Z(s) = \sum_{(L,M) \atop L + 2M = N} \frac{Z_{L,M}(s)}{2^M L! M!} = F(s).
\end{equation}
That is, the total partition function of the system is, up to a
trivial constant, the volume of $B_{(N+1)/s}$. 
\begin{thm}[S.-J. Chern, J. Vaaler \cite{MR1868596}]
Let $J$ be the integer part of $N/2$, and suppose $s > N$.  Then,
\begin{equation}
\label{eq:3}
F(s) = C_N \prod_{j=0}^J \frac{s}{s - (N-2j)} \qq{where} C_N = 2^N
\prod_{j=1}^J \left(\frac{2j}{2j+1}\right)^{N-2j}.
\end{equation}
\end{thm}
This theorem is surprising, in part because of the simplicity of the
representation as a function of $s$, but also because by (\ref{eq:2}),
$F(s)$ is initially given as a rather complicated sum whereas
(\ref{eq:3}) reveals it as a relatively simple product.  It should be
remarked that constituant summands, $Z_{L,M}(s)$, are {\em not} simple,
and Chern and Vaaler's remarkable product identity followed from
dozens of pages of complicated rational function identities.  

A conceptual explanation for Chern and Vaaler's product formulation of
$F(s)$ is given by the following theorem.
\begin{thm}[Sinclair \cite{sinclair-2005}] 
\label{thm:4}
Suppose $N$ is even, $p_0, p_1, \ldots, p_{N-1}$ are any family of
monic polynomials with $\deg p_n = n$ and define the skew-symmetric
inner products  
\begin{equation}
\label{eq:4}
\la f | g \ra_{\R} = \int_{\R} \int_{\R} \max\big\{1,
|x|\big\}^{-s} \max\big\{1, |y|\big\}^{-s} f(x) g(y) \sgn(y - x) \, \upd\mu_\R(x)
\upd\mu_\R(y)
\end{equation}
and
\begin{equation}
\label{eq:5}
\la f | g \ra_{\C} = -2\mathrm{i} \int_{\C} \max\big\{1, |z|\big\}^{-2s}
\overline{f(z)} g(z) \sgn(\ip{z}) \, \upd \mu_{\C}(z).
\end{equation}
Then,
\[
F(s) = \Pf \mathbf U \qq{where} \mathbf U = \left[\la p_{m-1} |
  p_{n-1} \ra_{\R} + \la p_{m-1} | 
  p_{n-1} \ra_{\C} \right]_{m,n=1}^N,
\]
and $\Pf \mathbf U$ denotes the Pfaffian of the antisymmetric matrix
$\mathbf U$.\footnote{The Pfaffian is an invariant
  of antisymmetric matrices with an even number of rows and columns.
  For our purposes here it suffices to note that the square of the
  Pfaffian is the determinant.}
\end{thm}
A similar formulation is valid when $N$ is odd, but for the purposes
of exposition, the details are unimportant here.  

The reason this theorem suggests a product formulation like
(\ref{eq:3}) for $F(s)$ is that the independence of $F(s)$ from the
specifics of $\{p_n\}$ means that by choosing the polynomials to be
skew-orthogonal with respect to $\la p_{m-1} | p_{n-1} \ra_{\R} + \la p_{m-1} |
  p_{n-1} \ra_{\C}$, (see Section~\ref{sec:family-skew-orth} below for
  the definition of skew-orthogonal polynomials), $\mathbf U$ looks like
\[
\begin{bmatrix}
0 & r_1 \\
-r_1 & 0 \\
 & & 0 & r_2 \\
 & & -r_2 & 0 \\
& & & & \ddots \\
& & & & & 0 & r_J \\
& & & & & -r_J & 0
\end{bmatrix} \qq{and} \Pf \mathbf U = \prod_{j=1}^J r_j.
\]
One of the results presented here is an explicit description of the
skew-orthogonal polynomials which make Chern and Vaaler's formula a
trivial consequence of Theorem~\ref{thm:4}.  

\subsection{Pfaffian Point Processes}

The fact that $F(s)$ can be written as a Pfaffian not only gives a
simple(r) proof of Chern and Vaaler's volume calculation, but it also
is the key observation necessary to show that the point process on the
roots is a Pfaffian point process.  We recall a few definitions and
salient features of this Pfaffian point process here.  

Loosely speaking, the types of questions we are interested in are
along the lines of:  Given non-negative integers $m_1, m_2, \ldots, m_n$
and pairwise disjoint sets $A_1, A_2, \ldots,
A_n$ in the complex plane, what is the probability that a random
polynomial has $m_1$ roots in $A_1$, $m_2$ roots in $A_2$, etc.?
Since the roots of our random polynomials come in two species: real
and complex conjugate pairs, we will specialize our definitions to
reflect this.  Suppose $l$ and $m$ are integers with $l + 2m \leq N$
and suppose $A_1, A_2, \ldots, A_{\ell}$ are pairwise disjoint subsets of
$\R$ and $B_1, B_2, \ldots, B_m$ are pairwise disjoint subsets of the
open upper half plane $\HH$.  If $\alpha_1, \ldots, \alpha_L$ denote
the random variables representing the real roots of a random
polynomial and $\beta_1, \ldots, \beta_M$ represent the complex
roots in the open upper half-plane (here $L$ and $M$ are random
variables too), then given $A \subseteq \R$ and $B \subseteq \HH$ we
define the random variables $N_A$ and $N_B$ by
\[
N_A = \# A \cap \{\alpha_1, \ldots, \alpha_L\} \qq{and} N_B = \# B
\cap \{\beta_1, \ldots \beta_M\}.
\]
That is, $N_A$ counts the number of roots in $A$ and takes values in
$1,2\ldots,N$.  Similarly, $N_B$ takes values in $1,2,\ldots,J=\lfloor
N/2 \rfloor$.  

If there exists a function $R_{\ell,m}: \R^{\ell} \times \HH^m \rightarrow
[0,\infty)$ so that 
\begin{equation}
\label{eq:16}
E\left[ N_{A_1} \cdots  N_{A_{l}} N_{B_1} \cdots N_{B_m} \right] =
\int\limits_{A_1} \cdots \int\limits_{A_{\ell}} \int\limits_{B_1} \cdots \int\limits_{B_m} R_{\ell,m}(\boldsymbol x, \boldsymbol z) \, \upd \mu_{\R}^{\ell}(\boldsymbol x)  \,
\upd \mu_{\C}^m(\boldsymbol z),
\end{equation}
then we call $R_{\ell,m}$ the $(\ell,m)$-{\em correlation} (or {\em
  intensity}) {\em function}.  See
\cite{borodin-2008} for a discussion of these types of two-species
correlation functions, or \cite{MR2552864} for a more in-depth
discussion of one-species correlation functions.  

Of particular note, at least for understanding the importance of
correlation functions, is the fact that $N_{\R} = L$, $N_{\HH}=M$, and 
\[
E[L] = \int_{\R} R_{1, 0}(x, -) \upd\mu_{\R}(x), \qquad E[M] =
\int_{\HH} R_{0,1}(-,z) \upd\mu_{\C}(z). 
\]

If we extend $R_{0,1}(-,z)$ to all of $\C$ by demanding
$R_{0,1}(-,\overline z) = R_{0,1}(-,z)$ (which we could likewise do for
the other correlation functions), we find 
\[
E[2M] = \int_{\C} R_{0,1}(-,z) \upd\mu_{\C}(z),
\]
and 
\begin{equation}
\label{eq:20}
N = \int_{\R} R_{1, 0}(x,-) \upd\mu_{\R}(x) + \int_{\C} R_{0, 1}(-,z)
\upd\mu_{\C}(z).
\end{equation}
When the random polynomials have i.i.d.~coefficients,
The functions $R_{1,0}$ and $R_{0,1}$ are exactly those given by Kac
\cite{MR0007812} and Shepp and Vanderbei \cite{MR1308023}.

Equation~(\ref{eq:20}) implies that $R_{1,0}/N$ gives the spatial
density of real roots of random polynomials, and $R_{0,1}/N$ gives the
spatial density of complex roots.

\begin{thm}[Borodin, Sinclair \cite{borodin-2007,borodin-2008}]
\label{thm:6}
The roots of our random polynomials form a {\em Pfaffian point
  process}.  That is, there exist a $2\times 2$ {\em
  matrix kernel} $\mathbf K_N: \C \times \C \rightarrow \C^{2 \times
  2}$, such that $R_{\ell,m}$ exists, and
\begin{equation}
\label{MatrixKernels}
R_{\ell,m}(\mathbf x, \mathbf z) = \Pf 
\begin{bmatrix}
\left[ \mathbf K_N(x_i, x_j) \right]_{i,j=1}^{\ell} & \left[ \mathbf K_N(x_i, z_n) \right]_{i,n=1}^{\ell,m} \smallskip\\
-\left[ \mathbf K_N^{\transpose}(z_k, x_j) \right]_{k,j=1}^{m,l} &
\left[ \mathbf K_N(z_k, z_n) \right]_{k,n=1}^{m}
\end{bmatrix}.
\end{equation}
\end{thm}
This kernel takes different forms depending on whether the arguments
are real or not; the exact details of this are described
below.  The importance of (\ref{MatrixKernels}) is
the fact that $\mathbf K_N$ is independent of $\ell$ and $m$; that is,
the same kernel appears in the Pfaffian formulation of all correlation
functions.  Moreover, $N$ appears as a parameter in the definition of
$\mathbf K_N$ in a way that allows for us to compute its limit as $N
\rightarrow \infty$ in various scaling regimes.  

The entries of $\mathbf K_N(u,v)$ are traditionally denoted something like
\[
\mathbf K_N(u,v) = 
\begin{bmatrix}
S_ND(u,v) & S_N(u,v) \smallskip\\
-S_N(v,u) & IS_N(u,v) + \frac12 \sgn(u - v)
\end{bmatrix}.
\]
This notation stems from the fact that, for $\beta=1$ Hermitian
ensembles (the eigenvalues of which also form Pfaffian point
processes---see \cite{MR2129906}), the analogous $(1,1)$ and
$(2,2)$-entries are given more-or-less by the {\em D}erivative (with
respect to the second variable) and the (running) {\em I}ntegral (with
respect to the first variable) of the $S_N$ term.  For the kernels we
consider here---those appearing in (\ref{MatrixKernels})---there is a
still a relationship between the various entries of $\mathbf K_N$,
though this relationship is dependent on whether the arguments are
real or complex.  We will thoroughly explain this relationship in the
sequel, but for now we remark only that when both arguments are real
the derivative/running integral relationship between $S_ND$, $IS_N$
and $S_N$ persists.  Once the relationship between the entries of
$\mathbf K_N$ is explained, it suffices to report on only one entry,
and for us it will be more convenient to describe one of $S_ND$ and
$IS_N$ instead of $S_N$.  Thus, we will use the notation for $\mathbf
K_N$ given as in the following theorem.
\begin{thm}[Borodin, Sinclair \cite{borodin-2007,borodin-2008}]
\label{thm:5}
Suppose $N$ is even. With $p_0, p_1, \ldots, p_{N-1}$ and $\mathbf U$ as in
Theorem~\ref{thm:4}, write $\mu_{m,n}$ for the $(m,n)$-entry of $\mathbf
U^{-1}$, and define
\[
\varkappa_N(u,v) = -2 \max\{1, |u|\}^{-s} \max\{1, |v|\}^{-s}
\sum_{n,m=1}^{N} \mu_{m,n} p_{m-1}(u) p_{n-1}(v)
\]
Then,
\begin{equation}
\label{eq:21}
\mathbf K_N(u,v) = 
\begin{bmatrix}\varkappa_N(u,v) & \varkappa_N \epsilon(u,v) \smallskip\\
\epsilon \varkappa_N(u,v) & \epsilon \varkappa_N \epsilon(u,v) +
\frac{1}{2} \sgn(u - v)
\end{bmatrix},
\end{equation}
where $\sgn(u - v)$ is taken to be 0 if either $u$ or $v$ is non-real and $\epsilon$ is the operator 
\begin{equation}
\label{eq:22}
\epsilon f(u) := \left\{
\begin{array}{ll}
{\displaystyle \frac{1}{2} \int_{\R} f(t) \sgn(t - u) \, \upd \mu_\R(t)} &
\mbox{if $u \in \R$,} \\ & \\
\mathrm i \sgn \big(\ip{u}\big) f(\overline u) & \mbox{if $u \in \C \setminus \R$,}
\end{array}
\right.
\end{equation}
which acts on $\varkappa_N(u,v)$ as a
function of $u$, when written on the left and acts on $\varkappa_N(u,v)$
as a function of $v$ when written on the right.
\end{thm}

Notice that if $x \in \R$ and $F$ is an antiderivative of $f$, then
\[
\epsilon f(x) = -\frac12 \int_{-\infty}^x f(t) \, \upd \mu_R(t) + \frac12
\int_x^{\infty} f(t) \, \upd \mu_R(t) = -F(x) + \frac{F(\infty) +
  F(-\infty)}{2}.
\]
It follows then that, if $D$ stands for differentiation, then $D
\epsilon f(x) = -f(x)$.  Hence, we can write the entries of $\mathbf K_N$ in
terms of $\epsilon \varkappa_N  \epsilon(x,y)$, $x,y \in \R$,
differentiation and complex conjugation.

\subsection{Scaling Limits}

Our primary interest is in the scaling limits of the various matrix
kernels as $N \rightarrow \infty$.

The scaled kernels hold
information about the limiting distribution of roots in a neighborhood
of a point on a scale where the expected number of roots in a
neighborhood is of order 1.  In order to describe the relevant kernels
we will use the heuristic assumption that, with probability tending to
1, the roots of random polynomials of large degree are nearly
uniformly distributed on the unit circle.  We will not prove this
assumption since it motivates the discussion here, but is not
logically necessary for the present discussion.

Under this assumption, if $\zeta$ is a point on the unit circle, then
when $N$ is large we expect to see $O(1)$ roots of a random polynomial
in a ball about $\zeta$ of radius $1/N$.  Supposing momentarily
that $\zeta$ is a non-real point on the unit circle, and $\varepsilon >
0$ is sufficiently small so that the ball of radius $\varepsilon$ about
$\zeta$ does not intersect the real line, then the expected number of
roots in this ball is given by 
\[
\int_{\zeta + \varepsilon \D} \Pf\left( \mathbf K_N(z, z) \right)\, \upd\mu_{\C}(z),
\]
where $\D \subset \C$ is the unit disk.  After a change of variables
this becomes
\begin{equation}
\label{scale1}
\int_\D \varepsilon^2 \Pf\left( \mathbf K_N(\zeta + \varepsilon z, \zeta +
\varepsilon z)\right) \, \upd\mu_{\C}(z).
\end{equation}
When $\varepsilon = 1/N$, using properties of Pfaffians (akin to the
multilinearity of the determinant), we have 
\begin{equation}
\label{scale2}
\int_\D \Pf\left( \frac{1}{N^2} \mathbf K_N\left(\zeta +
    \frac{z}{N}, \zeta + \frac{z}{N}\right)\right) \, \upd\mu_{\C}(z).
\end{equation}
We expect that this quantity will converge as $N \rightarrow \infty$,
and for $\zeta$ a non-real point on the unit circle, we define
\begin{equation}
\label{eq:19a}
\mathbf K_{\zeta}(z,w) := \lim_{N \rightarrow \infty} \frac{1}{N^2} \mathbf K_N\left(\zeta +
    \frac{z}{N}, \zeta + \frac{w}{N}\right).
\end{equation}
We will see that this scaled limit is essentially independent of
$\zeta$ (more specifically it depends only trivially on the argument
of $\zeta$).  

The real points on the unit circle, $\xi = \pm1$ are not generic,
since for any neighborhood, and all finite $N$, there will be a
positive (expected) proportion of real roots.  This fact is reflected
in the emergence of a new limiting kernel in scaled neighborhoods of
$\pm1$.  In this case, suppose $A_1, A_2, \ldots, A_{\ell}$ and $B_1,
B_2, \ldots, B_m$ are measurable subsets of $\R$ and $\C \setminus \R$
with positive measure, $\xi = \pm 1$, and define the shifted and
dilated sets
\[
\wt A_j := \xi + \frac{1}N A_j \qq{and} \wt B_k = \xi + \frac{1}N B_k,
\]
for $j = 1, \ldots , \ell$ and $k = 1, \ldots, m$.  By our previous
reasoning, and the fact that the point process on the roots has
Pfaffian correlations (equations~(\ref{eq:16}) and
(\ref{MatrixKernels})) we have
\begin{align}
& E\left[ N_{\wt A_1} \cdots  N_{\wt A_{l}} N_{\wt B_1} \cdots N_{\wt B_m} \right] =
\int\limits_{A_1} \cdots \int\limits_{A_{\ell}} \int\limits_{B_1}
\cdots \int\limits_{B_m} N^{-\ell - 2m}  \label{eq:18} \\
& \hspace{.5in} \times \Pf 
\begin{bmatrix}
\left[ \mathbf K_N(\xi + \frac{x_i}{N}, \xi + \frac{x_j}{N})
\right]_{i,j=1}^{\ell} & \left[ \mathbf K_N(\xi + \frac{x_i}{N}, \xi +
  \frac{z_n}{N})
\right]_{i,n=1}^{\ell,m} \\ 
-\left[ \mathbf K_N^{\transpose}(\xi + \frac{z_k}N, \xi + \frac{x_j}N)
\right]_{k,j=1}^{m,l} & \left[ \mathbf K_N(\xi + \frac{z_k}N, \xi +
  \frac{z_n}N) \right]_{k,n=1}^{m}  
\end{bmatrix}
 \,  \upd \mu_{\R}^{\ell}(\boldsymbol x) 
\upd \mu_{\C}^m(\boldsymbol z). \nonumber 
\end{align}
Note that the Jacobian of the change of variables that allows us to
integrate over the unscaled $A_j$ and $B_k$ (instead of their scaled
and shifted counterparts) introduced a factor of $N^{-1}$ for each
real variable and a factor of $N^{-2}$ for each complex variable.   

There are many ways to `move' the $N^{-\ell - 2m}$ factor inside the
Pfaffian and attach various powers of $N$ to the various matrix
entries; we wish to do this in a manner so that the resulting matrix
entries converge as $N \rightarrow \infty$.  We will be overly
pedantic here and use the fact that for any antisymmetric matrix $\mathbf
K$ and square matrix $\mathbf N$ (of the same size),
\begin{equation}
\label{eq:17}
\Pf(\mathbf N \mathbf K \mathbf{N}^{\transpose}) = \Pf \mathbf K \cdot
\det \mathbf N.
\end{equation}

Here we will use this observation with $\mathbf K$ the $2(\ell + m)
\times 2(\ell + m)$ antisymmetric matrix in the integrand of
(\ref{eq:18}), and 
\[
\mathbf N = \begin{bmatrix}
\mathbf D_{\R} \\
& \ddots \\
& & \mathbf D_{\R} \\
& & & \mathbf D_{\C} \\
& & & & \ddots \\
& & & & & \mathbf D_{\C}
\end{bmatrix}
\]
where for every real and every complex variable we introduce a $2 \times 2$
blocks of the 
form 
\[
\mathbf D_{\R} = \begin{bmatrix}
\frac{1}N & 0 \\
0 & 1 
\end{bmatrix} 
\quad \text{and} \quad
 \mathbf D_{\C} = \begin{bmatrix}
\frac{1}{N} & 0 \\
0 & \frac{1}{N}
\end{bmatrix},
\]
respectively. Clearly $\det \mathbf N = N^{-\ell -2m}$.  

From (\ref{eq:21}) we see that the entries of $\mathbf N \mathbf K \mathbf
N^{\transpose}$ are blocks of the form
\begin{align}
\label{everything}
\left\{
\begin{array}{ll}
\begin{bmatrix}
\frac{1}{N^2} \varkappa_N\left(\xi + \frac{u}{N}, \xi + \frac{v}{N}\right) & \frac1N
\varkappa_N \epsilon\left(\xi + \frac{u}{N}, \xi + \frac{v}{N}\right) \smallskip\\
\frac1N \epsilon \varkappa_N\left(\xi + \frac{u}{N}, \xi +
  \frac{v}{N}\right) &  
 \epsilon \varkappa_N \epsilon\left(\xi + \frac{u}{N}, \xi +
   \frac{v}{N}\right) + \frac{1}{2} \sgn(u - v)
\end{bmatrix}, & \quad u,v \in \R;  \\ & \\
\begin{bmatrix}
\frac{1}{N^2} \varkappa_N\left(\xi + \frac{u}{N}, \xi + \frac{v}{N}\right) & \frac1{N^2}
\varkappa_N \epsilon\left(\xi + \frac{u}{N}, \xi + \frac{v}{N}\right)
\smallskip\\ 
\frac1N \epsilon \varkappa_N\left(\xi + \frac{u}{N}, \xi + \frac{v}{N}\right) & 
\frac1N \epsilon \varkappa_N \epsilon\left(\xi + \frac{u}{N}, \xi + \frac{v}{N}\right) 
\end{bmatrix}, & \quad u \in \R, v \in \C \setminus \R. \\ & \\
\begin{bmatrix}
\frac1{N^2} \varkappa_N\left(\xi + \frac{u}{N}, \xi + \frac{v}{N}\right) & \frac1{N^2}
\varkappa_N \epsilon \left(\xi + \frac{u}{N}, \xi + \frac{v}{N}\right) \smallskip\\
\frac1{N^2} \epsilon \varkappa_N\left(\xi + \frac{u}{N}, \xi +
  \frac{v}{N}\right) &  
\frac1{N^2} \epsilon \varkappa_N \epsilon\left(\xi + \frac{u}{N}, \xi + \frac{v}{N}\right) 
\end{bmatrix}, & \quad u, v \in \C \setminus \R;
\end{array} \right.
\end{align}
This way of distributing powers of $N$ among the entries of $\mathbf
K$ will ensure that the matrix entries all converge in the limit as $N
\rightarrow \infty$.  And we define $\mathbf K_{\xi}(u,v)$ to be the
$N \rightarrow \infty$ limit of these scaled matrix kernels.  We will
find a limiting kernel $\varkappa_{\xi}$, which depends on $\xi$ in a
trivial manner, so that 
\begin{equation}
\label{limit_of_everything}
\mathbf K_{\xi}(u,v) = \begin{bmatrix}\varkappa_{\xi}(u,v) &
  \varkappa_{\xi} \epsilon(u,v) \smallskip\\ 
\epsilon \varkappa_{\xi}(u,v) & \epsilon \varkappa_{\xi} \epsilon(u,v) +
\frac12 \sgn(u - v)
\end{bmatrix}.
\end{equation}
This together with (\ref{eq:19}) defines the scaling limit of $\mathbf
K_N$ near every point on the unit circle. 

Besides the explicit identification of $\mathbf K_{\xi}$ and $\mathbf
K_{\zeta}$ we will also produce unscaled limits of $\mathbf K_N(u,v)$
when $u$ and $v$ are away from the unit circle. For $u,v$ in the open
unit disk we will find that this unscaled limit exists, and is
non-zero; a fact that implies (among other things) that the number of
roots in an open subset of the open unit disk has positive
expectation, and this expectation converges to a finite number as $N
\rightarrow \infty$.  When $u$ and $v$ are outside the closed disk,
the convergence of the unscaled limit depends on the asymptotic
behavior of $N/s$, and we will give an account of the situation
there\footnote{The dependence on the limit of $N/s$ is to be expected,
  since the larger $s$ is relative to $N$, the smaller the joint
  density of roots outside the closed unit disk.}.  These results
reflect the fact that although `most' of the roots accumulate nearby
the unit circle as $N \rightarrow \infty$ one expects that there will
always be a finite number of outliers away from the unit circle.

\subsection{Notation}

In order to expedite the presentation of the various kernels and their
scaling limits we introduce some simplifying notation.  Firstly we
will use $\T$ for the unit circle in the complex plane, $\mathbb D$ to
be the open unit disk, and $\mathbb O := \C \setminus \overline \D$.
We will continue to use $\zeta$ for a non-real point on $\T$, $\xi$
for $\pm 1$.  We will also use $u, v$ for generic complex variables,
$x, y$ for real variables and $w, z$ for non-real complex variables,
so that for instance $\mathbf K_{\xi}(x, y)$ will mean the scaling
limit of the kernel in a neighborhood of $\pm 1$ corresponding to
correlations between real roots.  Note that $\mathbf K_{\xi}(z, x) =
-\mathbf K_{\xi}(x, z)^{\transpose}$ so we need only report one of
these kernels.  

Notice that the kernels $\mathbf K_N$ depend on the parameter $s$. In what
follows we always assume that $s=s(N)>N$.  Since $s$ must
scale with $N$ in some manner, we shall always assume that 
\[
\lambda := \lim_{N\to\infty}Ns^{-1} \in [0,1]
\]
exists.  In our previous discussion, as the parameter for homogeneity
in the $M_{\lambda}$, $\lambda$ was {\em exactly} equal to $(N +
1)s^{-1}$ (or rather, $s$ was defined to be $(N+1)/\lambda$), however
for the following results we only need that $s(N) > N$ and an
asymptotic description for $\lambda$.  This generalization will also
be useful in Section~\ref{sec:conn-with-other}, where we introduce
other models with the same statistics as the roots of our random
polynomials, and in which the parameters $s$ and $\lambda$ have a
different meaning.  We include the possibility $s=\infty$, in which
case we take $\lambda = 0$ and we interpret
$\max\big\{1,|z|\big\}^{-\infty}$ as the characteristic function of
the closed unit disk.  

We remark that, since $\mathbf K_N$ is implicitly dependent on $s$,
and similarly $\mathbf K_{\zeta}$ and $K_{\xi}$ are dependent on
$\lambda$.  To simplify notation we will often leave any dependence on
$s$ and $\lambda$ implicit.   

\subsection{The Mahler Ensemble of Complex Polynomials}

Before proceeding to our results we will review the complex version of
the Mahler ensemble since it both provides context and sets us up to
demonstrate a non-trivial connection between these two ensembles.  

The complex Mahler ensemble of random polynomials is that formed by
choosing degree $N$ polynomials with {\em complex} coefficients uniformly
(with respect to Lebesgue measure on coefficient vectors) from the set
of polynomials with Mahler measure at most 1.  The complex Mahler
ensemble has many features in common with its real counterpart, for
instance we still expect the empirical measure for the roots of a
random complex polynomial to weakly converge to uniform measure on the
unit circle (in fact, in this case a large deviation result---due to the
second author---quantifying this convergence is known \cite{uY2}).
There are striking differences as well, most conspicuously the roots are
generically complex and the spatial density of roots is radial for
finite $N$.

The joint density of roots is easily computed in this situation
(in fact, it is much easier to compute than the conditional joint
densities for the real ensemble) and is proportional to
\[
\prod_{n=1}^N \max\big\{1, |z_n|\big\}^{-2s} \prod_{m < n} \log|z_n - z_m|^2.
\]

From the joint density of roots, it is straightforward to show that
the spatial process on roots forms a determinantal point process on
$\C$.  That is, there exists a function $K_N : \C \times \C \rightarrow
\C$ such that if $B_1, \ldots, B_n$ are disjoint subsets of $\C$, then
\[
E[N_{B_1} \cdots N_{B_n}] = \int_{B_1} \cdots \int_{B_n} \det\left[
  K_N(z_j, z_k) \right]_{j,k=1}^n \, \upd\mu_{\C}^n(\boldsymbol z).
\]
In other words, the $n$th correlation function can be expressed
as the determinant of an $n \times n$ matrix, the entries of which are
formed from the scalar kernel $K_N$.   

Our previous arguments suggest the relevant scaling limit is
\[
K_{\zeta}(z,w) := \lim_{N \rightarrow \infty} \frac{1}{N^2}
K_N\left(\zeta + \frac{z}{N}, \zeta + \frac{w}{N}\right).
\]
We recount the scaling limit of the kernel
here, since we will find a relationship between it and the scaling
limit for the matrix kernel(s) for the real Mahler ensemble.
\begin{thm}[Sinclair, Yattselev \cite{Sinclair2012682}]
\label{thm:1}
Let $\lambda=\lim_{N\to\infty}Ns^{-1}$. Then $K_{\zeta}(z,w)$, $\zeta \in \T$,
 exists and \footnote{In \cite{Sinclair2012682},  
  the integral in formula \eqref{complexkernel} is evaluated explicitly.
  The form given there can be easily obtained from \cite[Eq. (23) \&
  (71)]{Sinclair2012682} and elementary integration.  The form given
  here is more readily generalized, a fact which becomes useful later.} 
\begin{equation}
\label{complexkernel}
K_{\zeta}(z,w) =
\omega\big(z \overline \zeta\big) \omega\big(\overline w
  \zeta\big) \frac1\pi\int_0^1x(1-\lambda 
x)e^{\left(z\overline \zeta + \overline w \zeta\right) x}\upd x, 
\end{equation}
where 
\begin{equation}
\label{omega}
\omega(\tau) := \min\left\{1, e^{-\rp{\tau}/\lambda}\right\} = \lim_{N\to\infty}\max\left\{1,\left|1+\frac\tau N\right|\right\}^{-s}.
\end{equation}
Moreover, it holds that
\[
\lim_{N\to\infty} K_N(z,w) = \frac1\pi\frac1{(1-z\overline w)^2}
\]
locally uniformly in $\D\times\D$, and, if $s<\infty$ for each finite $N$,
\[
\lim_{N\to\infty} \frac{|z\overline w|^s}{(z\overline w)^N}\frac{K_N(z,w)}{s-N} = \frac1\pi\frac1{z\overline w-1}\left[1+\frac{c^{-1}}{z\overline w-1}\right]
\]
locally uniformly in $\Om\times\Om$, where $c:=\lim_{N\to\infty}(s-N)$.
\end{thm}
The main result of \cite{Sinclair2012682} was the universality of the
kernel $K_{\zeta}$ under conformal maps which map the exterior of the
unit disk onto the exterior of compact sets with smooth boundary.  The
details of this universality are unimportant here; instead we recount
Theorem~\ref{thm:1} since the entries for the kernel(s) for the real
Mahler ensemble depend also on~$K_{\zeta}$.

\section{Main Results}
\label{sec:main}

\subsection{The Expected Number of Real Roots}

\begin{thm}
\label{thm:expected}
Let $N_\mathsf{in}$ and $N_\mathsf{out}$ be the number of real roots on $[-1,1]$ and $\R\setminus(-1,1)$, respectively, of a random degree $N$ polynomial chosen from the real Mahler ensemble.  Then
\[
\left\{
\begin{array}{lcl}
E[N_\mathsf{in}] &=& \displaystyle \frac1\pi\log N+O_N(1) \smallskip \\
E[N_\mathsf{out}] &=& \displaystyle - \frac1\pi\frac{\sqrt{N(2s-N)}}s\log\big(1-Ns^{-1}\big) + \sqrt{Ns^{-1}}O_N(1),
\end{array}
\right.
\]
where the implicit constants depend on~$s$ but are uniform with respect to $N$.  
\end{thm}

Observe that
\[
E[N_\mathsf{out}] = \left\{ 
\begin{array}{rl}
\sqrt{Ns^{-1}}O_N(1), & \limsup_{N\to\infty}Ns^{-1}<1, \smallskip \\
\frac\alpha\pi\log N+O_N(1), & s=N+N^{1-\alpha}, \alpha\in[0,1], \smallskip \\
\frac1\pi\log N+O_N(1), & \limsup_{N\to\infty} (s-N)<\infty.
\end{array}
\right.
\]
In particular, in the third case, the leading term of the expected number of real roots of a random polynomial from the real Mahler ensemble is $\frac2\pi\log N$, which matches exactly the leading term of the expected number of real roots of a random polynomial with independent standard normal coefficients obtained by Kac \cite{MR0007812}.

\subsection{Kernel Limits Near the Unit Circle}

Our first result states that the limiting local correlations at $\zeta
\in \T \setminus \{-1,1\}$ can be given in terms of the scaled scalar kernel for
the complex Mahler ensemble.  
\begin{thm}
\label{thm:Scaling}
Let $\zeta\in\T\setminus\{\pm1\}$ and $\mathbf K_{\zeta}$ be the scaling limit of the matrix kernel \eqref{eq:21} defined by \eqref{eq:19a}. Then 
\[
\mathbf K_{\zeta}(z,w) = \left[
\begin{array}{cc}
0 & K_{\zeta}(z,w) \\
-K_{\zeta}(w,z) & 0
\end{array}
\right]
\]
with the limit in \eqref{eq:19a} holding locally uniformly for $z,w\in\C$, where $K_{\zeta}(z,w)$ is the
scaling limit of the scalar kernel for the complex Mahler ensemble
given by \eqref{complexkernel}. 
\end{thm}

Observe that
\[
\Pf\left[
\begin{array}{cc}
0 & K_{\zeta}(z_j, z_k) \\
-K_{\zeta}(z_k,z_j) & 0
\end{array}
\right]_{j,k=1}^n = \det \left[ K_{\zeta}(z_j, z_k) \right]_{j,k=1}^n,
\]
and thus we have that the limiting local distribution of roots of
real random polynomials at a point in $\T \setminus \{-1,1\}$ collapses to a
determinantal point process identical to 
that for the complex Mahler ensemble at the same point. 

A new kernel arises in a neighborhood of $\xi = \pm 1$.  As in
(\ref{omega}), $\omega(\tau) = \min\{1, e^{-\rp{\tau}/\lambda}\}$,
interpreting this as the characteristic function on the closed unit
disk in the case $\lambda = 0$.  We also define 
\[
\displaystyle M(z):=\frac1{\Gamma(3/2)}\sum_{n=0}^\infty\frac{\G{n+ 3/2}}{\G{n+1}}\frac{z^n}{n!} = {}_1F_1(3/2, 1; z).
\]
We remark that $M(z)$ can be expressed
rather succinctly in terms of modified Bessel functions, though we
have no reason to do that here.  

\begin{thm}
\label{thm:7}
For $\xi=\pm1$, let $\varkappa_{\xi}(u,v)$ be defined by \eqref{everything} \& \eqref{limit_of_everything}. Then
\begin{equation}
\label{kappa}
\varkappa_{\xi}(u,v) = 
\omega\left(u\xi\right)\omega\left(v\xi\right)\frac\xi4\int_0^1\tau(1-\lambda \tau)\big[M^\prime(u\xi \tau)M(v\xi \tau)-M(u\xi \tau)M^\prime(v\xi \tau)\big]\upd \tau,
\end{equation}
where $\omega(\tau)$ is defined by \eqref{omega} and the convergence in \eqref{everything}
is uniform on compact subsets of $\C \times \C$.  
\end{thm}

For the sake of brevity, we shall use the following notation:
\begin{equation}
\label{iota}
\iota(z) := \mathrm{i}\sgn\big(\ip{z}\big).
\end{equation}

Because $\epsilon$ operator \eqref{eq:22} amounts to conjugation and multiplication by $\iota$ for complex arguments and since $\overline{\xi+\frac uN}=\xi+\frac{\overline u}N$, it is clear that \eqref{limit_of_everything} is indeed the limit of \eqref{everything} when $u,v\in\C\setminus\R$ and the following corollary takes place.

\begin{cor}
\label{cor:Scaling1}
If $z,w \in \C \setminus \R$, then
\begin{equation}
\label{1complex}
\mathbf K_{\xi}(z,w) = \left[ 
\begin{array}{cc}
\varkappa_{\xi}\big(z,w\big) & \iota(w)\varkappa_\xi\big(z,\overline w\big) \smallskip \\ 
-\iota(z)\varkappa_\xi\big(w,\overline z\big) & \iota(z)\iota(w)
\varkappa_\xi\big(\overline z,\overline w\big) 
\end{array}
\right].
\end{equation}
\end{cor}

As already been suggested in \eqref{limit_of_everything} following the discussion after Theorem~\ref{thm:5}, the remaining kernels are most conveniently reported in terms of the
$(2,2)$ entry of $\mathbf K_{\xi}(x,y)$, and thus we introduce the following
notation.  
\begin{equation}
\label{three-kernels}
\left\{
\begin{array}{lcl}
\displaystyle \mathbf K[A](x,y) &=& 
\displaystyle \left[ 
\begin{array}{cc}
DAD(x,y) & -DA(x,y) \smallskip \\ 
AD(x,y) & A\big(x,y\big) + \frac12\sgn(y-x)
\end{array}
\right], \bigskip \\
\displaystyle \mathbf K[A](z,y) &=& 
\displaystyle \left[ 
\begin{array}{cc}
DAD(z,y) & -DA(z,y) \smallskip \\ 
\iota(z)DAD\big(\overline z,y\big) &
-\iota(z) DA\big(\overline z,y\big)  
\end{array}
\right], \bigskip \\
\displaystyle \mathbf K[A](y,z) &=& -\mathbf K[A]^{\transpose}(z,y), \bigskip \\
\displaystyle \mathbf K[A](z,w) &=& \displaystyle \left[ 
\begin{array}{cc}
DAD\big(z,w\big) & \iota(w)DAD\big(z,\overline w\big) \smallskip \\ 
\iota(z) DAD\big(\overline z, w\big) & \iota(z)\iota(w)
DAD\big(\overline z,\overline w\big)  
\end{array}
\right],
\end{array}
\right.
\end{equation}
where $D$ is differentiation with respect to the first (second) variable when
written on the left (right).  

\begin{thm}
\label{thm:Scaling2}
For $\xi=\pm1$, define
\begin{equation}
\label{auxkappa}
A_\xi(a,b):= \int_0^a\int_0^b\varkappa_\xi(u,v)\upd u \upd v +
\left(\int_0^b-\int_0^a\right)\left[\frac{\omega(v\xi)}4\int_0^1(1-\lambda
  u)M(v\xi u)\upd u\right]\upd v.
\end{equation}
Then $\mathbf K_{\xi}(u,v) = \mathbf K[A_{\xi}](u,v)$
where the convergence in \eqref{everything} is locally uniform on $\C \times \C$.  
\end{thm}

Expression \eqref{auxkappa} can be simplified when $u\xi,v\xi<0$ as in
this case $\omega\big(u\xi\big)=\omega\big(v\xi\big)=1$. Recall that
$M$ is a confluent hypergeometric function and therefore is a solution
of the second order differential equation, namely,
$zM^{\prime\prime}(z)+(1-z)M^\prime(z)-\frac32M(z)=0$. Then, if we
define $I(z):=2z\big(M^\prime-M\big)(z)$ (and therefore 
$I^\prime(z)=M(z)$), we can write 
\[
A_\xi(a,b) =
\frac{\xi}4\int_0^1\frac{1-\lambda\tau}{\tau}\bigg(I'(a\xi\tau)I(b\xi\tau)
- I(a\xi\tau)I'(b\xi\tau)\bigg) \upd \tau, 
\]
which bears a striking structural resemblance to $\varkappa_{\xi}$.

Let us compare one consequence of Theorem~\ref{thm:7} (via
Corollary~\ref{cor:Scaling1}) with the analogous situation for complex
random polynomials.  This theorem implies that we can compute the
large $N$ limit of the expected number of roots in a set of the form
$\xi + \frac{1}{N} B$, disjoint from the real axis, by integrating
\[
\Pf~ \mathbf K_\xi(z,z) = \iota(z)\varkappa_{\xi}\big(z, \overline z\big)
\]
over $B$.  (This function is the scaled {\em intensity} of complex roots near $\xi$).  
\begin{figure}[!ht]
\centering
\includegraphics[trim = 0pt 2pt 0pt 0pt, clip,scale=.48]{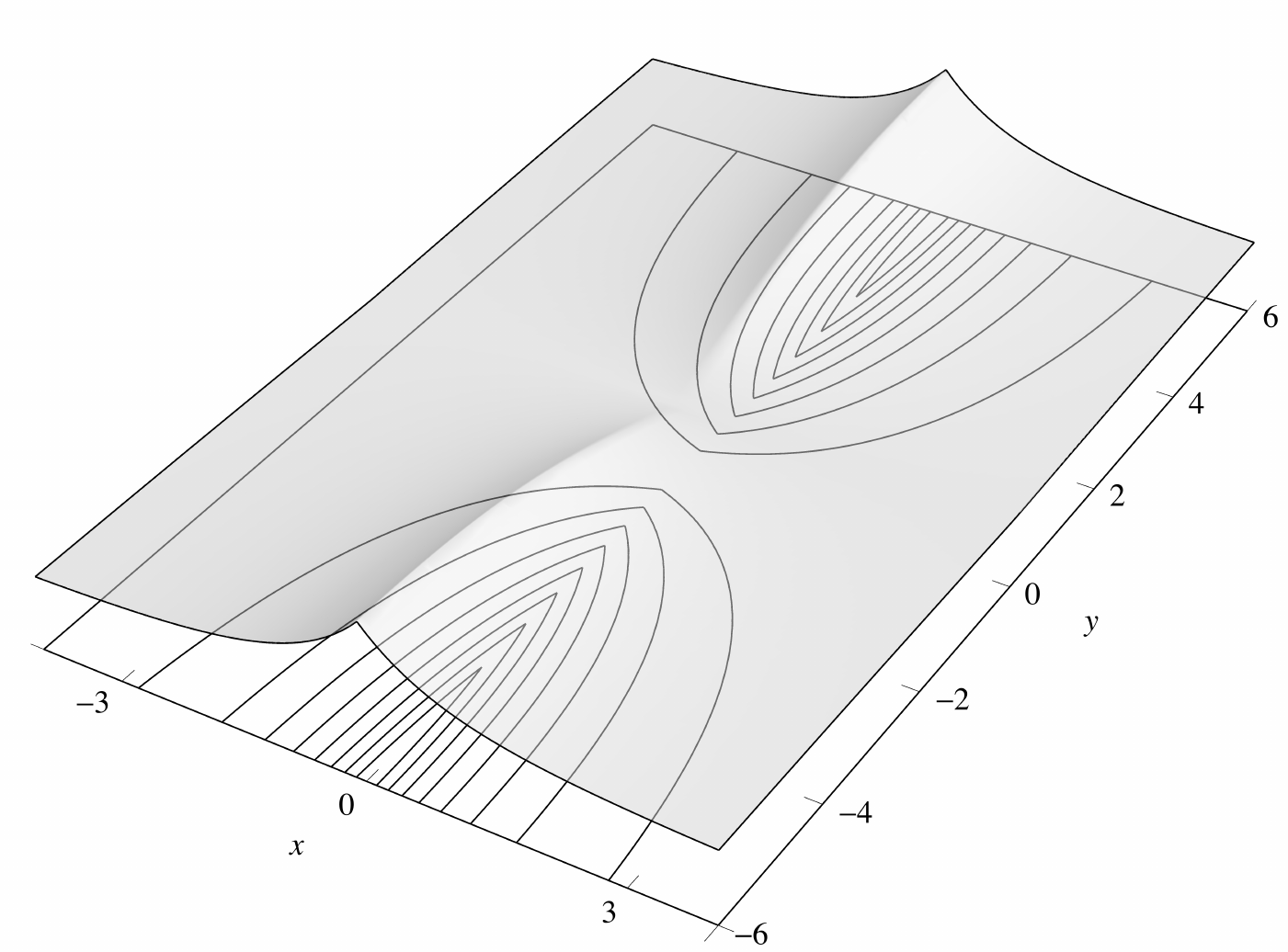}
\includegraphics[trim = 0pt 0pt 0pt 2pt, clip,scale=.4]{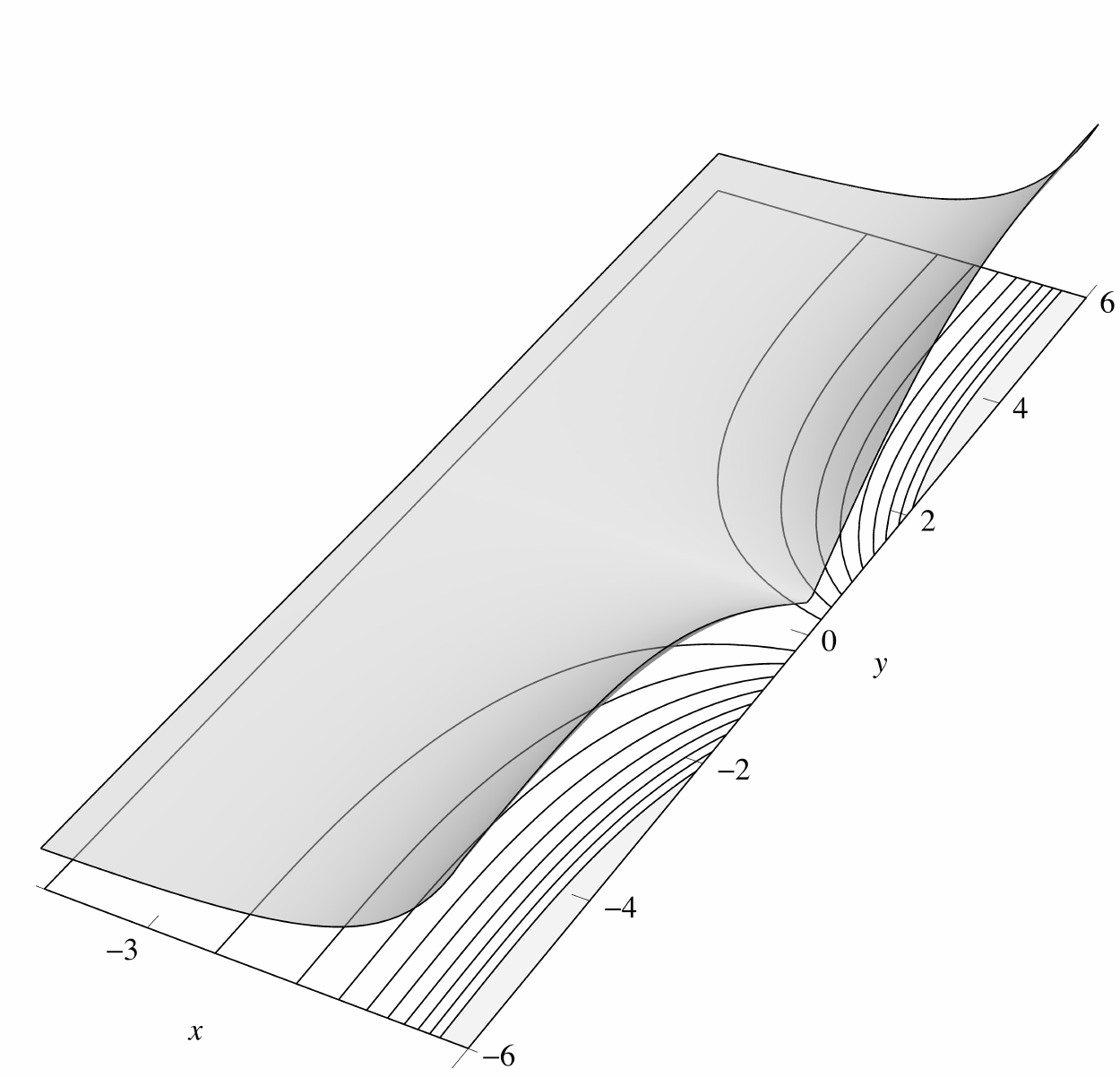}
\caption{The scaled intensity of complex roots
  near 1, for $\lambda = 1$ (left) and $\lambda = 0$ (right). Note how the roots tend
  to accumulate near the unit disk (the $y$-axis here) and repel from
  the real axis.}
\label{fig:1}
\end{figure}

This formula is not valid when $z = \overline z$, since in that case,
the Pfaffian of $\mathbf K[A_\xi](x,x)$ is responsible for the expected
number of (real) roots.  However, as $z \rightarrow \overline z~(=w)$,
we see that the integrand in (\ref{kappa}), and hence the intensity of
complex roots goes to zero.  This is not surprising, since the
conditional joint densities of roots, (\ref{eq:6}), vanishes there.
Loosely speaking, roots repel each other, including those that are
complex conjugates of each other, and this causes an expected scarcity
of complex roots near the real axis.  The fact that this phenomenon is
visible on the scale $1/N$ is worth noting if not particularly
surprising.

Since the complex ensemble has no conjugate symmetry, we should not
expect the corresponding integrand to vanish.  In this case, according
to Theorem~\ref{thm:1}
%, and using the fact that $H_1(x) = e^x$, 
the local intensity of roots near $\xi$ is governed by
\[
\omega\big(z\xi\big) \omega\big(\overline z\xi\big) \frac{1}{\pi} \int_0^1 \tau(1 - \lambda \tau)
e^{2\rp{z}\xi \tau} \upd \tau.
\]
which is remarkably similar to $\iota(w) \varkappa_{\xi}(z,\overline
z)$ aside from the obvious considerations due to root symmetry. 

One expects that, as $|\ip{z}| \rightarrow \infty$ the local intensity
of complex roots in the real ensemble will approach the local intensity
of roots in the complex ensemble, since in this limit the repulsion
from complex conjugates vanishes.  This is indeed the case.

\begin{lemma}
\label{lem:compare}
It holds that
\[
\lim_{|\ip{z}|\to\infty} \frac{\iota(z)}4\bigg[M^\prime\big(z\big)M\big(\overline z\big)-M\big(z\big)M^\prime\big(\overline z\big)\bigg] = \frac{e^{2\rp{z}}}\pi
\]
\end{lemma}

\subsection{Kernel Limits Away from the Unit Circle}

In this section we discuss the asymptotics of the matrix kernels
$\mathbf K_N$ away from the unit circle. In this case, the scale of
the neighborhood about a point $\omega \in\C\setminus\T$  will not
depend on $N$.  As we show below, the kernels $\mathbf K_N$ converge uniformly in a neighborhood of any such point, and
hence it suffices to investigate the asymptotics of $\mathbf K_N(u,
v)$.  It will turn out that the asymptotics are different depending on
whether or not $z, w$ are inside or outside the disk.

We start by considering the case of the unit disk.

\begin{thm}
\label{thm:Inside}
For $u,v\in\D$, set
\begin{equation}
\label{ISNInside}
A_\D\big(u,v\big) :=
\frac1{4\pi}\int_\T\frac{\left(v\sqrt{-\tau}-u\sqrt{-\overline\tau}\right)|\upd\tau|}{\big(1-u^2\overline\tau\big)^{1/2}\big(1-v^2\tau\big)^{1/2}},  
\end{equation}
where $\sqrt{-\tau}$ is the branch defined by
$-\frac2\pi\sum_{-\infty}^\infty \frac{\tau^m}{2m-1}$,
$\tau\in\T$. Then $\displaystyle\lim_{N\to\infty} \mathbf K_N(u,v)=\mathbf K[A_\D](u,v)$ locally uniformly on $\D \times \D$.
\end{thm}

Among the implications of this theorem is that the unscaled intensity
of roots converges on compact subsets of $\D$.  In particular, this
implies that if $B \subset \D$ has non-empty interior, then the expected
number of roots in $B$ converges to a positive quantity (obviously
dependent on $B$).  In particular, for any $\varepsilon > 0$ the
expected number of roots in the disk $\{z : |z| < 1 - \varepsilon \}$ converges to a positive
number.  In words, even though we expect the roots to accumulate
uniformly on the unit circle, we also should expect to find a positive
number of roots away from $\T$. 

In order to investigate the situation when $z, w$ are outside the
closed unit disk, we first record the following theorem. In what follows we always assume that $s<\infty$ for each finite $N$ as otherwise $\mathbf K_N$ is identically zero in $\Om\times\Om$. 

\begin{thm}
\label{thm:DSN}
Set $c:=\lim_{N\to\infty}\big( s - N \big)\in[1,\infty]$ and let $\varkappa_N(u,v)$ be the $(1,1)$-entry of $\mathbf K_N$ as in \eqref{eq:21}. Then, for
$u, v \in \Om$, 
\begin{equation}
\label{eq:19}
\lim_{N \rightarrow \infty} \frac{|uv|^s}{(uv)^N} \frac{\varkappa_N(u,v)}{s-N} = \frac\lambda\pi\left[1+\frac{c^{-1}}{uv-1}\right] \frac1{uv-1} \frac{v - u}{\sqrt{u^2 - 1}\sqrt{v^2 - 1}}
\end{equation}
locally uniformly in $\Om\times\Om$, where $c^{-1}=0$ when $c=\infty$.
\end{thm}

Note the factor of $|uv|^s/(uv)^N$.  When $\lambda<1$, this factor
diverges at least geometrically fast which yields that  
\begin{equation}
\label{DSNtoZero}
\lim_{N \rightarrow \infty} \varkappa_N(u,v) = 0
\end{equation}
locally uniformly in $\Om\times\Om$.
Furthermore, if $\lambda=1$ but $s-N\to\infty$, then
$|uv|^{s-N}/(s-N)$ diverges and the conclusion \eqref{DSNtoZero}
holds again. Only in the case where $c<\infty$, do we get the non-trivial
limit
\begin{equation}
\label{C}
\lim_{N \rightarrow \infty}\frac{|uv|^N}{(uv)^N} \varkappa_N(u,v) = \frac1{\pi}\left[c+\frac1{uv-1}\right] \frac1{|uv|^c}\frac{1}{uv-1} \frac{v - u}{\sqrt{u^2 - 1}\sqrt{v^2 - 1}} =: B(u,v).
\end{equation}

It remains to explain the seemingly superfluous term
$\big(|uv|/(uv)\big)^N$.  To do this, let $u_1, u_2, \ldots, u_M$ be
points outside the unit disk (some could be real, some complex).  
A general correlation function of $M$ roots, is given as the
Pfaffian of a $2M \times 2M$  matrix of the form 
\[
\mathbf A = \left [
\begin{array}{cc}
\mathbf K_N\big(u_i, u_k\big)
\end{array}
\right]_{i,k=1}^M,
\]
where the exact structure of  $\mathbf K_N\big(u_i, u_k\big)$ will
depend on whether $u_i$ and/or $u_k$ are real or not.  (This is
essentially the content of (\ref{MatrixKernels})). It readily follows from \eqref{eq:21} \& \eqref{eq:22} that for $u,v\in\C\setminus\R$ each entry of $\mathbf K_N$ possesses a limit identical up to conjugation to \eqref{eq:19}. Thus, we define the block diagonal matrix 
\[
\mathbf D := \begin{bmatrix}
\delta_{i,k} \left(|u_i|~/~ u_i\right)^N & 0\\
0 & \delta_{i,k} \left(|u_i|~/~\overline u_i\right)^N 
\end{bmatrix}_{i,k=1}^M.
\]
As $\mathbf D$ is diagonal, we have that $\Pf(\mathbf D \mathbf A \mathbf D^{\ast}) = \det
\mathbf D \cdot \Pf \mathbf A = \Pf \mathbf A$, and that
\[
\mathbf D \mathbf A \mathbf D^{\ast} = \begin{bmatrix}
   \left( \frac{|u_i|}{u_i} \frac{|u_k|}{u_k} \right)^N \mathbf
   \varkappa_N \big(u_i, u_k\big) & \left( 
    \frac{|u_i|}{u_i} \frac{|u_k|}{\overline u_k} \right)^N
  \varkappa_N \epsilon \big(u_i,u_k\big) \smallskip\\ 
  \left( \frac{|u_i|}{\overline u_i}  \frac{|u_k|}{u_k} \right)^N
  \mathbf \epsilon \varkappa_N \big(u_i, u_k\big) & \left( 
  \frac{|u_i|}{\overline u_i} \frac{|u_k|}{\overline u_k} \right)^N
\left\{ \epsilon \varkappa_N \epsilon \big(u_i, u_k\big) + 
\frac12 \sgn(u_i - u_k) \right\}
\end{bmatrix}_{i,k=1}^m.
\]
Since our primary interest is in $\Pf(\mathbf A)$ and not $\mathbf A$
itself and in light of Theorem~\ref{thm:DSN} we shall study 
\begin{equation}
\label{K_N-hat}
\wh{\mathbf{K}}_N(u,v) := \left [
\begin{array}{cc}
\left( \frac{|uv|}{uv}\right)^N \varkappa_N\big(u, v\big) & \left(
  \frac{|uv|}{u\overline v}\right)^N \varkappa_N \epsilon \big(u,
v\big) \smallskip \\  \left( \frac{|uv|}{\overline uv}\right)^N
\epsilon \varkappa_N^{(2,1)}\big(u,v\big) & \left( 
  \frac{|uv|}{\overline u\overline v}\right)^N
\left\{ \epsilon \varkappa_N \epsilon \big(u, v\big) + 
\frac12 \sgn(u - v) \right\} 
\end{array}
\right]
\end{equation}
rather than $\mathbf K_N$.

\begin{thm}
\label{thm:Outside}
Let $c = \lim_{N \rightarrow \infty} (s - N) \in [1,\infty]$. Set
$A_\Om \equiv 0$ when $c = \infty$, and otherwise define 
\begin{equation}
\label{A_Om}
\begin{array}{l}
\displaystyle A_\Om(x,y) := \int_{\sgn(x)\infty}^x\int_{\sgn(y)\infty}^y B(u,v)\,
\upd u \, \upd v \\
\displaystyle \hspace{3cm} +  \frac1{\sqrt\pi}\frac{\Gamma\left(\frac
    {c+1}2\right)}{\Gamma\left(\frac
    c2\right)}\left(\sgn(x)
  \int_{\sgn(y)\infty}^y-\sgn(y)\int_{\sgn(x)\infty}^x\right)\frac{\upd 
  u}{|u|^c\sqrt{u^2-1}},
\end{array}
\end{equation}
where $B(u,v)$ is defined in \eqref{C}. Then $\displaystyle \lim_{N\to\infty} \wh{\mathbf K}_N(u,v) = \mathbf K[A_\Om](u,v)$ locally uniformly in $\Om \times \Om$. 
\end{thm}

\begin{rem}
The root $\sqrt{u^2-1}$ occurring within the integrals in \eqref{A_Om} should be understood as the trace on $\R\setminus(-1,1)$ of $\sqrt{z^2-1}$ that is holomorphic in $\C\setminus[-1,1]$. In particular, it is negative for negative $x$.
\end{rem}

\begin{rem}
Even though the function $A_\Om$ is defined for real arguments only,
it is a simple algebraic computation to see that $D A_\Om$ is well
defined when the first argument is complex 
and $D A_\Om D$ is nothing else but $B(u,v)$ in \eqref{C}. 
\end{rem}

The intensity of complex roots outside the unit circle is given by
$\iota(z) B(z,\overline z)$.  Integrating this over a set $B \subset
\mathbb O$ yields the expected number of complex roots of random
degree $N$ polynomial in $B$.  When $\lambda < 1$ (or more generally
when $c = \infty$), we see from \eqref{DSNtoZero} that the limiting expectation goes to 0.
In particular, in this situation, the expected number of roots outside
the unit disk goes to 0 with $N$.  When $c$ is finite and $B$ is
bounded away from $\T$ with positive Lebesgue measure, then the
expected number of roots in $B$ will converge to a positive number
dependent on $B$ and $c$.  

\begin{figure}[!ht]
\centering
\includegraphics[scale=.5]{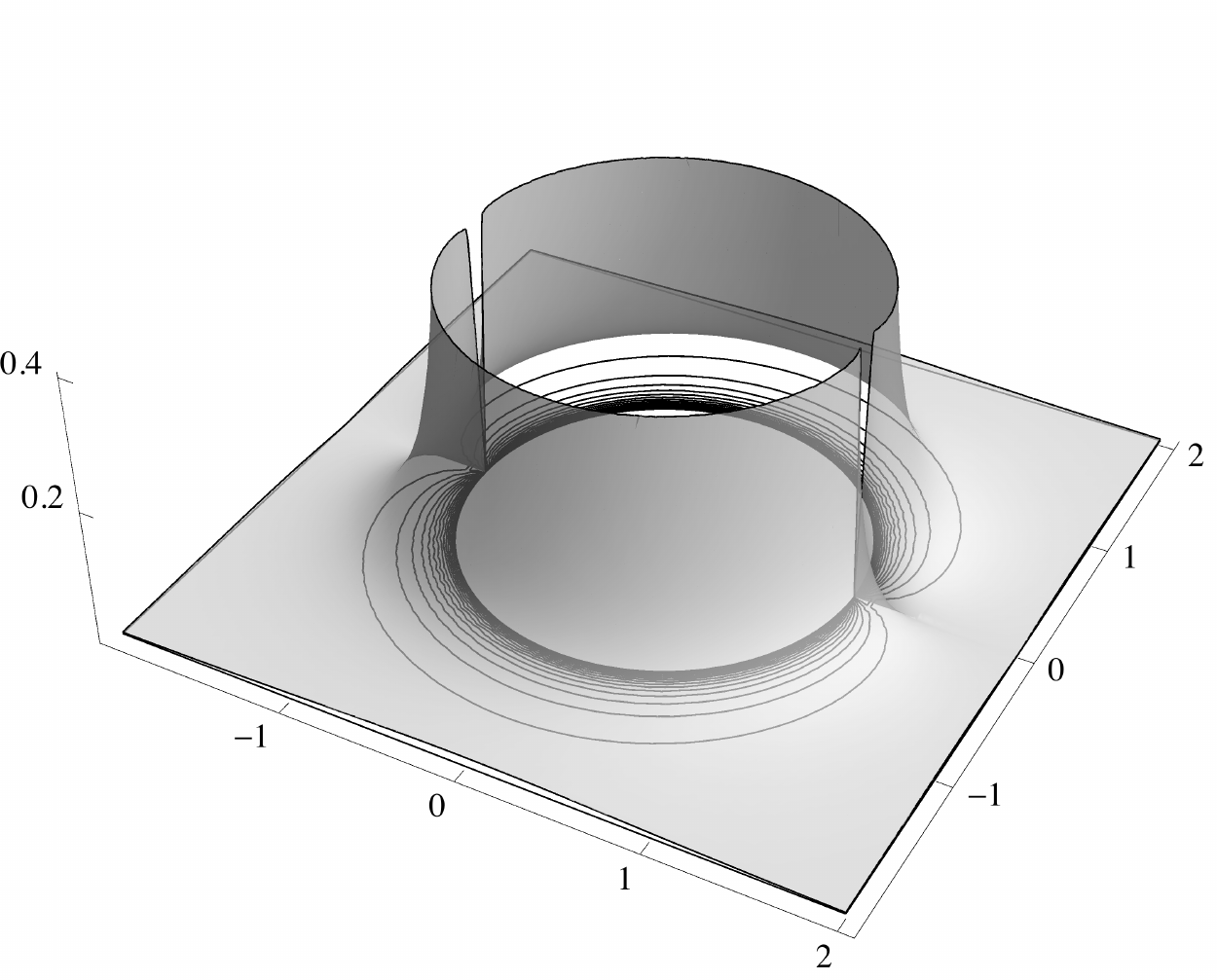}
\includegraphics[scale=.5]{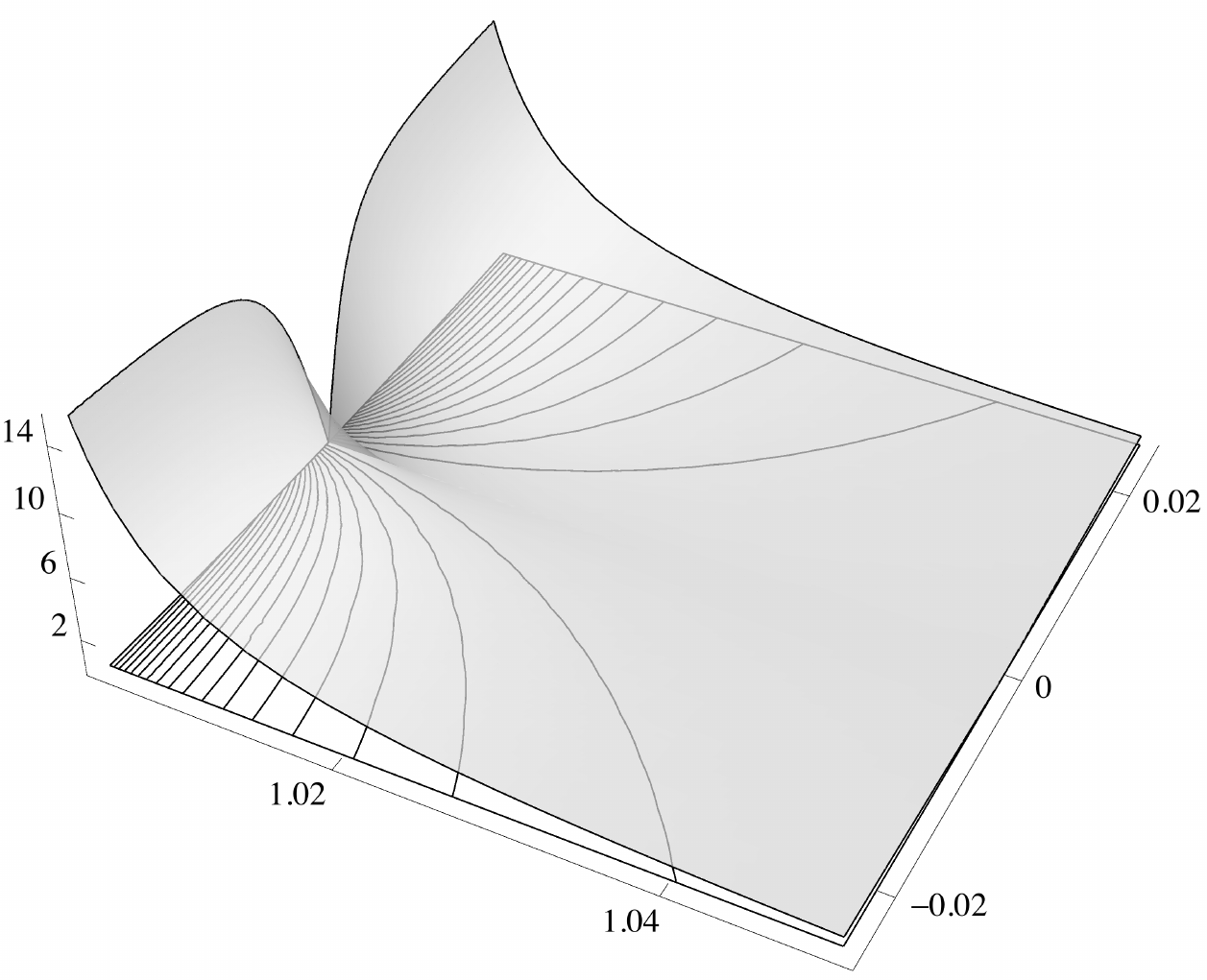}
\caption{The limiting intensity of complex roots outside the
  disk, with a close up view near $z=1$, for the Mahler measure
  ($c=1$) case.  }
\label{fig:4}
\end{figure}

\section{Other Ensembles with the Same Statistics}
\label{sec:conn-with-other}

Before proving our main results we also present an electrostatic
model and a matrix model which produce the same point process.  

\subsection{An Electrostatic Model}

In two dimensional electrostatics, we identify charged particles with
points in the complex plane.  An electrostatic system with unit
charges located at $z, z'$ has potential energy $-\log|z - z'|$ and a
system with $N$ unit charges located at the coordinates of $\boldsymbol z
\in \C^N$ has potential energy
\[
-\sum_{m < n} \log|z_n - z_m|.
\]
The states which minimize this energy correspond to those where the
particles are infinitely far apart ({\it i.e.} the particles repel
each other) and in order to counteract this repulsion we introduce a
confining potential.  There are many possibilities for this confining
potential, but in order to arrive at a model with particle statistics
identical to those of the roots of our random polynomials, the
potential we introduce is that formed from an oppositely charged
region, identified with the unit disk, with total charge $s$ and
charge density representing its equilibrium state.  More precisely,
the equilibrium charge density is given by its equilibrium measure
(in the sense of potential theory in the complex plane
\cite{MR1334766}) and by symmetry this is simply normalized Lebesgue
measure on the unit circle.  That is, the interaction energy between
the charged unit disk at equilibrium (with total charge $s$) and a unit charge
particle at $z \in \C$ is given by
\[
s \int_0^1 \log|z - e^{2 \pi i \theta}| \, \upd\theta = s \log\max\big\{1, |z|\big\},
\]
where equality is a consequence of Jensen's formula.  It follows that the
total energy of the system of particles located at the coordinates of
$\boldsymbol z$ in the presence of the equilibrized charged disk is given by
\[
E(\boldsymbol z) = s \sum_{n=1}^N \log\max\big\{1, |z_n|\big\}  -\sum_{m < n} \log|z_n - z_m|.
\]
When the system is at temperature specified by the dimensionless
inverse temperature parameter $\beta = (k T)^{-1}$ the 
probability density of finding the system in state $\boldsymbol z$
is given by
\[
\frac{e^{-\beta E(\boldsymbol z)}}{Z} = \frac{1}{Z} \bigg\{ \prod_{n=1}^N \max\{1,
|z_n|\}^{-\beta s} \bigg\} \prod_{m<n} \log|z_n - z_m|^{\beta}
\]
where $Z$ is the {\em partition function} of the system,
\begin{equation}
\label{eq:7}
Z = \int_{\C^N} e^{-\beta E(\boldsymbol z)} \, \mathrm{d}\mu_{\C}^N(\boldsymbol z).
\end{equation}

If we specify that $\beta=1$, exactly $L$ coordinates of $\boldsymbol z$
are real and the remaining $2M = N - L$ come in complex conjugate
pairs (which we can think of as mirrored particles) then the
probability density of states is exactly the conditional density of
roots of random polynomials with exactly $L$ real roots $P_{L,M}$, and the
partition function is the conditional partition function $Z_{L,M}$.  

If we specify that the total charge of all particles is $N$, but allow
the number of real and complex conjugate pairs of particles to vary, then
we arrive at a zero-current (i.e., conserved charge) grand canonical
ensemble, whose conditional density for the population vector $(L,M)$
is given by $X^L P_{L,M}$ where $X$ is the {\em fugacity}, a quantity
that encodes how easily the system can change population vectors.  The
partition function, as a function of the fugacity and the charge on
the unit disk, is given by 
\[
Z(X; s) = \sum_{(L,M) \atop L + 2M = N} \frac{X^L Z_{L,M}(s)}{2^M L! M!},
\]
which, when $X = 1$ is, up to the factor $2/(N+1)$, is the volume of
the Mahler measure star body, and as a function of $X$ is the
probability generating function for the probability that the
electrostatic configuration has exactly $L$ real particles, or
equivalently that a random polynomial has exactly $L$ real roots.

\subsection{A Normal Matrix Model}

Given a self-map on a metric space, the entropy is a measure of orbit
complexity under repeated iteration of this map.  Loosely speaking
this quantity measures how far neighboring points can get away from
each other under iteration by this map.  We will not give a definition
of this quantity, since the formulation is complicated and not really
necessary here; see \cite{MR648108} for a discussion.  When the metric
space is $\R^N$ (or $\C^N$) and the self-map is an $N \times N$
matrix, then a theorem of Yuzvinsky has that the entropy is the
logarithm of the Mahler measure of the characteristic polynomial
\cite{Yuzvinski:1965lr}.  That is, if $\mathbf M$ is an $N \times N$
matrix with eigenvalues $\lambda_1, \lambda_2, \ldots, \lambda_N$ in
$\C$, then the entropy of $\mathbf M$ is given by
\[
h(\mathbf M) = \sum_{n=1}^N \log \max\big\{1, |\lambda_n|\big\}.
\]
Despite not giving a definition of entropy, this is a sensible result
since it is clear that the `big' eigenvalues are responsible for
nearby points moving away from each other under repeated iteration of
$\mathbf M$.  

If we wish to use the entropy to form a probability measure on some
set of matrices (equipped with some natural reference measure), an
obvious choice is to have the Radon-Nikodym derivative with respect to
the reference measure be given by $e^{-s h(\mathbf M)}$ where $s$ is
some sufficiently large constant necessary so that the resulting
measure is actually finite.  Under such a probability measure we would
be more likely to choose a matrix with small entropy, large entropies
being exponentially rare.  

A natural choice for the set of matrices is that of normal matrices.
That is matrices which commute with their conjugate transpose.  (One
reason this is a natural choice is that normal matrices are unitarily
equivalent to diagonal matrices, and the entries in the diagonal
matrix are independent of the random variables parametrizing the
unitary group).  When restricting to real normal matrices, the
eigenvalues come in real and complex conjugate pairs, and the joint
density of eigenvalues naturally decomposes into conditional densities
dependent on the number of real eigenvalues.  These conditional
densities are identical to those given for the real Mahler ensemble of
polynomials given in equation (\ref{eq:6}). The derivation of the
joint density of eigenvalues uses standard random matrix theory
techniques, augmented to deal with the two species of eigenvalues, the
details of which are given in Appendix~\ref{sec:rand-norm-matr}.

We conclude that the eigenvalue statistics for this {\em entropic}
ensemble of random normal matrices is identical to the particle
statistics in the electrostatic model and the root statistics of
polynomials chosen randomly from Mahler measure star bodies.

\section{Matrix Kernels and Skew-Orthogonal Polynomials}
\label{sec:matrixkernel}

The matrix kernel $\mathbf K_N$ can be most simply represented via
weighted sums of the \emph{skew-orthonormal polynomials} for the
skew-symmetric bilinear form  
\begin{eqnarray}
\label{BilinearForm}
\la f | h \ra &=& \la f | h \ra_{\R} + \la f | h \ra_{\C} \\
&=& \int\left(\widetilde f(u)\epsilon \widetilde h(u)-\epsilon \widetilde f(u)\widetilde h(u)\right)\upd \big(\mu_\C+\mu_\R\big)(u) \nonumber
\end{eqnarray}
where $\la \cdot | \cdot \ra_{\R}$ and $\la \cdot | \cdot \ra_{\C}$
are given as in Theorem~\ref{thm:4}, the operator $\epsilon$ is given by \eqref{eq:22}, $\widetilde f(u):=f(u)\max\big\{1,|u|\big\}^{-s}$, and the functions $f,h$ satisfy the symmetry $g(\overline u)=\overline{g(u)}$. Namely, let
$\{\pi_n\}$, $\deg \pi_n = n$, be a sequence of polynomials such that 
\[
\la \pi_{2n} | \pi_{2m} \ra = \la \pi_{2n+1}, \pi_{2m+1} \ra = 0
\qq{and} 
\la \pi_{2n} | \pi_{2m+1} \ra = -\la \pi_{2m+1} | \pi_{2n} \ra = \delta_{m,n}.
\]
Note that this sequence is not uniquely determined since we may
replace $\pi_{2n+1}$ with $\pi_{2n+1} + c \pi_{2n}$ without disturbing
the skew-orthogonality.

\begin{thm}
\label{thm:2}
For each fixed $s$, one possible family of skew-orthonormal polynomials corresponding to bilinear form \eqref{BilinearForm} is given by\footnote{When $s=\infty$, it is understood that $(s-(2k+2))/s=1$.}
\[
\left\{
\begin{array}{lll}
\pi_{2n}(z) &=& \displaystyle\frac{2}{\pi} \sum_{k=0}^n \frac{\G{k +3/2} \G{n - k + 1/2}}{\G{k+1} \G{n-k+1} }z^{2k}, \smallskip \\
\pi_{2n+1}(z) &=& \displaystyle-\frac{1}{2 \pi} \sum_{k=0}^n \frac{s -(2k+2)}{2s} \frac{ \Gamma(k + 3/2) \Gamma(n - k - 1/2)}{\G{k+1}\G{n -k + 1} } z^{2k+1}.
\end{array}
\right.
\]
\end{thm}

These polynomials were originally produced using the skew analog of the
Gram-Schmidt procedure from the previously computed {\em skew}-moments, see \cite[Lemma 4.1]{sinclair-2005}.
\begin{lemma} 
\label{lemma:1}
$\la z^{2n} | z^{2m} \ra = \la z^{2n-1} | z^{2m-1} \ra = 0$ and, 
\begin{equation} 
\label{eq:8}
\la z^{2n} | z^{2m+1} \ra = \left(\frac{s}{s-2m-2}\right)
\frac{1}{\left(n+\frac12\right)\left(m-n+\frac12\right)}.
\end{equation}
\end{lemma}
Theorem~\ref{thm:2} follows from Lemma~\ref{lemma:1} and the following
theorem.
\begin{thm}
\label{thm:3}
Suppose $\{C_m\}$ is a sequence of non-zero real numbers, and $\alpha,
\beta \in \R \setminus \{-1, -2, \ldots\}$ and suppose $\la \cdot |
\cdot \ra^{\alpha, \beta}$ is a skew-symmetric inner product with $\la
z^{2n} | z^{2m} \ra^{\alpha, \beta} = \la z^{2n+1} |
z^{2m+1} \ra^{\alpha, \beta} = 0$, and 
\begin{equation}
\label{eq:9}
\la z^{2n} | z^{2m+1} \ra^{\alpha, \beta} = C_m  \bigg\{\prod_{j = 1}^n \frac{j
  - 1 - \beta}{j + \alpha} \bigg\} \frac{1}{m - n + 1 + \beta}.
\end{equation}
Define
\[
\pi_{2n}^{\alpha,
  \beta}(z) = \frac{1}{\G{1+\alpha} \G{1+\beta}}
\sum_{\ell=0}^n \frac{\G{\ell + \alpha + 1} \G{n - \ell +  \beta + 1}}{\G{\ell+1}
  \G{n - \ell + 1}}  z^{2\ell},  
\]
and
\[
\pi_{2n+1}^{\beta}(z) = \frac{1}{\G{-\beta - 1} \G{1+\beta}}
\sum_{\ell=0}^n \frac{\G{\ell + \beta  
    + 2} \G{n - \ell -
    \beta - 1} }{\G{\ell+1} \G{n - \ell + 1}} \frac{z^{2\ell+1}}{C_\ell}.
\]
Then, $\big\{ \pi_0^{\alpha, \beta}, \pi_1^{\beta}, \pi_2^{\alpha,
  \beta}, \pi_3^{\beta}, \ldots \big\}$ is a family of skew-orthonormal
polynomials for the skew-symmetric inner product $\la \cdot | \cdot
\ra^{\alpha, \beta}$. 
\end{thm}

It follows immediately from \eqref{eq:8} and \eqref{eq:9} that $\pi_{2n} = \pi_{2n}^{1/2,-1/2}$ and $\pi_{2n+1} = \pi_{2n+1}^{-1/2}$, where $C_m=2s/(s-2m-2)$. Thus, with a little bit of algebra, we get that
\begin{equation}
\label{pies}
\left\{
\begin{array}{lll}
\pi_{2n}(z) &=& \pi_{2n}^{1/2,-1/2}(z) \smallskip \\
\pi_{2n+1}(z) &=& \displaystyle\frac z4\left[\left(1+\frac1s\right)\pi_{2n}^{1/2,-3/2}(z) - \frac 3s\pi_{2n}^{3/2,-3/2}(z)\right].
\end{array}
\right.
\end{equation}

To be able to write down an explicit expression for $\mathbf K_N$ we
shall need the weighted versions of the skew orthogonal polynomials
defined by 
\[
\widetilde \pi_n(z) := \pi_n(z)\max\big\{1,|z|\big\}^{-s}.
\]

Then according to Theorem~\ref{thm:5}, the entries of $\mathbf K_N$
are, when $N=2J$ is even, 
\begin{equation}
\label{IDS1}
\left\{
\begin{array}{rrl}
\mathbf K^{(1,1)}_N(u,v) &:=& \varkappa_{2J}(u,v) = {\displaystyle
  2\sum_{j=0}^{J-1}\big[\widetilde\pi_{2j}(u)\widetilde\pi_{2j+1}(v) -
  \widetilde\pi_{2j}(v)\widetilde\pi_{2j+1}(u) \big]} \smallskip \\ 
\mathbf K^{(1,2)}_N(u,v) &:=& \varkappa_{2J} \epsilon(u,v) = {\displaystyle
  2\sum_{j=0}^{J-1}\big[\widetilde\pi_{2j}(u)\epsilon\widetilde\pi_{2j+1}(v)
  - \epsilon\widetilde\pi_{2j}(v)\widetilde\pi_{2j+1}(u)
  \big]} \smallskip \\ 
\mathbf K^{(2,2)}_N(u,v) &:=& \displaystyle \epsilon \varkappa_{2J} \epsilon (u,v) + \frac12 \sgn(u - v) \smallskip \\
&=& {\displaystyle 2\sum_{j=0}^{J-1}\big[\epsilon\widetilde\pi_{2j}(u)\epsilon\widetilde\pi_{2j+1}(v) - \epsilon\widetilde\pi_{2j}(v)\epsilon\widetilde\pi_{2j+1}(u) \big] + \frac12 \sgn(u - v)}
\end{array}
\right..
\end{equation}
We introduce this new notation for the matrix entries, because when $N
= 2J+1$ is odd, the entries are not given as simply as in
Theorem~\ref{thm:5}.  However, in this situation, the 
entries of $\mathbf K_N$ can be computed from \cite{sinclair-2008} or
\cite{mays-forrester} to be 
\begin{equation}
\label{IDS2}
\left\{
\begin{array}{rll}
\mathbf{K}^{(1,1)}_{2J+1}(u,v) &:=& \displaystyle \mathbf{K}^{(1,1)}_{2J}(u,v) - 2\sum_{j=0}^{J-1}\frac{s_{2j}}{s_{2J}}\big[\widetilde\pi_{2J}(u)\widetilde\pi_{2j+1}(v)-\widetilde\pi_{2J}(v)\widetilde\pi_{2j+1}(u)\big] \smallskip \\ 
\mathbf{K}^{(1,2)}_{2J+1}(u,v) &:=& \displaystyle \mathbf{K}^{(1,2)}_{2J}(u,v) - 2\sum_{j=0}^{J-1}\frac{s_{2j}}{s_{2J}}\big[\widetilde\pi_{2J}(u)\epsilon\widetilde\pi_{2j+1}(v)-\epsilon\widetilde\pi_{2J}(v)\widetilde\pi_{2j+1}(u)\big] \smallskip \\
&& \displaystyle + \frac{\widetilde\pi_{2J}(u)\chi_\R(v)}{s_{2J}} \smallskip \\ 
\mathbf{K}^{(2,2)}_{2J+1}(u,v) &:=& \displaystyle
\mathbf{K}^{(2,2)}_{2J}(u,v) - 2\sum_{j=0}^{J-1}\frac{s_{2j}}{s_{2J}}\big[\epsilon\widetilde\pi_{2J}(u)\epsilon\widetilde\pi_{2j+1}(v)-\epsilon\widetilde\pi_{2J}(v)\epsilon\widetilde\pi_{2j+1}(u)\big] \smallskip \\
&&\displaystyle+\frac{\epsilon\widetilde\pi_{2J}(u)\chi_\R(v)-\epsilon\widetilde\pi_{2J}(v)\chi_\R(u)}{s_{2J}}
\end{array}
\right.,
\end{equation}
where $\chi_{\R}$ is the characteristic function of $\R$, and 
\begin{equation}
\label{s_k}
s_k := \int_\R\widetilde\pi_k(x)\upd x.
\end{equation}
In general, expressions in \eqref{IDS2} must contain terms corresponding to constants $s_{2n+1}$ as well. However, it is easy to see from Theorem~\ref{thm:2} that $s_{2n+1}=0$ for all $n$. Thus, only the terms corresponding to $s_{2n}$ remain.
\begin{lemma}
\label{lemma:2}
It holds that
\[
s_{2n} = 2\frac{\Gamma\left(\frac{s+2}2\right)\Gamma\left(\frac{s-2n-1}2\right)}{\Gamma\left(\frac{s+1}2\right)\Gamma\left(\frac{s-2n}2\right)},
\]
where it is understood that $s_{2n}=2$ when $s=\infty$. 
\end{lemma}

\section{A Family of Polynomials}
\label{sec:family-skew-orth}

As one can see from \eqref{pies}, all the skew orthogonal polynomials $\{\pi_{2n},\pi_{2n+1}\}$ can be expressed solely via even degree polynomials $\big\{\pi_{2n}^{\alpha,\beta}\big\}$ for three pairs of parameters $(\alpha,\beta)$. Hence, to derive the results announced in Section~\ref{sec:main}, we shall study polynomials
\begin{equation}
\label{pes}
P_n^{\alpha,\beta}(z) := \frac{1}{\G{1+\alpha} \G{1+\beta}}
\sum_{k=0}^n \frac{\G{k + 1+ \alpha} \G{n - k + 1 + \beta}}{\G{k+1}
  \G{n - k + 1}}  z^k,
\end{equation}
where $\alpha,\beta\not\in\{-1,-2,\ldots\}$. Clearly,
$ \pi_{2n}^{\alpha,\beta}(z) = P_n^{\alpha,\beta}(z^2)$. 

\subsection{Algebraic Properties of the Polynomials}
\label{ssec:algebraic}

The polynomials $P_n^{\alpha,\beta}$ satisfy the following relations.

\begin{prop}
\label{prop:algebraic}
It holds that
\begin{eqnarray}
\label{nn-1}
P_n^{\alpha,\beta}(z) &=&  P_n^{\alpha,\beta-1}(z) + P_{n-1}^{\alpha,\beta}(z) \\
\label{reciprocal}
&=&  z^nP_n^{\beta,\alpha}(1/z) \\
\label{recurrence}
 &=& \left[\frac{n+\alpha} nz + \frac{n+\beta} n\right]P_{n-1}^{\alpha,\beta}(z) - \frac{n+\alpha+\beta}nzP_{n-2}^{\alpha,\beta}(z),  \\
 \label{integral}
 &=& \frac{\Gamma(n+2+\alpha+\beta)}{\Gamma(1+\alpha)\Gamma(1+\beta)\Gamma(n+1)} \int_CB_{\alpha,\beta}(t)(1 - t + tz)^n\mathrm{d}t, 
\end{eqnarray}
where recurrence relations \eqref{recurrence} hold for $n\geq2$ with $P_0^{\alpha,\beta}(z)\equiv1$, $P_1^{\alpha,\beta}(z)=(1+\beta)+(1+\alpha) z$, $C$ is the Pochhammer contour, and $B_{\alpha,\beta}(t):=t^{\alpha}(1-t)^\beta/(1-e^{2\pi\mathrm{i}\alpha})(1-e^{2\pi\mathrm{i}\beta})$.
\end{prop}

Recall that the Pochhammer contour is a contour that winds clockwise around $1$, then clockwise around another $-1$, then counterclockwise around $1$, and then counterclockwise around $-1$. 

The polynomials $P_n^{\alpha,\beta}$ can be expressed via non-standard Jacobi polynomials.

\begin{prop}
\label{prop:Jacobi}
It holds that
\[
P_n^{\alpha,\beta}(z) = (1-z)^n J_n^{-n-1-\alpha,-n-1-\beta}\left(\frac{z+1}{z-1}\right),
\]
where $J_n^{a,b}$ is the $n$-th Jacobi polynomial with parameters $a,b$.
\end{prop}

For a large set of parameters the zeros of $P_n^{\alpha,\beta}$ exhibit definite behavior with respect to the unit circle. Observe that due to \eqref{reciprocal}, we only need to consider the case $\alpha\geq\beta$. Recall also that $\alpha,\beta\not\in\{-1,-2,\ldots\}$.

\begin{prop}
\label{prop:zeros}
\begin{itemize}
\item [(i)] $P_n^{\alpha,\beta}$ has a zero of order $m$ at 1 if and
  only if $n\geq m$ and $m+1+\alpha+\beta=0$ for some $m\in\N$. 
\item [(ii)] The zeros of $P_n^{\alpha,\beta}$ in $\C\setminus\{1\}$
  are simple. 
\item [(iii)] Let $\alpha>\beta$. If either $2+\alpha+\beta>0$ or
  $m+1+\alpha+\beta=0$ for some $m\in\N$, then the zeros of
  $P_n^{\alpha,\beta}$ are contained in $\D\cup\{1\}$. 
\item [(iv)] Let $\alpha=\beta$. If $3+2\alpha>0$ or $m+1+2\alpha=0$
  for some $m$ even, then the zeros of $P_n^{\alpha,\alpha}$ belong to
  $\T$. 
\end{itemize}
\end{prop}

\begin{figure}[!ht]
\centering
\includegraphics[scale=.2]{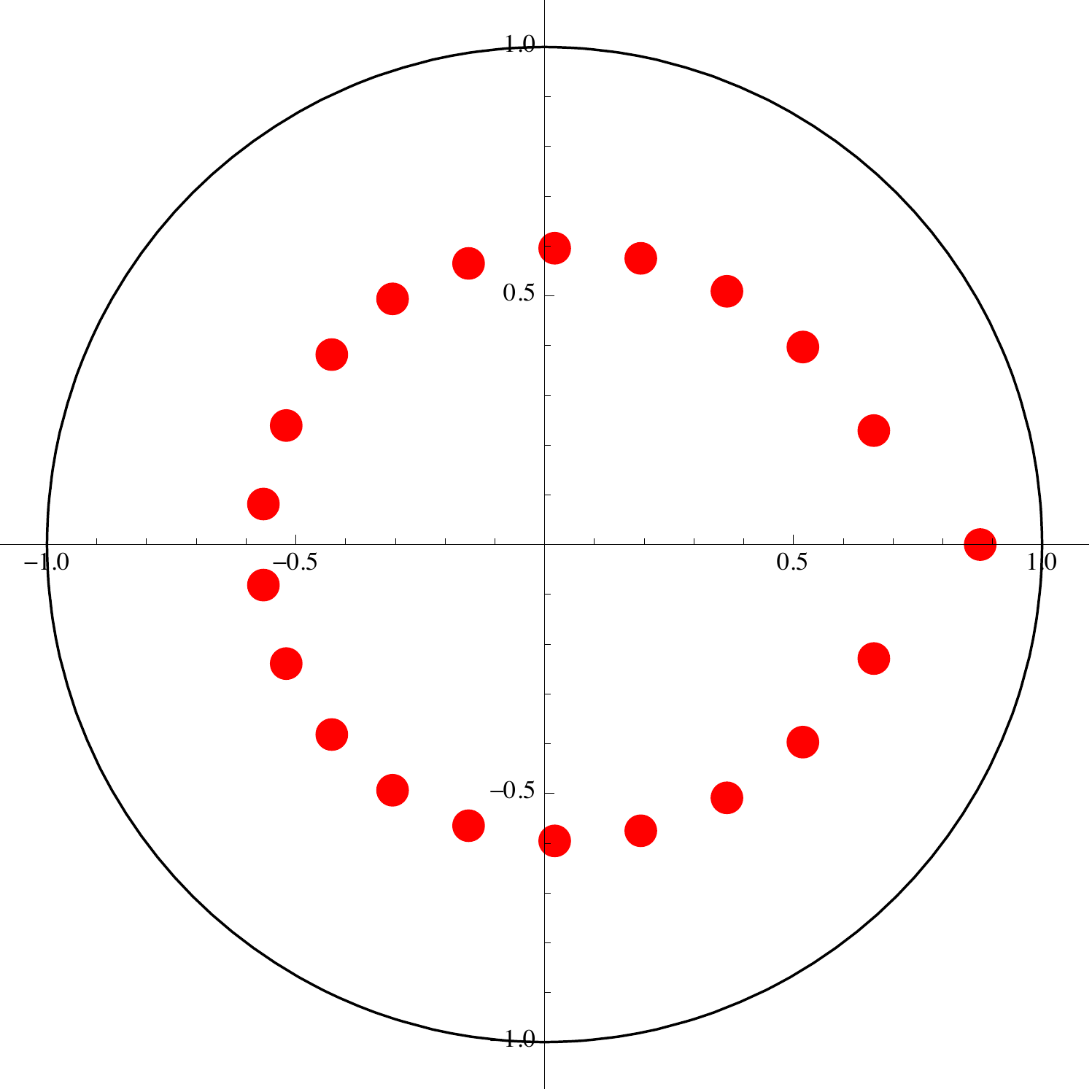}\quad
\includegraphics[scale=.2]{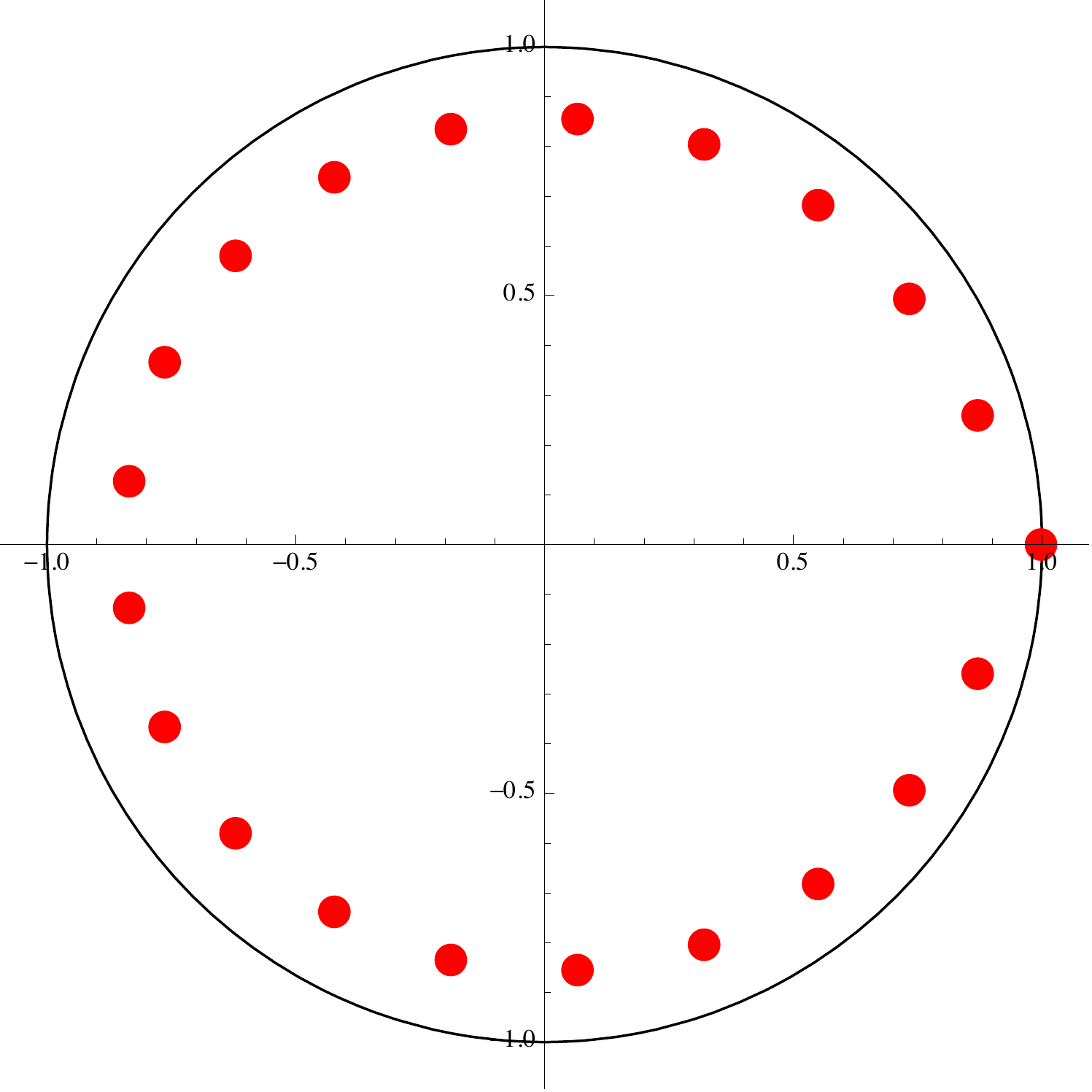}\quad
\includegraphics[scale=.2]{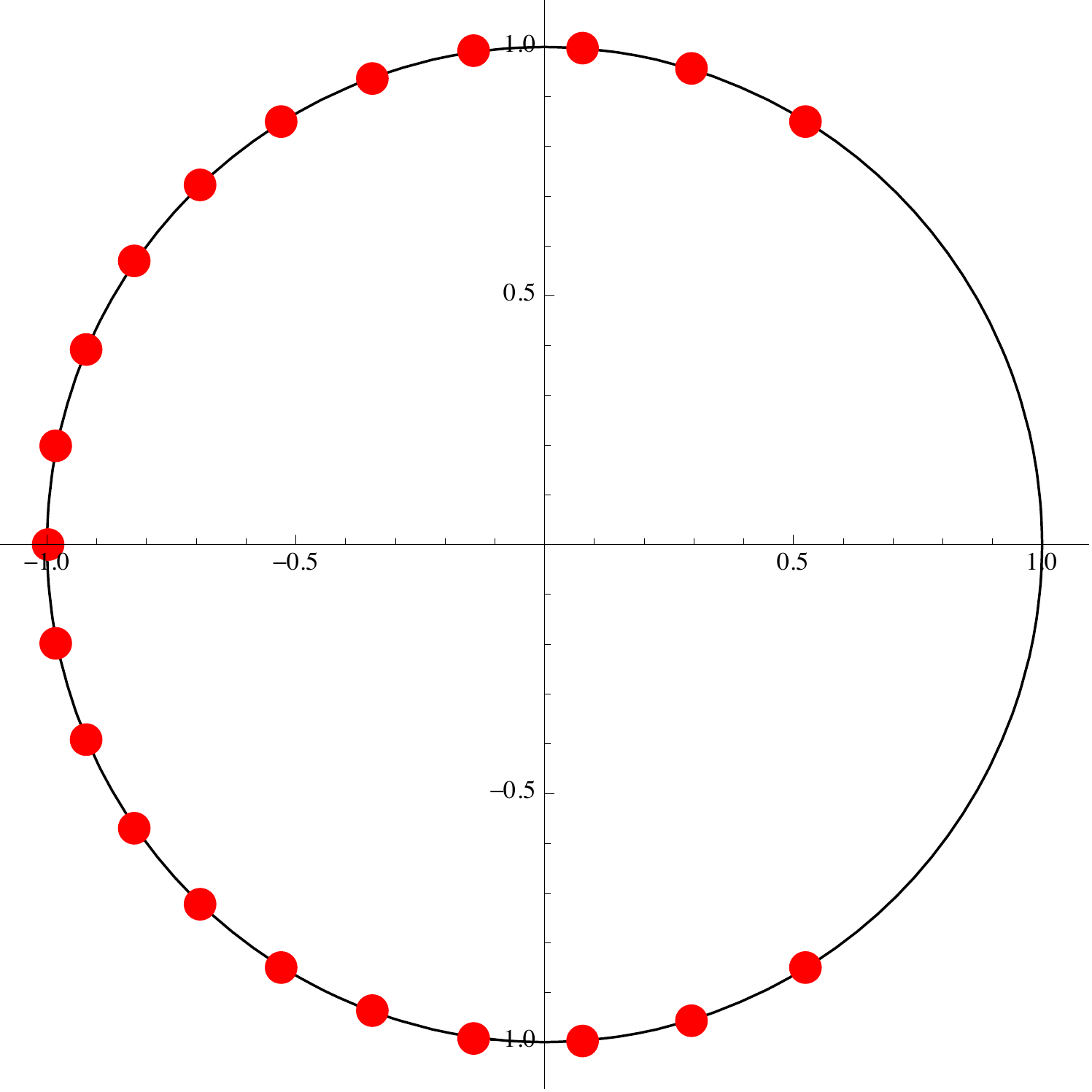}\quad
\includegraphics[scale=.2]{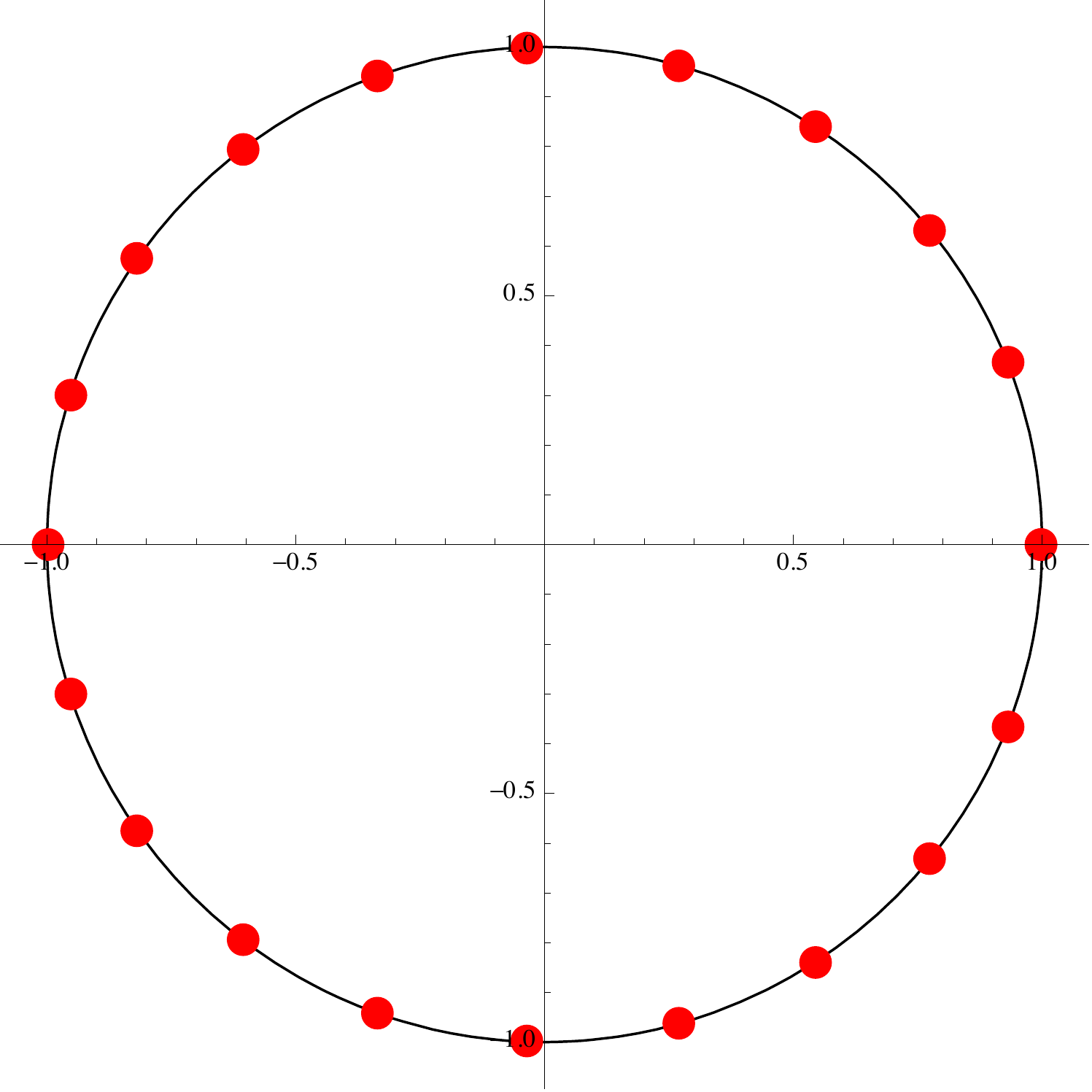}
\caption{\small From left to right: the zeros of $P_{21}^{1.5,-1.5}$, $P_{21}^{-.6,-1.4}$, $P_{21}^{11.5,11.5}$, and $P_{21}^{-1.5,-1.5}$. }
\label{figure1}
\end{figure}

\subsection{Asymptotic Properties of the Polynomials}
\label{ssec:asymptotic}

The polynomials $P_n^{\alpha,\beta}$ enjoy the following asymptotic properties.

\begin{thm}
\label{thm:asymptotic}
Let $(1-z)^{-(1+\gamma)}$ be the branch holomorphic in
$\C\setminus[1,\infty)$ and positive for $z \in (-\infty, 1)$. Then
$as n \rightarrow \infty$, 
\begin{equation}
\label{convergence1}
\left|\frac{P_n^{\alpha,\beta}(z)}{P_n^{\alpha,\beta}(0)} - \frac1{(1-z)^{1+\alpha}}\right| = \frac{o(1)}{|1-z|},
\end{equation}
where $o(1)$ holds uniformly in $\overline\D$ when $\alpha<0$ and $\beta>0$, and in $\D$ otherwise. Furthermore,
\begin{equation}
\label{convergence2}
\left|\frac1{z^n}\frac{P_n^{\alpha,\beta}(z)}{P_n^{\beta,\alpha}(0)} - \frac{1}{(1-1/z)^{1+\beta}}\right| = \frac{o(1)}{|1-1/z|},
\end{equation}
where $o(1)$ holds uniformly in $\overline\Om$ when $\alpha>0$ and $\beta<0$, and in $\Om$ otherwise.
\end{thm}

Observe that
\[
P_n^{\alpha,\beta}(0) = \frac{\Gamma(n+1+\beta)}{\Gamma(1+\beta)\Gamma(n+1)} = \big(1+o_n(1)\big)\frac{(n+1)^\beta}{\Gamma(1+\beta)}
\]
by the properties of the Gamma function.

When $\zeta +z/n\in\overline\Om$, it readily follows from \eqref{convergence2}, the maximum modulus principle, and normal family argument that the following corollary takes place.

\begin{cor}
\label{cor:asymptotic}
Let $\zeta\in\T\setminus\{1\}$, $\alpha>0$, and $\beta<0$. Then
\[
\lim_{n\to\infty}\frac{\Gamma(1+\alpha)}{(n+1)^\alpha}\left(\zeta+\frac zn\right)^{-n}P_n^{\alpha,\beta}\left(\zeta+\frac zn\right) = \left(1-\overline\zeta\right)^{-1-\beta}
\]
locally uniformly in $\C$.
\end{cor}

It is obvious from the previous results that the point $1$ is a special point for polynomials $P_n^{\alpha,\beta}$. To describe the behavior of the polynomials there we need the following definition:
\begin{equation}
\label{Mab}
M_{\alpha,\beta}(z) := \frac{\Gamma(1+\gamma)}{\Gamma(1+\alpha)}\sum_{n=0}^\infty\frac{\Gamma(n+1+\alpha)}{\Gamma(n+1+\gamma)}\frac{z^n}{n!}, \quad  \gamma:=1+\alpha+\beta\not\in\{-1,-2,\ldots\},
\end{equation}
which is a confluent hypergeometric function with parameters $\alpha,\gamma$. 

\begin{prop}
\label{prop:asymptoticat1} 
If $\gamma=1+\alpha+\beta\not\in\{-1,-2,\ldots\}$, then as $n
\rightarrow \infty$, 
\begin{equation}
\label{Pnat1}
P_n^{\alpha,\beta}\left(1+\frac zn\right) = \big(1+o_n(1)\big)\frac{\Gamma(n+1+\gamma)}{\Gamma(1+\gamma)\Gamma(n+1)}M_{\alpha,\beta}(z)
\end{equation}
locally uniformly in $\C$, where $o_n(1)=0$ when $z=0$. 
\end{prop}

\subsection{Asymptotic Properties of the Sums away from the Unit Circle}
\label{ssec:away}

As apparent from \eqref{pies} ---\eqref{IDS2}, the main focus of this work is the asymptotic behavior of the sums
\begin{equation}
\label{sums}
K_N^{\alpha_1,\beta_1,\alpha_2,\beta_2}(z,w) :=\sum_{n=0}^{N-1}P_n^{\alpha_1,\beta_1}(z)P_n^{\alpha_2,\beta_2}(w). 
\end{equation}
Properly renormalized, these sums converge locally uniformly in
$\D\times\D$ and $\Om\times\Om$. To state the results, we shall need
the following notation. For $\beta_1+\beta_2+1<0$ set
$\Lambda_{\beta_1,\beta_2}(\zeta)$ to be 
\[
\left\{
\begin{array}{ll}
\frac{\Gamma(-\beta_1-\beta_2-1)}{\Gamma(-\beta_1)\Gamma(-\beta_2)}\left[{}_2F_1\big(1,1+\beta_1;-\beta_2;\overline\zeta\big) + {}_2F_1\big(1,1+\beta_2;-\beta_1;\zeta\big) - 1\right], & \beta_1,\beta_2\notin\Z_+, \bigskip \\
\frac{\Gamma(-\beta_1-\beta_2-1)}{\Gamma(-\beta_2)\Gamma(1+\beta_2)}\left[\zeta^{1+\beta_1}{}_2F_1\big(1,2+\beta_1+\beta_2;1;\zeta\big)\right], & \beta_1\in\Z_+, \bigskip \\
\frac{\Gamma(-\beta_1-\beta_2-1)}{\Gamma(-\beta_2)\Gamma(1+\beta_2)}\left[\overline\zeta^{1+\beta_2}{}_2F_1\big(1,2+\beta_1+\beta_2;1;\overline\zeta\big)\right], & \beta_2\in\Z_+,
\end{array}
\right.
\]
where
\begin{equation}
\label{hyper_geom}
{}_2F_1\big(a,b;c;z\big) := \frac{\Gamma(c)}{\Gamma(a)\Gamma(b)}\sum_{n=0}^\infty\frac{\Gamma(n+a)\Gamma(n+b)}{\Gamma(n+c)\Gamma(n+1)}z^n, \quad a,b,c\not\in\Z_-.
\end{equation}
The function $\Lambda_{\beta_1,\beta_2}$ is continuous on $\T\setminus\{-1\}$ with an integrable singularity at $1$ when $\beta_1+\beta_2+1\geq-1$ \cite[Sec.~15.4]{Olver} and is continuous on the whole circle when $\beta_1+\beta_2+1<-1$. It can be verified that the Fourier series corresponding to $\Lambda_{\beta_1,\beta_2}$ is given by
\begin{equation}
\label{Lambda12}
\frac{\Gamma(-\beta_1-\beta_2-1)}{\Gamma(1+\beta_1)\Gamma(1+\beta_2)}\left[\sum_{m=0}^\infty \frac{\Gamma(1+\beta_1)}{\Gamma(-\beta_2)}\frac{\Gamma(m+1+\beta_2)}{\Gamma(m-\beta_1)}\zeta^m + \sum_{m=1}^\infty\frac{\Gamma(1+\beta_2)}{\Gamma(-\beta_1)}\frac{\Gamma(m+1+\beta_1)}{\Gamma(m-\beta_2)}\overline\zeta^m\right],
\end{equation}
where it is understood that the terms containing $\Gamma(-\beta_j)$ are zero when $\beta_j$ is a non-negative integer. 

\begin{thm}
\label{thm:Kasymp}
It holds that
\begin{equation}
\label{KN00}
\lim_{N\to\infty}K_N^{\alpha_1,\beta_1,\alpha_2,\beta_2}(0,0) = \left\{
\begin{array}{ll}
\pm\infty, & \beta_1+\beta_2 +1 \geq 0, \smallskip \\
\frac{\Gamma(-\beta_1-\beta_2-1)}{\Gamma(-\beta_1)\Gamma(-\beta_2)}, & \beta_1+\beta_2+1<0,
\end{array}
\right.
\end{equation}
where the sign in the first case is the same as the sign of $\Gamma(1+\beta_1)\Gamma(1+\beta_2)$. Moreover, we have that
\begin{equation}
\label{SumInside1}
\lim_{N\to\infty}\frac{K_N^{\alpha_1,\beta_1,\alpha_2,\beta_2}(z,w)}{K_N^{\alpha_1,\beta_1,\alpha_2,\beta_2}(0,0)} = \frac1{(1-z)^{1+\alpha_1}(1-w)^{1+\alpha_2}}
\end{equation}
locally uniformly in $\D\times\D$ when $\beta_1+\beta_2+1\geq0$, and
\begin{equation}
\label{SumInside2}
\lim_{N\to\infty}K_N^{\alpha_1,\beta_1,\alpha_2,\beta_2}(z,w) = \frac1{2\pi}\int_\T\frac{\Lambda_{\beta_1,\beta_2}(\zeta)|\upd\zeta|}{\big(1-z\overline\zeta\big)^{1+\alpha_1}\big(1-w\zeta\big)^{1+\alpha_2}}
\end{equation}
locally uniformly in $\D\times\D$ when $\beta_1+\beta_2+1<0$. Finally, it holds that
\begin{equation}
\label{SumOutside}
\lim_{N\to\infty}\frac{K_N^{\alpha_1,\beta_1,\alpha_2,\beta_2}(z,w)}{N^{\alpha_1+\alpha_2}(zw)^N} = \frac1{\Gamma(1+\alpha_1)\Gamma(1+\alpha_2)}\frac1{(zw-1)(1-1/z)^{1+\beta_1}(1-1/w)^{1+\beta_2}}
\end{equation}
uniformly on closed subsets of $\Om\times\Om$.
\end{thm}

\subsection{Asymptotic Properties of the Sums on the Unit Circle}
\label{ssec:on}

Theorem~\ref{thm:Kasymp} shows that non-trivial scaling limits of the sums $K_N^{\alpha_1,\beta_1,\alpha_2,\beta_2}$ can appear only on $\T\times\T$. To derive such limits we need rather precise knowledge of the behavior of the polynomials on the unit circle. Hence, in the light of \eqref{convergence2}, we shall only consider parameters satisfying $\alpha>0$ and $\beta<0$. To describe the aforementioned scaling limits, set
\begin{equation}
\label{Egamma}
E_\gamma(\tau) := (1+\gamma)\int_0^1x^\gamma e^{\tau x}\mathrm{d}x, \qquad \gamma>-1,
\end{equation}
where the normalization is chosen so $E_\gamma(0)=1$. Clearly, it holds that $E_0(\tau):=\frac{e^\tau-1}{\tau}$ and $E_\gamma^\prime(\tau)=\frac{\gamma+1}{\gamma+2}E_{\gamma+1}(\tau)$.

\begin{thm}
\label{thm:scaling}
Let $\alpha_i>0$ and $\beta_i<0$, $i\in\{1,2\}$. Then for every $\zeta\in\T\setminus\{1\}$ it holds that
\[
K_N^{\alpha_1,\beta_1,\alpha_2,\beta_2}\big(\zeta,\overline\zeta\big) = \left[\frac{(1-\overline\zeta)^{-1-\beta_1}(1-\zeta)^{-1-\beta_2}}{1+\alpha_1+\alpha_2}+o_N(1)\right]\frac{\Gamma(N+1+\alpha_1+\alpha_2)}{\Gamma(1+\alpha_1)\Gamma(1+\alpha_2)\Gamma(N)} 
\]
as $N\to\infty$ and
\[
\lim_{N\to\infty} \frac{K_N^{\alpha_1,\beta_1,\alpha_2,\beta_2}\left(\zeta+\frac {a_1}N,\overline\zeta+\frac {\overline a_2}N\right)}{K_N^{\alpha_1,\beta_1,\alpha_2,\beta_2}\big(\zeta,\overline\zeta\big)} = E_{\alpha_1+\alpha_2}\left(a_1\overline\zeta+\overline a_2\zeta\right),
\]
uniformly for $a_1,a_2$ on compact subsets of $\C$. 
\end{thm}

The scaling limit at 1 is no longer described by \eqref{Egamma}, but rather by 
\begin{equation}
\label{secondE}
E_{\alpha_1,\beta_1,\alpha_2,\beta_2}(\tau_1,\tau_2) := (1+\gamma) \int_0^1 x^\gamma M_{\alpha_1,\beta_1}(\tau_1x)M_{\alpha_2,\beta_2}(\tau_2x) \mathrm{d}x,
\end{equation}
where $\gamma:=2+\alpha_1+\beta_1+\alpha_2+\beta_2>-1$ and $M_{\alpha,\beta}$ was defined in \eqref{Mab}. 

\begin{thm}
\label{thm:scaling1}
Let $\alpha_i>0$ and $\beta_i<0$, $i\in\{1,2\}$. If $\gamma=2+\alpha_1+\beta_1+\alpha_2+\beta_2>-1$, then
\[
K_N^{\alpha_1,\beta_1,\alpha_2,\beta_2}(1,1) = \frac{1+o_N(1)}{\Gamma(2+\alpha_1+\beta_1)\Gamma(2+\alpha_2+\beta_2)}\frac1{1+\gamma}\frac{\Gamma(N+1+\gamma)}{\Gamma(N)}
\]
as $N\to\infty$ and
\begin{equation}
\label{scaling1}
\lim_{N\to\infty} \frac{K_N^{\alpha_1,\beta_1,\alpha_2,\beta_2}\left(1+\frac {a_1}N,1+\frac {a_2}N\right)}{K_N^{\alpha_1,\beta_1,\alpha_2,\beta_2}(1,1)} = E_{\alpha_1,\beta_1,\alpha_2,\beta_2}(a_1,a_2)
\end{equation}
uniformly for $a_1,a_2$ on compact subsets of $\C$.
\end{thm}

\section{Proofs}

\subsection{Proof of Theorem~\ref{thm:3}}

The following lemma is needed both for the proofs of Theorem~\ref{thm:3} and Lemma~\ref{lemma:2}.
\begin{lemma}
\label{lemma:7}
Let $a,b \in \R$ with $b \not\in \{0, -1, -2, \ldots\}$.  Then, 
\begin{align*}
& \frac{1}{\G{b} \G{a}} \sum_{k=0}^n \frac{\G{k +
    a} \G{n - k + b}}{\G{k+1} \G{n - k + 1}} \bigg\{
\prod_{j = 1}^k \frac{j - b}{j + a - 1} \bigg\} \frac{1}{x -
  k +b} \\
 & \hspace{7cm}  = \frac{x(x-1) \cdots (x - n + 1)}{(x + b) (x +
  b - 1) \cdots (x + b - n) }.
\end{align*}
\end{lemma}
\begin{proof}
Since
\[
\frac{\G{k + a}}{\G{a}} = \prod_{j=1}^k \frac{1}{j+a-1},
\]
it suffices to prove that
\[
\frac{1}{\G{b}} \sum_{k=0}^n \frac{\G{n - k + b}}{\G{k+1} \G{n - k +
    1}} \bigg\{ 
\prod_{j = 1}^k (j - b) \bigg\} \frac{1}{x -
  k +b}  = \frac{x(x-1) \cdots (x - n + 1)}{(x + b) (x +
  b - 1) \cdots (x + b - n) }.
\]
The coefficient of $(x-k+b)^{-1}$ in the partial fractions decomposition
of the rational function on the right hand side is 
\begin{align*}
& \frac{x(x-1) \cdots (x - n + 1)}{(x + b) (x +
  b - 1) \cdots (x + b - n) } (x + b -k) \bigg|_{x = k-b} \\
& \hspace{2cm}  =
\bigg( \prod_{j=0}^{k-1} (k - b - j)  \prod_{\ell=k}^{n-1} (k - b - \ell) \bigg) \bigg/ \bigg( \prod_{j=0}^{k-1}
(k - j) \prod_{\ell = k+1}^n (k - \ell) \bigg) \\
& \hspace{2cm}  =
\bigg( \bigg\{\prod_{j=1}^{k} (j - b)\bigg\}  (-1)^{n-k}
\prod_{\ell=0}^{n-k-1} (b + \ell) 
\bigg) \bigg/ \bigg( \Gamma(k+1) 
(-1)^{n-k} \Gamma(n-k+1) \bigg) \\
& \hspace{2cm}  =
\bigg( \bigg\{\prod_{j=1}^{k} (j - b) \bigg\} \Gamma(n-k+b)
\bigg) \bigg/ \bigg(\Gamma(b) \Gamma(k+1) 
\Gamma(n-k+1) \bigg),
\end{align*}
which proves the lemma.
\end{proof}

\begin{proof}[Proof of Theorem~\ref{thm:3}]
For the moment, let us write $\pi_{2n} = \Gamma(\alpha+1) \Gamma(\beta+1)
\pi_{2n}^{\alpha, \beta}$ and $\pi_{2n+1} = \Gamma(-\beta-1)
\Gamma(\beta + 1)\pi_{2n+1}^{\beta}$. Then
\begin{align*}
&\la  \pi_{2n}(z) | z^{2m+1} \ra^{\alpha, \beta} =
 \sum_{k=0}^n \frac{\G{k + \alpha + 1} \G{n - k + 
    \beta + 1}}{\G{k+1} \G{n - k + 1}} \la z^{2k} | z^{2m+1}
\ra^{\alpha, \beta} \\
& \hspace{2cm} = C_m \sum_{k=0}^n \frac{\G{k + \alpha + 1} \G{n - k +
    \beta + 1}}{\G{k+1} \G{n - k + 1}} \bigg\{\prod_{j = 1}^k \frac{j
  - 1 - \beta}{j + \alpha} \bigg\} \frac{1}{m -
  k + 1 + \beta},
\end{align*}
which by setting $a = \alpha + 1$ and $b = \beta + 1$ in
Lemma~\ref{lemma:7}, is equal to 0 for $m = 0, 1, \ldots, n$. 
Since $\pi_{2m+1}^{\alpha, \beta}$ is odd of degree $2m+1$, we have
$\la \pi_{2n}^{\alpha, \beta} | \pi_{2m+1}^{\alpha, \beta}
\ra^{\alpha, \beta} = 0$ for these values of $m$.  

Similarly, looking at 
\begin{align*}
\la z^{2m} |  \pi_{2n+1}(z) \ra^{\alpha, \beta} &=
 \sum_{k=0}^n \frac{\G{k + \beta
    + 2} \G{n - k -
    \beta - 1} }{\G{k+1} \G{n - k + 1} C_k} \la z^{2m}
| z^{2k+1} \ra^{\alpha, \beta}  \\
&=
- \bigg\{\prod_{j = 1}^m \frac{j
  - 1 - \beta}{j + \alpha} \bigg\} \sum_{k=0}^n \frac{\G{k + \beta
    + 2} \G{n - k -
    \beta - 1} }{\G{k+1} \G{n - k + 1}} \frac{1}{m - k - 1 - \beta} \\
&=
-   (\beta + 1) \bigg\{\prod_{j = 1}^m \frac{j
  - 1 - \beta}{j + \alpha} \bigg\} \\
& \hspace{1cm} \times \sum_{k=0}^n \frac{\G{k + \beta
    + 1} \G{n - k -
    \beta - 1} }{\G{k+1} \G{n - k + 1}} \bigg\{ \frac{k + \beta +
  1}{\beta + 1} 
\bigg\}\frac{1}{m - k - 1 - \beta},
\end{align*}
which by Lemma~\ref{lemma:7} is equal to 0 for $m=0,1,\ldots, n-1$ by
setting $a = \beta + 1$ and $b = -\beta - 1$.

Turning to $\la  \pi^{\alpha, \beta}_{2n} |  \pi^{\beta}_{2n+1} \ra^{\alpha, \beta}$,
\begin{align*}
\la  \pi_{2n}^{\alpha, \beta} |  \pi_{2n+1}^\beta \ra^{\alpha, \beta} &=  \frac{1}{\G{\alpha
    + 1} \G{\beta + 1}} \sum_{k=0}^n
\frac{\G{k + \alpha + 1} \G{n - k + \beta + 1}}{\G{k+1} \G{n - k + 1}}
\la z^{2k} |  \pi_{2n+1}^{\beta}(z) \ra^{\alpha, \beta} \\
&= \frac{\G{n + \alpha + 1}}{\G{\alpha + 1} \G{n+1}}
\la z^{2n} |  \pi^{\beta}_{2n+1}(z) \ra^{\alpha, \beta} \\
&= -\frac{\G{n + \alpha + 1} (\beta+1)}{\G{\alpha + 1} \G{n+1}}
\bigg\{\prod_{j = 1}^n \frac{j - 1 - \beta}{j + \alpha} \bigg\}
\frac{1}{\G{-\beta-1} \G{\beta+1}} \\
& \hspace{1cm} \times \sum_{k=0}^n \frac{\G{k + \beta
    + 1} \G{n - k -
    \beta - 1} }{\G{k+1} \G{n - k + 1}} \bigg\{ \frac{k + \beta +
  1}{\beta + 1} 
\bigg\}\frac{1}{n - k - 1 - \beta} \\
&= -\frac{ (\beta+1)}{\G{n+1}}
\bigg\{\prod_{j = 1}^n (j - 1 - \beta) \bigg\} \Gamma(n+1)
\prod_{j=0}^n \frac{1}{n-\beta-j-1} =1,
\end{align*}
where again, we use Lemma~\ref{lemma:7} with $a = \beta+1$ and
$b=-\beta-1$.
\end{proof}
\begin{proof}[Proof of Lemma~\ref{lemma:2}]
It holds that
\[
\int_\R x^{2k}\max\big\{1,|x|\big\}^{-s}\upd x = 2\int_0^1 x^{2k}\upd x + 2\int_1^\infty x^{2k-s}\upd x = \frac{2s}{(2k+1)(s-2k-1)},
\]
Then it follows from Theorem~\ref{thm:2} and Lemma~\ref{lemma:7} that
\begin{eqnarray}
s_{2n} &=& \frac s{\Gamma(1/2)\Gamma(1/2)}\sum_{k=0}^n \frac{\G{k +1/2} \G{n - k + 1/2}}{\G{k+1} \G{n-k+1} }\left\{\prod_{j=1}^k\frac{j-1/2}{j+1/2-1}\right\}\frac{1}{s/2-1-k+1/2} \nonumber\\
&=& 2 \frac{(s/2)(s/2-1)(s/2-2) \cdots (s/2 - n)}{(s/2 -1/2) (s/2 - 1/2-1) \cdots (s/2 -1/2 - n) } = 2\frac{\Gamma\left(\frac{s+2}2\right)\Gamma\left(\frac{s-2n-1}2\right)}{\Gamma\left(\frac{s+1}2\right)\Gamma\left(\frac{s-2n}2\right)} \nonumber
\end{eqnarray}
and the case $s=\infty$ follows by taking the limit.
\end{proof}

\subsection{Proofs of Propositions~\ref{prop:algebraic}---\ref{prop:zeros}}
\begin{proof}[Proof of Proposition~\ref{prop:algebraic}]
By the very definition the difference $\Gamma(1+\alpha)\Gamma(1+\beta)\big(P_n^{\alpha,\beta}(z) - P_{n-1}^{\alpha,\beta}(z)\big)$ is equal to
\begin{align*}
&
\frac{\Gamma(n+1+\alpha)}{\Gamma(n+1)}\frac{\Gamma(\beta+1)}{\Gamma(1)}z^n 
  +
  \sum_{k=0}^{n-1}\frac{\Gamma(k+1+\alpha)}{\Gamma(k+1)}\left(\frac{\Gamma(n-k+1+\beta)}{\Gamma(n-k+1)}-\frac{\Gamma(n-k+\beta)}{\Gamma(n-k)}\right)z^k 
  \\ 
  & \hspace{3cm}
  =\frac{\Gamma(n+\alpha)}{\Gamma(n+1)}\frac{\Gamma(\beta+1)}{\Gamma(1)}z^n
  +
  \beta\sum_{k=0}^{n-1}\frac{\Gamma(k+1+\alpha)}{\Gamma(k+1)}\frac{\Gamma(n-k+\beta)}{\Gamma(n-k+1)}z^k
  \\
 & \hspace{3cm} 
  =\beta\sum_{k=0}^n\frac{\Gamma(k+1+\alpha)}{\Gamma(k+1)}\frac{\Gamma(n-k+\beta)}{\Gamma(n-k+1)}z^k 
  = \beta \Gamma(1+\alpha)\Gamma(\beta) P_n^{\alpha,\beta-1}(z),
\end{align*}
which establishes \eqref{nn-1}. Relation \eqref{reciprocal} is rather obvious. 

Consider the right-hand side of \eqref{recurrence} multiplied by $\Gamma(1+\alpha)\Gamma(1+\beta)$. The coefficient of $z^n$ is
\[
\frac{n+\alpha} n\frac{\Gamma(n+\alpha)}{\Gamma(n)}\frac{\Gamma(1+\beta)}{\Gamma(1)} = \frac{\Gamma(n+1+\alpha)}{\Gamma(n+1)}\frac{\Gamma(1+\beta)}{\Gamma(1)};
\]
the constant coefficient is
\[
\frac{n+\beta} n\frac{\Gamma(1+\alpha)}{\Gamma(1)}\frac{\Gamma(n+\beta)}{\Gamma(n)} = \frac{\Gamma(1+\alpha)}{\Gamma(1)}\frac{\Gamma(n+1+\beta)}{\Gamma(n+1)};
\]
and the coefficient of $z^k$, $k\in\{1,\ldots,n-1\}$, is
\[
\begin{array}{l}
\displaystyle \frac{n+ \alpha} n\frac{\Gamma(k+\alpha)}{\Gamma(k)}\frac{\Gamma(n-k+1+\beta)}{\Gamma(n-k+1)} + \frac{n+\beta} n\frac{\Gamma(k+1+\alpha)}{\Gamma(k+1)}\frac{\Gamma(n-k+\beta)}{\Gamma(n-k)}  \smallskip \\
\displaystyle \hspace{2.5in}- \frac{n+\alpha+\beta}n\frac{\Gamma(k+\alpha)}{\Gamma(k)}\frac{\Gamma(n-k+\beta)}{\Gamma(n-k)},
\end{array}
\]
which is equal to
\begin{align*}
& \frac{\Gamma(k+\alpha)}{\Gamma(k)}\frac{\Gamma(n-k+\beta)}{\Gamma(n-k)}\left[\frac{n+ \alpha} n\frac {n-k+\beta} {n-k} + \frac{n+\beta} n\frac{n+\alpha} k - \frac{n+ \alpha + \beta} n\right] \\
& \hspace{3cm} =
\frac{\Gamma(k+\alpha)}{\Gamma(k)}\frac{\Gamma(n-k+\beta)}{\Gamma(n-k)}
\frac{\beta k+\alpha(n-k)+\alpha\beta+k(n-k)}{k(n-k)} \\
& \hspace{3cm} =
\frac{\Gamma(k+1+\alpha)}{\Gamma(k+1)}\frac{\Gamma(n-k+1+\beta)}{\Gamma(n-k+1)}. 
\end{align*}

Finally, recall that for any values $u,v$ it holds that
\begin{equation}
\label{beta}
\frac{\Gamma(u)\Gamma(v)}{\Gamma(u+v)} = B(u,v) = \frac1{(1-e^{2\pi\mathrm{i}u})(1-e^{2\pi\mathrm{i}v})}\int_C t^{u-1}(1-t)^{v-1}\mathrm{d}t,
\end{equation}
where $C$ is the Pochhammer contour. Then
\begin{eqnarray}
P_n^{\alpha,\beta}(z) &=& \frac{\Gamma(n+2+\alpha+\beta)}{\Gamma(1+\alpha)\Gamma(1+\beta)\Gamma(n+1)}\sum_{k=0}^n B\big(k+1+\alpha,n-k+1+\beta\big)\binom{n}{k}z^k\nonumber \\
&=& \frac{\Gamma(n+2+\alpha+\beta)}{(1-e^{2\pi\mathrm{i}\alpha})(1-e^{2\pi\mathrm{i}\beta})\Gamma(n+1)} \int_C\sum_{k=0}^n\binom nk (1-t)^{n-k+\beta}t^{k+\alpha}z^k\mathrm{d}t, \nonumber
\end{eqnarray}
from which \eqref{integral} easily follows.
\end{proof}

\begin{proof}[Proof of Proposition~\ref{prop:Jacobi}]
By definition, $(1-z)^n
J_n^{-n-1-\alpha,-n-1-\beta}\left(\frac{z+1}{z-1}\right)$ is equal to 
\[
\frac{(1-z)^n}{n!}\frac{\Gamma(-\alpha)}{\Gamma(-n-1-\alpha-\beta)}\sum_{m=0}^n\binom nm \frac{\Gamma(-n-1+m-\alpha-\beta)}{\Gamma(-n+m-\alpha)}\frac1{(z-1)^m},
\]
which can be rewritten as
\[
(-1)^n\frac{\Gamma(-\alpha)}{\Gamma(-n-1-\alpha-\beta)}\sum_{j=0}^n\frac1{j!(n-j)!} \frac{\Gamma(-j-1-\alpha-\beta)}{\Gamma(-j-\alpha)}(z-1)^j.
\]
Expanding $(z-1)^j$ into the powers of $z$, we get that the above polynomial can be expressed as
\[
(-1)^n\sum_{k=0}^n\left(\sum_{j=k}^n\frac{(-1)^{j-k}}{(n-j)!(j-k)!}\frac{\Gamma(-j-1-\alpha-\beta)}{\Gamma(-n-1-\alpha-\beta)}\frac{\Gamma(-\alpha)}{\Gamma(-j-\alpha)}\right)\frac{z^k}{\Gamma(k+1)}
\]
and respectively as
\[
\sum_{k=0}^n\left(\sum_{j=k}^n(-1)^{j-k}\binom{n-k}{j-k}\frac{\Gamma(j+1+\alpha)}{\Gamma(1+\alpha)}\frac{\Gamma(n+2+\alpha+\beta)}{\Gamma(j+2+\alpha+\beta)}\right)\frac{z^k}{\Gamma(k+1)\Gamma(n-k+1)}.
\]
Thus, the claim of the proposition will follow if we show that
\begin{equation}
\label{Ineedaname}
\sum_{m=0}^M(-1)^m\binom Mm\frac{\Gamma(m+x)}{\Gamma(m+x+y)} = \frac{\Gamma(x)}{\Gamma(y)}\frac{\Gamma(M+y)}{\Gamma(M+x+y)},
\end{equation}
where $M:=n-k$, $m:=j-k$, $x=k+1+\alpha$, and $y=1+\beta$. As the left-hand side of \eqref{Ineedaname} is equal to
\begin{align*}
& \sum_{m=0}^M(-1)^m\binom
Mm\frac{\Gamma(m+x)}{\Gamma(m+x+y)}+\sum_{m=1}^{M+1}(-1)^m\binom
M{m-1}\frac{\Gamma(m+x)}{\Gamma(m+x+y)} \\ 
& \hspace{4cm} = \sum_{m=0}^M(-1)^m\binom Mm\left(\frac{\Gamma(m+x)}{\Gamma(m+x+y)} - \frac{\Gamma(m+1+x)}{\Gamma(m+1+x+y)}\right)
\end{align*}
and since
\[
\frac{\Gamma(x)}{\Gamma(y)}\frac{\Gamma(M+y)}{\Gamma(M+x+y)} - \frac{\Gamma(1+x)}{\Gamma(y)}\frac{\Gamma(M+y)}{\Gamma(M+1+x+y)} = \frac{\Gamma(x)}{\Gamma(y)}\frac{\Gamma(M+1+y)}{\Gamma(M+1+x+y)}
\]
claim \eqref{Ineedaname} follows by induction.
\end{proof}

\begin{proof}[Proof of Proposition~\ref{prop:zeros}]
We start by showing that
\begin{equation}
\label{m1}
P_m^{\alpha,\beta}(z) = \frac{\Gamma(m+1+\alpha)}{\Gamma(m+1)\Gamma(1+\alpha)}(z-1)^m
\end{equation}
if and only if $m+1+\alpha+\beta=0$. Indeed, if \eqref{m1} takes
place, it is enough to compare the coefficient next to $z^{m-1}$ in
\eqref{m1} and \eqref{pes} to deduce that $m+1+\alpha+\beta=0$. To
prove the claim in the other direction, assume that \eqref{m1} holds for some fixed $m$ and all $\alpha,\beta\notin\{-1,-2,\ldots\}$ such that $m+1+\alpha+\beta=0$. Then it follows from recurrence formula \eqref{recurrence} that
\begin{equation}
\label{m2}
P_{m+1}^{\alpha,\beta}(z) = \left[\frac{m+1+\alpha}{m+1}z+\frac{m+1+\beta}{m+1}\right] P_m^{\alpha,\beta}(z).
\end{equation}
Now, take $\alpha,\beta$ such that $m+2+\alpha+\beta=0$. Then \eqref{m1} and \eqref{m2} hold with $\beta$ replaced by $\beta+1$. Hence, we get from \eqref{nn-1} that
\[
P_{m+1}^{\alpha,\beta}(z) = \left[\frac{m+1+\alpha}{m+1}z+\frac{\beta+1}{m+1}\right] P_m^{\alpha,\beta+1}(z) = \frac{\Gamma(m+2+\alpha)}{\Gamma(m+2)\Gamma(1+\alpha)}(z-1)^{m+1}.
\]
Thus, to prove \eqref{m1} in full generality it only remains to
establish the base case $m=1$, which follows easily since
$P_1^{\alpha,\beta}(z) = (1+\alpha)(z-1)$ when $2+\alpha+\beta=0$.  

We just established in \eqref{m1} and \eqref{m2} that $P_m^{\alpha,\beta}$ and $P_{m+1}^{\alpha,\beta}$ vanish at 1 with order $m$ whenever $m+1+\alpha+\beta=0$. Recurrence relations \eqref{recurrence} immediately yield that the same is true for all $P_n^{\alpha,\beta}$, $n\geq m$. Reciprocally, assume that $P_n^{\alpha,\beta}$ vanishes at $1$ with order $m$. We can suppose that $n>m$ as the case $n=m$ is covered by \eqref{m1}. By Proposition~\ref{prop:Jacobi}, we get that
\[
P_n^{\alpha,\beta}(z) = (1-z)^n\frac{(-1)^n}{2^nn!}J_n(x); \qquad
x=\frac{z+1}{z-1}, 
\]
where $J_n$ is a constant multiple of the Jacobi
polynomial $J_n^{-n-1-\alpha,-n-1-\beta}$. Observe that the map
$z\mapsto(z+1)/(z-1)$ is conformal, maps 1 to $\infty$ and sends the
unit disk $\D$ onto the left half-plane (the unit circle into the
imaginary axis). In particular, $P_n^{\alpha,\beta}$ vanishes at $1$
with order $m$ if and only if $\deg(J_n)=n-m$. According to the
Rodrigues' formula for the Jacobi polynomials it holds that 
\[
\frac{\mathrm{d}^n}{\mathrm{d}x^n}\left\{\frac1{(1-x)^{1+\alpha}(1+x)^{1+\beta}}\right\} = \frac{J_n(x)}{(1-x)^{n+1+\alpha}(1+x)^{n+1+\beta}}
\]
and therefore
\begin{equation}
\label{Jn+1}
J_{n+1}(x) = (1-x^2)J_n^\prime(x) + \big[(2n+2+\alpha+\beta)x+\alpha-\beta\big]J_n(x).
\end{equation}
Particularly, it follows that $\deg(J_{n+1})=n+1-m$. That is, $P_{n+1}^{\alpha,\beta}$ vanishes at 1 with order $m$ as well. Furthermore, recurrence formula \eqref{recurrence} yields in this case that $P_{n-1}^{\alpha,\beta}$ is divisible by $(z-1)^m$. Repeatedly applying \eqref{recurrence}, we obtain that $P_m^{\alpha,\beta}$ must be a multiple of $(z-1)^m$ too and therefore $m+1+\alpha+\beta=0$ by \eqref{m1}.  This finishes the prove of the first claim of the proposition.

Proving the second claim of the proposition is tantamount to showing that the zeros of $J_n$ are simple. To the contrary, assume that $J_n$ has a zero, say $x_0$, of multiplicity $k\geq2$. Observe that $x_0\neq\pm1$ as otherwise $P_n^{\alpha,\beta}$ would have to vanish at 0 or had a degree less than $n$, which contradicts the very definition of this polynomial. From our assumption, $x_0$ is a zero of $J_n^\prime$ of multiplicity $k-1$ and therefore it is a zero of $J_{n+1}$ of multiplicity exactly $k-1$ by \eqref{Jn+1}. Then we can infer from \eqref{recurrence} that $J_{n-1}$ must vanish at $x_0$ with order exactly $k-1$. Further, using \eqref{Jn+1} with $n$ replaced by $n-1$, we get that $J_{n-1}^\prime$ has to vanish at $x_0$ with the same order which is clearly impossible.

Now suppose $\alpha>\beta$. Assume that either $2+\alpha+\beta>0$, in
which case set $m=0$, or $m+1+\alpha+\beta=0$ for some $m\in\N$. Under
these conditions we have that $2m+2+\alpha+\beta>0$ and
$\deg(J_n)=n-m$, $n\geq m$. Recall that the interior of the unit disk
gets mapped into the left half-plane and therefore we want to
establishe that this is where the zeros of $J_n$ are. As $J_m$
is a constant by \eqref{m1}, it holds that
\begin{equation}
\label{Jm1}
J_{m+1}(x) = \big[(2m+2+\alpha+\beta)x+\alpha-\beta\big]J_m(x)
\end{equation}
is of degree 1 and vanishes on the negative real axis. Denote by $-x_i$ the zeros of $J_n$. Then for $n\geq m+1$ we have that
\begin{equation}
\label{Jn1Jn}
\big(J_{n+1}/J_n\big)(x)= \sum\frac{1+xx_i}{x+x_i} + (n+m+2+\alpha+\beta)x+(\alpha-\beta).
\end{equation}
It can be easily verified that the ratio $J_{n+1}/J_n$ has strictly positive real part in the closed right half-plane when the numbers $x_i$ have positive real parts. That is, if all the zeros of $J_n$ belong to the left half-plane, then all the zeros of $J_{n+1}$ belong to the left half-plane as well. The proof of the third claim of the proposition now follows from the principle of mathematical induction.

Finally, let $\alpha=\beta$. In this case $J_{m+1}$ in \eqref{Jm1} is a linear function vanishing at the origin. Furthermore, \eqref{Jn1Jn} implies that the ratio $J_{n+1}/J_n$ has positive real part in the right half-plane and negative real part in the left half-plane when $n+m+2+2\alpha\geq0$ and $x_i\in\mathrm{i}\R$. That is, if the zeros of $J_n$ are on the imaginary axis so are the zeros of $J_{n+1}$. This finishes the proof of the proposition.
\end{proof}

\subsection{Proofs of Theorem~\ref{thm:asymptotic} and Proposition~\ref{prop:asymptoticat1}}

For reasons of brevity, below we shall often employ the following notation:
\begin{equation}
\label{ckalpha}
c_x(\alpha) := \frac{\Gamma(x+1+\alpha)}{\Gamma(\alpha+1)\Gamma(x+1)}.
\end{equation}
Using this notation we can write $P_n^{\alpha,\beta}(z)=\sum_{k=0}^nc_k(\alpha)c_{n-k}(\beta)z^k$. Since
\[
\Gamma(x)=\sqrt{2\pi/x}(x/e)^x\big(1+O(1/x)\big) \quad  \mbox{as} \quad x\to\infty,
\]
it holds that
\begin{equation}
\label{ckalphaas}
\Gamma(\alpha+1)c_x(\alpha) = (x+1)^\alpha\big(1+O(1/x)\big) \quad  \mbox{as} \quad x\to\infty.
\end{equation}
and respectively
\begin{equation}
\label{prodckalpha}
B(\alpha_1,\alpha_2)c_n(\alpha_1)c_n(\alpha_2) = \big(1+O(1/n)\big)c_n(\alpha_1+\alpha_2),
\end{equation}
where $B(\alpha_1,\alpha_2)$ is the beta function; see \eqref{beta}.

We also employ the notation $f(x)\sim g(x)$ as $x\to\infty$ which means that $g(x)\ll f(x)\ll g(x)$ where $f(x)\ll_cg(x)$ stands for $f(x)\leq A(c)g(x)$ and $A(c)$ is a constant depending only on~$c$.

The following simple lemma is needed for the proof of
Theorem~\ref{thm:asymptotic} and is an application of summation by
parts. 

\begin{lemma}
\label{lem:Dirichlet}
Suppose $\{a_k\}_{k=0}^\infty$ is a non-increasing sequence of positive
numbers and let $\big\{\{b_{k,m}\}_{k=0}^m\big\}_{m=0}^\infty$ be a
collection of non-decreasing, non-negative sequences. Then 
\[
\left|\sum_{k=0}^m a_kb_{k,m}z^k\right|\leq 4\frac{a_0b_{m,m}}{|1-z|}, \quad z\in\overline\D.
\]
\end{lemma}
\begin{proof}
Define $B_{k,m}:=\sum_{j=0}^k b_{j,m}z^j$. According to the summation by parts, it holds that
\[
B_{k,m} = b_{k,m}\frac{1-z^{k+1}}{1-z} + \sum_{j=0}^{k-1}\frac{1-z^{j+1}}{1-z}(b_{j,m}-b_{j+1,m})
\]
and consequently that
\[
| B_{k,m} | \leq \frac{2b_{k,m}}{|1-z|} + \frac2{|1-z|}\sum_{k=0}^{k-1}(b_{j+1,m}-b_{j,m})  \leq \frac{4b_{k,m}}{|1-z|} \leq \frac{4b_{m,m}}{|1-z|}
\]
for $z\in\overline\D$. Hence, applying summation by parts once more, we get that
\[
\left|\sum_{k=0}^m a_kb_{k,m}z^k\right| = a_m|B_{m,m}| + \sum_{k=0}^{m-1}|B_{k,m}|(a_k-a_{k+1}) \leq 4\frac{b_{m,m}a_0}{|1-z|}. \qedhere
\]
\end{proof}

\begin{proof}[Proof of Theorem~\ref{thm:asymptotic}]
Let $c_k(\alpha)$ be defined by \eqref{ckalpha}. The sequence $\{c_k(\alpha)\}_{k=0}^\infty$ is positive and increasing when $\alpha>0$. If $\alpha<0$, let $k_\alpha$ be the first integer such that $k_\alpha+\alpha>-1$. Then the numbers $c_k(\alpha)$ have the same sign for all $k\geq k_\alpha$ and the sequence $\{|c_k(\alpha)|\}_{k=k_\alpha}^\infty$ is decreasing. Observe also that
\[
(1-z)^{-(1+\alpha)}=\sum_{k=0}^\infty c_k(\alpha)z^k,
\]
where the series converges for all $z\in\overline\D\setminus\{1\}$ when $\alpha<0$ by the virtue of Lemma~\ref{lem:Dirichlet}. Further, put
\[
d_{k,n} := \frac{\beta}{|\beta|}\left[1-\frac{\Gamma(n+1)}{\Gamma(n+1+\beta)}\frac{\Gamma(n-k+1+\beta)}{\Gamma(n-k+1)}\right].
\]
Then the sequences $\{d_{k,n}\}_{k=0}^n$ are positive and increasing for all $n$ large enough, and bounded above by 1 when $\beta>0$. Moreover, it holds that $d_{m,n} \to 0$ as $n\to\infty$, where $m=m(n)$ is such that $m/n\to0$ as $n\to\infty$.

The left-hand side of \eqref{convergence1} can be estimated from above by
\begin{equation}
\label{brr}
\begin{array}{l}
\displaystyle \left|\sum_{k=0}^{k_\alpha-1} c_k(\alpha)d_{k,n}z^k\right| + \left|\sum_{k=k_\alpha}^m |c_k(\alpha)|d_{k,n}z^k\right| + |c_{n+1}(\alpha)||z|^{n+1}\left|\sum_{k=0}^\infty \bigg|\frac{c_{k+n+1}(\alpha)}{c_{n+1}(\alpha)}\bigg|z^k\right| \smallskip \\
\displaystyle \hspace{2in} + |c_{m+1}(\alpha)||z|^{m+1}\left|\sum_{k=0}^{n-m-1} \bigg|\frac{c_{k+m+1}(\alpha)}{c_{m+1}(\alpha)}\bigg|d_{k+m+1,n}z^k\right|,
\end{array}
\end{equation}
where $m=m(n)$ is such that $m\to\infty$ and $m/n\to0$ as $n\to\infty$. When $\alpha<0$ and $\beta>0$ the sum \eqref{brr} is bounded by
\begin{align} 
\label{bound}
& |d_{k_\alpha-1,n}|\max_{0 \leq k < k_\alpha}|c_k(\alpha)| \\
& \hspace{2cm} + 4|1-z|^{-1}\left(|c_{k_\alpha}(\alpha)|d_{m,n} +
  |c_{m+1}(\alpha)|d_{n,n}|z|^{m+1} +
  |c_{n+1}(\alpha)||z|^{n+1}\right)  \nonumber
\end{align}
according to Lemma~\ref{lem:Dirichlet} from which \eqref{convergence1}
clearly follows. When $\alpha<0$ and $\beta<0$, the bound
in \eqref{bound} still holds with the only difference being that
$d_{n,n}$ is no longer bounded by 1 but rather grows like
$n^{-\beta}$. Hence, if $m$ is chosen so that $m/\log n\to\infty$ as
$n\to\infty$, the term $d_{n,n}|z|^{m+1}$ converges to zero locally
uniformly in $\D$. When $\alpha>0$, the bound in \eqref{bound} is replaced by 
\[
4|1-z|^{-1}\left(c_m(\alpha)d_{m,n} + c_n(\alpha)d_{n,n}|z|^{m+1} + c_{n+1}(\alpha)|z|^{n+1}\right),
\]
again, due to Lemma~\ref{lem:Dirichlet}. Since $c_n(\alpha)$ grows
like $n^{\alpha}$ and $d_{n,n}$ grows no faster than $n^{|\beta|}$,
the second and the third terms in the parenthesis converge to zero
locally uniformly in $\D$. If, in addition, we require that
$m^{1+\alpha}/n\to0$ as $n\to\infty$, The first term converges to
zero.

Finally, \eqref{convergence2} follows immediately from \eqref{reciprocal} and \eqref{convergence1}.
\end{proof}

\begin{proof}[Proof of Proposition~\ref{prop:asymptoticat1}]
It holds that
\begin{equation}
\label{integralMab}
M_{\alpha,\beta}(z) = \int_C\frac{B_{\alpha,\beta}(t)}{B(1+\alpha,1+\beta)}e^{tz}\upd t
\end{equation}
by \eqref{beta} and the definition of $B_{\alpha,\beta}$, see Proposition~\ref{prop:algebraic}. Hence, using the notation from \eqref{beta} and \eqref{ckalpha}, we get that \eqref{Pnat1} follows from \eqref{integral} and the computation
\[
\begin{array}{lcl}
\displaystyle P_n^{\alpha,\beta}\left(1+\frac zn\right) &=&\displaystyle c_n(\gamma)\int_C\frac{B_{\alpha,\beta}(t)}{B(1+\alpha,1+\beta)}\left(1+\frac {zt}n\right)^n\upd t  \smallskip \\
&=& \displaystyle \big(1+o_n(1)\big)c_n(\gamma)\int_C\frac{B_{\alpha,\beta}(t)}{B(1+\alpha,1+\beta)}e^{zt}\upd t.
\end{array} 
\qedhere
\]
\end{proof}

\subsection{Proof of Theorem~\ref{thm:Kasymp}}

\begin{proof}[Proof of Theorem~\ref{thm:Kasymp}]
Set for brevity $K_N=K_N^{\alpha_1,\beta_1,\alpha_2,\beta_2}$. It follows from \eqref{prodckalpha} that
\[
K_N(0,0) = \sum_{n=0}^{N-1}c_n(\beta_1)c_n(\beta_2) = B^{-1}(\beta_1,\beta_2)\sum_{n=0}^{N-1}\big(1+\mathcal{O}(1/n)\big)c_n(\beta_1+\beta_2).
\]
Hence, the sequence $\{K_N(0,0)\}_N$ is divergent when $\beta_1+\beta_2+1\geq0$ by \eqref{ckalphaas}. If $\beta_1+\beta_2+1<0$, this sequence is eventually increasing or decreasing, depending on the sign of $\Gamma(1+\beta_1)\Gamma(1+\beta_2)$, and converges to
\begin{equation}
\label{sumcs}
\sum_{n=0}^{\infty}c_n(\beta_1)c_n(\beta_2) = {}_2F_1(1+\beta_1,1+\beta_2;1;1) = \frac{\Gamma(-1-\beta_1-\beta_2)}{\Gamma(-\beta_1)\Gamma(-\beta_2)},
\end{equation}
where the second equality follows from \cite[Eq. 15.4.20]{Olver}. This proves \eqref{KN00}.

In order to establish \eqref{SumInside1}, we first deduce from \eqref{convergence1} that
\begin{eqnarray}
K_N(z,w) &=& \sum_{n=0}^{N-1}c_n(\beta_1)c_n(\beta_2)\left(\frac1{(1-z)^{1+\alpha_1}}- o_n(1)\right)\left(\frac1{(1-w)^{1+\alpha_2}}- o_n(1)\right) \nonumber \\
&=& \frac{K_N(0,0)}{(1-z)^{1+\alpha_1}(1-w)^{1+\alpha_2}} + \sum_{n=0}^{N-1}o_n(1)c_n(\beta_1)c_n(\beta_2), \nonumber
\end{eqnarray}
where the functions $o_n(1)$ hold locally uniformly in either in $\D$ or $\D\times\D$. If $\beta_1+\beta_2+1\geq0$, the sequence $\{K_N(0,0)\}$ diverges to either $\infty$ or $-\infty$ and therefore
\[
\sum_{n=0}^{N-1}o_n(1)c_n(\beta_1)c_n(\beta_2) =  o_N(1)K_N(0,0),
\]
which shows the validity of \eqref{SumInside1}. If $\beta_1+\beta_2+1<0$, the limit $\lim_{N\to\infty} K_N(0,0)$ is finite and therefore
\[
\left|\sum_{n=0}^{N-1}o_n(1)c_n(\beta_1)c_n(\beta_2)\right| \ll 1
\]
locally uniformly in $\D\times\D$. That is, the family $\{K_N(z,w)\}_N$ is normal. In this case it is sufficient to examine the behavior of the Fourier coefficients of $K_N(z,w)$. It holds that
\[
K_N(z,w) = \sum_{j,k=0}^{N-1}c_j(\alpha_1)c_k(\alpha_2)\left[\sum_{n=\max\{j,k\}}^{N-1}c_{n-j}(\beta_1)c_{n-k}(\beta_2)\right]z^jw^k.
\]
Using \cite[Eq. 15.4.20]{Olver} as in \eqref{sumcs}, one can compute that the limit of the term in square brackets is equal to
\[
\frac{\Gamma(-\beta_1-\beta_2-1)}{\Gamma(1+\beta_1)\Gamma(1+\beta_2)}\left\{
\begin{array}{ll}
\displaystyle\frac{\Gamma(1+\beta_1)\Gamma(1+\beta_2+j-k)}{\Gamma(-\beta_2)\Gamma(j-k-\beta_1)}, & j-k\geq 0, \smallskip \\
\displaystyle\frac{\Gamma(1+\beta_1+k-j)\Gamma(1+\beta_2)}{\Gamma(k-j-\beta_2)\Gamma(-\beta_1)}, & j-k<0,
\end{array}
\right.
\]
which is exactly the $(j-k)$-th Fourier coefficient of $\Lambda_{\beta_1,\beta_2}$, see \eqref{Lambda12}. As $\Lambda_{\beta_1,\beta_2}$ is integrable on $\T$, we get from Fubini-Tonelli's theorem that
\begin{eqnarray}
 \sum_{j,k\geq0}c_j(\alpha_1)c_k(\alpha_2)\Lambda_{j-k}z^jw^k &=&  \sum_{j,k\geq0}c_j(\alpha_1)c_k(\alpha_2)z^jw^k\frac{1}{2\pi}\int_\T\overline\zeta^{j-k}\Lambda(\zeta)|\upd\zeta| \nonumber \\
 &=& \frac{1}{2\pi}\int_\T \frac{\Lambda_{\beta_1,\beta_2}(\zeta)}{\big(1-z\overline\zeta\big)^{1+\alpha_1}\big(1-w\zeta\big)^{1+\alpha_2}}|\upd\zeta| \nonumber
\end{eqnarray}
which finishes the proof of \eqref{SumInside2}.

To prove \eqref{SumOutside}, write
\begin{eqnarray}
\frac{K_N(z,w)}{(zw)^N} &=& \frac1{zw}\sum_{n=0}^{N-1}\frac{c_n(\alpha_1)c_n(\alpha_2)}{(zw)^{N-1-n}} \left(\frac1{(1-1/z)^{1+\beta_1}}- o_n(1)\right)\left(\frac1{(1-1/w)^{1+\beta_2}} - o_n(1)\right)\nonumber \\
&=&\frac1{zw}\frac1{(1-1/z)^{1+\beta_1}}\frac1{(1-1/w)^{1+\beta_2}}\sum_{n=0}^{N-1}\big(1+o_{N-1-n}(1)\big)\frac{c_{N-1-n}(\alpha_1)c_{N-1-n}(\alpha_2)}{(zw)^n}, \nonumber
\end{eqnarray}
where we used \eqref{convergence2} and the estimates $o_n(1)$ hold uniformly on closed subsets of either $\Om$ or $\Om\times\Om$. Observe that
\[
c_{N-1}(\alpha_1)c_{N-1}(\alpha_2) = \left(1+O\big(N^{-1}\big)\right)\frac{N^{\alpha_1+\alpha_2}}{\Gamma(1+\alpha_1)\Gamma(1+\alpha_2)}
\]
as $N\to\infty$ by \eqref{ckalphaas} and that
\[
\left|\frac{c_m(\alpha_1)c_m(\alpha_2)}{c_{N-1}(\alpha_1)c_{N-1}(\alpha_2)}\right| \ll N^{|\alpha_1+\alpha_2|}, \quad 0\leq m\leq N-1.
\]
This, in particular, implies that the family $\bigg\{\frac{K_N(z,w)}{(zw)^NN^{\alpha_1+\alpha_2}}\bigg\}_N$ is normal in $\Om\times\Om$. As before, this means that we only need to examine the asymptotic behavior of the Fourier coefficients. As
\[
\lim_{N\to\infty}\big(1+o_{N-1-n}(1)\big)\frac{c_{N-1-n}(\alpha_1)c_{N-1-n}(\alpha_2)}{c_{N-1}(\alpha_1)c_{N-1}(\alpha_2)} = 1
\]
for each fixed $n$, the proof of \eqref{SumOutside} follows.
\end{proof}

\subsection{Proofs of Theorems~\ref{thm:scaling}~\&~\ref{thm:scaling1}}

For the proof of Theorem~\ref{thm:scaling} we shall need the following fact.

\begin{lemma}
\label{lem:identity}
For $\gamma>-1$, it holds that
\begin{equation}
\label{identity}
\frac{\Gamma(N)}{\Gamma(N+1+\gamma)}\sum_{n=0}^{N-1}\frac{\Gamma(n+1+\gamma)}{\Gamma(n+1)}(1+\eta)^n = \int_0^1x^\gamma(1+\eta x)^{N-1}\mathrm{d}x.
\end{equation}
\end{lemma}
\begin{proof}
It can be readily verified by the principle of mathematical induction that
\begin{equation}
\label{induction}
\sum_{j=0}^J\frac{\Gamma(j+1+x)}{\Gamma(j+1)} = \frac1{1+x}\frac{\Gamma(J+2+x)}{\Gamma(J+1)},
\end{equation}
which, upon setting $J=N-k-1$, $j=n-k$, and $x=k+\gamma$, is the same as
\[
\sum_{n=k}^{N-1}\frac{\Gamma(n+1+\gamma)}{\Gamma(n+1-k)} = \frac1{k+\gamma+1}\frac{\Gamma(N+1+\gamma)}{\Gamma(N-k)}.
\]
Hence, it holds that
\begin{eqnarray}
\sum_{n=0}^{N-1}\frac{\Gamma(n+1+\gamma)}{\Gamma(n+1)}(1+\eta)^n &=& \sum_{k=0}^{N-1}\left[\sum_{n=k}^{N-1}\frac{\Gamma(n+1+\gamma)}{\Gamma(n+1)}\binom nk\right]\eta^k\nonumber \\
&=& \sum_{k=0}^{N-1} \frac{1}{k+\gamma+1}\frac{\Gamma(N+1+\gamma)}{\Gamma(k+1)\Gamma(N-k)}\eta^k \nonumber \\
&=& \frac{\Gamma(N+1+\gamma)}{\Gamma(N)}\int_0^1\sum_{k=0}^{N-1} \binom{N-1}k x^\gamma(x\eta)^k\mathrm{d}x,
\end{eqnarray}
which finishes the proof of the lemma.
\end{proof}

Theorem~\ref{thm:scaling} is an easy corollary to the following more
general claim. 

\begin{lemma}
\label{lem:kernel}
Let $\alpha_1,\alpha_2>-1/2$ and $\left\{P_n^{\alpha_j}\right\}_{n\in\N}$ be two sequences of polynomials such that
\begin{equation}
\label{exterior}
P_n^{\alpha_j}(z) = \left[F_j(z)+o_n(1)\right]\frac{\Gamma(n+1+\alpha_j)}{\Gamma(n+1)}z^n, \quad z\in\overline\Om,
\end{equation}
for all $n\in\N$, where the functions $F_j$ are holomorphic in $\Om$,
 continuous in $\overline\Om$, and $o_n(1)$ holds uniformly in
$\overline\Om$. Then for any $\zeta\in\T$ it holds that 
\[
K_N^{\alpha_1,\alpha_2}(\zeta,\zeta) = \left[\frac{F_1(\zeta)\overline{F_2(\zeta)}}{1+\alpha_1+\alpha_2}+o_N(1)\right]\frac{\Gamma(N+1+\alpha_1+\alpha_2)}{\Gamma(N)},
\]
where $\displaystyle K_N^{\alpha_1,\alpha_2}(z,w):=\sum_{n=0}^{N-1}P_n^{\alpha_1}(z)\overline{P_n^{\alpha_2}(w)}$. Moreover, if $F_1(\zeta)F_2(\zeta)\neq0$, then
\[
\lim_{N\to\infty} \frac{K_N^{\alpha_1,\alpha_2}\left(\zeta+\frac{a_1}N,\zeta+\frac{a_2}N\right)}{K_N^{\alpha_1,\alpha_2}(\zeta,\zeta)} = E_{\alpha_1+\alpha_2}\left(a_1\overline\zeta+\overline a_2\zeta\right),
\]
uniformly for $a_1,a_2$ on compact subsets of $\C$.
\end{lemma}
\begin{proof}
Let us first show that the functions $K_N^{\alpha_1,\alpha_2}\left(\zeta+\frac{a_1}N,\zeta+\frac{a_2}N\right)/K_N^{\alpha_1,\alpha_2}(\zeta,\zeta)$ form a normal family with respect to $a_1,\overline a_2$ whenever the latter belong to a bounded set. To this end, observe that
\begin{equation}
\label{CSch}
\big| K_N^{\alpha_1,\alpha_2}(z,w)\big|^2 \leq \big| K_N^{\alpha_1}(z,z)\big| \big|K_N^{\alpha_2}(w,w)\big|
\end{equation}
by the Cauchy-Schwartz inequality, where $K_N^{\alpha_j}(z,w):=\sum_{n=0}^{N-1}P_n^{\alpha_j}(z)\overline{P_n^{\alpha_j}(w)}$. It follows immediately from their definition that the functions $ K_N^{\alpha_j}(z,z)$ are subharmonic in $\C$. Therefore
\[
\left|K_N^{\alpha_j}(z,z)\right| \leq \max_{\eta\in\T}K_N^{\alpha_j}(\eta,\eta), \quad z\in\overline\D,
\]
by the maximum principle. Furthermore, as the functions $K_N^{\alpha_j}(z,z)/|z|^{2N-2}$ are subharmonic in $\overline\C\setminus\D$, it holds that
\[
\left| K_N^{\alpha_j}(z,z)\right| \leq |z|^{2N-2}\max_{\eta\in\T} K_N^{\alpha_j}(\eta,\eta), \quad z\in\overline\Om.
\]
Hence, for any constant $c>0$ it is true that
\begin{equation}
\label{normality-bound}
\left|K_N^{\alpha_j}(z,z)\right| \ll_c \max_{\eta\in\T}K_N^{\alpha_j}(\eta,\eta), \quad |z|\leq 1+\frac cN.
\end{equation}

Bounds \eqref{CSch} and \eqref{normality-bound} are already sufficient for establishing normality, but we still have to show that the claimed normalization constant is proportional to the one coming from \eqref{CSch} and \eqref{normality-bound}. To accomplish this goal, 
observe that by \eqref{prodckalpha}
\[
\frac{\Gamma(n+1+x_1)}{\Gamma(n+1)}\frac{\Gamma(n+1+x_2)}{\Gamma(n+1)} = \big(1+o_n(1)\big)\frac{\Gamma(n+1+x_1+x_2)}{\Gamma(n+1)}
\]
and therefore
\begin{equation}
\label{ateta1}
K_N^{\alpha_j}(\zeta,\zeta) = \sum_{n=0}^{N-1}\big|F_j(\zeta)+o_n(1)\big|^2\frac{\Gamma(n+1+2\alpha_j)}{\Gamma(n+1)} = \left(\frac{|F_j(\zeta)|^2}{1+2\alpha_j}+o_N(1)\right)\frac{\Gamma(N+1+2\alpha_j)}{\Gamma(N)}
\end{equation}
for $\zeta\in\T$ by the conditions of the lemma and \eqref{induction}. Analogously, it holds that
\begin{eqnarray}
K_N^{\alpha_1,\alpha_2}(\zeta,\zeta) &=& \sum_{n=0}^{N-1}\left[F_1(\zeta)\overline{F_2(\zeta)}+o_n(1)\right]\frac{\Gamma(n+1+\alpha_1+\alpha_2)}{\Gamma(n+1)}\nonumber \\
\label{ateta2}
&=& \left[\frac{F_1(\zeta)\overline{F_2(\zeta)}}{1+\alpha_1+\alpha_2}+o_N(1)\right]\frac{\Gamma(N+1+\alpha_1+\alpha_2)}{\Gamma(N)}
\end{eqnarray}
for $\zeta\in\T$. Formulas \eqref{ateta1} and \eqref{ateta2} show that
$\left|K_N^{\alpha_1,\alpha_2}(\zeta,\zeta)\right|^2$ and
$\left|K_N^{\alpha_1}(\zeta,\zeta) K_N^{\alpha_2}(\zeta,\zeta)\right|$
are of the same order of magnitude when $F_1(\zeta)F_2(\zeta)\neq0$,
since as $N \rightarrow \infty$, 
\begin{equation}
\label{comparealphas}
\frac{\Gamma(N+1+\alpha_1+\alpha_2)^2}{\Gamma(N+1+2\alpha_1)\Gamma(N+1+2\alpha_2)} = 1+ o_N(1)
\end{equation}
Combining \eqref{CSch} and \eqref{normality-bound} with the last
observation, we deduce that
\[
\left|K_N^{\alpha_1,\alpha_2}(\zeta,\zeta)\right|^{-1}\left| K_N^{\alpha_1,\alpha_2}\left(\zeta+\frac{a_1}N,\zeta+\frac{a_2}N\right)\right| \ll_c 1
\]
for all $|a_1|,|a_2|<c$ whenever $F_1(\zeta)F_2(\zeta)\neq0$ as claimed.

Given normality, it is enough to establish convergence in some
subregion of $|a_1|,|a_2|<c$. Hence, in what follows we can assume
that $\zeta+\frac{a_j}{N}\in\Om$ for all $N$ large. Then we deduce
from the asymptotic formulae for $P_n^{\alpha_j}$ that
$K_N^{\alpha_1,\alpha_2}\left(\zeta+\frac{a_1}N,\zeta+\frac{a_2}N\right)$
is equal to 
\[
\sum_{n=0}^{N-1}\left[F_1(\zeta)\overline{F_2(\zeta)}+o_n(1)\right]\frac{\Gamma(n+1+\alpha_1+\alpha_2)}{\Gamma(n+1)}\left(1+\frac{a_1\overline\zeta+\overline a_2\zeta} N + \frac{a_1\overline a_2}{N^2}\right)^n
\]
and therefore to
\begin{equation}
\label{ineedthisone}
\left[F_1(\zeta)\overline{F_2(\zeta)}+o_N(1)\right]\sum_{n=0}^{N-1}\frac{\Gamma(n+1+\alpha_1+\alpha_2)}{\Gamma(n+1)}\left(1+\frac{a_1\overline\zeta+\overline a_2\zeta} N + \frac{a_1\overline a_2}{N^2}\right)^n
\end{equation}
since $1+\alpha_1+\alpha_2>0$. The latter expression can be rewritten as
\[
\left[F_1(\zeta)\overline{F_2(\zeta)}+o_N(1)\right]\frac{\Gamma(N+1+\alpha_1+\alpha_2)}{\Gamma(N)}\int_0^1x^{\alpha_1+\alpha_2}\left(1+x\left[\frac{a_1\overline\zeta+\overline a_2\zeta} N + \frac{a_1\overline a_2}{N^2}\right]\right)^{N-1}\mathrm{d}x
\]
by Lemma~\ref{identity}. Since
\[
\int_0^1x^{\alpha_1+\alpha_2}\left(1+x\left[\frac{a_1\overline\zeta+\overline a_2\zeta} N + \frac{a_1\overline a_2}{N^2}\right]\right)^{N-1}\mathrm{d}x =  \frac{1+o_N(1)}{1+\alpha_1+\alpha_2}E_{\alpha_1+\alpha_2}\big(a_1\overline\zeta+\overline a_2\zeta\big),
\]
the lemma follows.
\end{proof}

For the proof of the main results, we shall need the following lemma.
\begin{lemma}
\label{lem:weird}
In the setting of Lemma~\ref{lem:kernel}, it holds that
\[
\lim_{N\to\infty}\frac1{N^{1+\alpha_2+\alpha_2}}K_N^{\alpha_1,\alpha_2}\left(\zeta+\frac{a_1}N,\overline\zeta+\frac{\overline a_2}N\right) = 0, \quad \zeta\in\T\setminus\{\pm1\},
\]
uniformly for $a_1,a_2$ on compact subsets of $\C$.
\end{lemma}
\begin{proof}
As in the previous lemma, we get from \eqref{CSch}, \eqref{normality-bound}, \eqref{ateta1}, and \eqref{comparealphas} that the family $\left\{N^{-1-\alpha_1-\alpha_2}K_N^{\alpha_1,\alpha_2}\left(\zeta+\frac{a_1}N,\overline\zeta+\frac{\overline a_2}N\right)\right\}$ is normal with respect to $a_1,a_2$ on compact subsets of $\C$. Hence, we can assume that $\zeta+\frac{a_j}{N}\in\Om$ for all $N$ large, $|a_1|,|a_2|<c$. In the present case, \eqref{ineedthisone} is replaced by
\[
\big[F_1(\zeta)F_2(\zeta)+o_N(1)\big]\sum_{n=0}^{N-1}\frac{\Gamma(n+1+\alpha_1+\alpha_2)}{\Gamma(n+1)}\left(\zeta^2+\frac{a_1+a_2}N\zeta + \frac{a_1a_2}{N^2}\right)^n.
\]
Then exactly as in Lemma~\ref{lem:Dirichlet} we get that 
\begin{align*}
& \left|\sum_{n=0}^{N-1}\frac{\Gamma(n+1+\alpha_1+\alpha_2)}{\Gamma(n+1)}\left(\zeta^2+\frac{a_1+a_2}N\zeta
    + \frac{a_1a_2}{N^2}\right)^n\right| \\ & \hspace{4cm} \leq
\frac4{|1-\zeta^2-\frac{a_1+a_2}N\zeta -
  \frac{a_1a_2}{N^2}|}\frac{\Gamma(N+\alpha_1+\alpha_2)}{\Gamma(N)} 
\end{align*}
As the numbers $|1-\zeta^2-\frac{a_1+a_2}N\zeta - \frac{a_1a_2}{N^2}|$ are bounded away from 0, the claim follows.
\end{proof}

\begin{proof}[Proof of Theorem~\ref{thm:scaling}]
Since $P_n^{\beta,\alpha}(0)=c_n(\alpha)$, see \eqref{ckalpha}, it follows from \eqref{convergence2} that
\[
(z-1)P_n^{\alpha_j,\beta_j}(z) = \left[(1-1/z)^{-\beta_j} + o_n(1)\right]c_n(\alpha_j)z^{n+1}
\]
uniformly in $\overline\Om$. Hence, the theorem is deduced from Lemma~\ref{lem:kernel} as
\[
(z-1)(w-1) K_N^{\alpha_1,\beta_1\alpha_2,\beta_2}(z,w)=K_{N+1}^{\alpha_1,\alpha_2}(z,\overline w) - 1
\]
where we set $P_0^{\alpha_j}(z)\equiv1$ and $P_{n+1}^{\alpha_j}(z):=(z-1)P_n^{\alpha_j,\beta_j}(z)$ for $n\geq1$.
\end{proof}

\begin{proof}[Proof of Theorem~\ref{thm:scaling1}]
Set, for brevity, $\gamma_j:=1+\alpha_j+\beta_j$. It follows from
\eqref{Pnat1}, \eqref{prodckalpha}, and \eqref{induction} that, since
$\gamma > -1$,  
\begin{align*}
K_N^{\alpha_1,\beta_1,\alpha_2,\beta_2}(1,1) &=
\sum_{n=0}^{N-1}\frac{1+o_n(1)}{\Gamma(1+\gamma_1)\Gamma(1+\gamma_2)}\frac{\Gamma(n+1+\gamma)}{\Gamma(n+1)} \\
&=
\frac{1+o_N(1)}{\Gamma(1+\gamma_1)\Gamma(1+\gamma_2)}\frac1{1+\gamma}\frac{\Gamma(N+1+\gamma)}{\Gamma(N)}.
\end{align*}
Using \eqref{integral}, we get that $K_N^{\alpha_1,\beta_1,\alpha_2,\beta_2}\left(1+\frac {a_1}N,1+\frac {a_2}N\right)$ is equal to
\[
\int_C\int_C \frac{B_{\alpha_1,\beta_1}(t)B_{\alpha_2,\beta_2}(u)}{\Gamma(1+\alpha_1)\Gamma(1+\beta_1)\Gamma(1+\alpha_2)\Gamma(1+\beta_2)} \sum_{n=0}^{N-1}\frac{\Gamma(n+1+\gamma_1)}{\Gamma(n+1)}\frac{\Gamma(n+1+\gamma_2)}{\Gamma(n+1)}v_N^n\mathrm{d}t\mathrm{d}u,
\]
where $v_N:=\displaystyle 1+\frac {a_1t+a_2u}N +\frac{a_1a_2tu}{N^2}$. As in the proof of Lemma~\ref{lem:kernel}, it holds that
\[
\sum_{n=0}^{N-1}\frac{\Gamma(n+1+\gamma_1)}{\Gamma(n+1)}\frac{\Gamma(n+1+\gamma_2)}{\Gamma(n+1)}v_N^n = \frac{1+o_N(1)}{1+\gamma}\frac{\Gamma(N+1+\gamma)}{\Gamma(N)}E_\gamma\big(a_1t+a_2u\big),
\]
from which we deduce that the left-hand side of \eqref{scaling1} is equal to
\[
\int_C\int_C\frac{B_{\alpha_1,\beta_1}(t)}{B(1+\alpha_1,1+\beta_1)}\frac{B_{\alpha_2,\beta_2}(u)}{B(1+\alpha_2,1+\beta_2)}
E_\gamma(a_1t+a_2u)\mathrm{d}t\mathrm{d}u 
\]
uniformly for $a,b$ on compact sets. By the very definition of $E_\gamma$, we have that
\begin{eqnarray}
\int_C\frac{B_{\alpha_1,\beta_1}(t)}{B(1+\alpha_1,1+\beta_1)}E_\gamma(a_1t+a_2u)\mathrm{d}t &=& \int_C\frac{B_{\alpha_1,\beta_1}(t)}{B(1+\alpha_1,1+\beta_1)}(\gamma+1)\int_0^1x^\gamma e^{(a_1t+a_2u)x}\mathrm{d}x\mathrm{d}t \nonumber \\
&=& (\gamma+1)\int_0^1x^\gamma e^{a_2ux}\int_C\frac{B_{\alpha_1,\beta_1}(t)}{B(1+\alpha_1,1+\beta_1)}e^{a_1xt}\mathrm{d}t\mathrm{d}x \nonumber \\
&=& (\gamma+1)\int_0^1x^\gamma e^{a_1ux}M_{\alpha_1,\beta_1}(a_2x)\mathrm{d}x, \nonumber
\end{eqnarray}
where the last equality follows from \eqref{integralMab}. Thus, the
left-hand side of \eqref{scaling1} is equal to 
\begin{align*}
& (\gamma+1)\int_0^1x^\gamma M_{\alpha_1,\beta_1}(a_1x)
\int_C\frac{B_{\alpha_2,\beta_2}(u)}{B(1+\alpha_2,1+\beta_2)}
e^{a_2ux} \mathrm{d}u \mathrm{d}x \\ & \hspace{4cm} =
(\gamma+1)\int_0^1x^\gamma
M_{\alpha_1,\beta_1}(a_1x)M_{\alpha_2,\beta_2}(a_2x)\mathrm{d}x,  
\end{align*}
which finishes the proof of the theorem.
\end{proof}

\subsection{Proofs of Theorems~\ref{thm:Scaling}---\ref{thm:Scaling2}}

We shall need the following lemma for the proof of Theorem~\ref{thm:Scaling}.
\begin{lemma}
\label{lem:sums}
Let $\zeta\in\T\setminus\{1\}$, $\alpha>0$ and $\beta<0$. Then for
any $\varepsilon>0$ it holds that 
\[
\lim_{J\to\infty}J^{-\varepsilon-\alpha}\sum_{j=0}^{J-1}\frac{s_{2j}}{s_{2J}}P_j^{\alpha,\beta}\left(\zeta+\frac zJ\right) = 0
\]
locally uniformly in $\C$.
\end{lemma}
\begin{proof}

Assume first that $\zeta+z/J\in\overline\Om$. As $P_j^{\beta,\alpha}(0)=c_j(\alpha)$ in the notation \eqref{ckalpha}, we get from \eqref{convergence2} that
\[
\sum_{j=0}^{J-1}\frac{s_{2j}}{s_{2J}}P_j^{\alpha,\beta}\left(\zeta+\frac zJ\right) = \sum_{j=0}^{J-1}\left[\left(1-\overline\zeta\right)^{-1-\beta}+o_j(1)\right]\frac{s_{2j}}{s_{2J}}c_j(\alpha)\left(\zeta+\frac zJ\right)^j.
\]
Since $\alpha>0$, $c_j(\alpha)\to\infty$ and the numbers $s_{2j}/s_{2J}$ are increasing with $j$, it holds that
\[
\sum_{j=0}^{J-1}\frac{s_{2j}}{s_{2J}}P_j^{\alpha,\beta}\left(\zeta+\frac zJ\right) = \left[\left(1-\overline\zeta\right)^{-1-\beta}+o_J(1)\right]\sum_{j=0}^{J-1}\frac{s_{2j}}{s_{2J}}c_j(\alpha)\left(\zeta+\frac zJ\right)^j.
\]
Furthermore, we deduce from Lemma~\ref{lem:Dirichlet} that
\begin{equation}
\label{stupidlimit2}
\left|\sum_{j=0}^{J-1}\frac{s_{2j}}{s_{2J}}c_j(\alpha)\left(\eta+\frac zJ\right)^n\right| \leq \frac{4c_{J-1}(\alpha)}{|1-\zeta-\frac zJ|}.
\end{equation}
Finally, the bound in \eqref{stupidlimit2} can be extended to all
$\zeta+z/J$ by the maximum modulus principle and the normal family
argument. The claim now follows from \eqref{ckalphaas}. 
\end{proof}

\begin{proof}[Proof of Theorem~\ref{thm:Scaling}]
{\bf Case $N=2J$:} Using relations \eqref{pies}, \eqref{IDS1} and
definitions \eqref{pes}, \eqref{sums}, we get that
\begin{align}
\label{DSNK}
& \max\left\{1,|z|\right\}^s\max\left\{1,|w|\right\}^s \mathbf
K^{(1,1)}_{2J}\left(z,w\right) \\ 
& \hspace{3cm} =
\displaystyle\frac12\left(1+\frac1s\right)\left(wK_J^{(1)}\big(z^2,w^2\big)-zK_J^{(1)}\big(w^2,z^2\big)\right)
\nonumber \\
& \hspace{4cm} - \frac3{2s}
\left(wK_J^{(2)}\big(z^2,w^2\big)-zK_J^{(2)}\big(w^2,z^2\big)\right)
\nonumber 
\end{align}
where
\begin{equation}
\label{Ki}
K_J^{(i)}(z,w):=K_J^{1/2,-1/2,i-1/2,-3/2}(z,w), \quad i\in\{1,2\}.
\end{equation}
Given \eqref{omega}, we only need to compute the scaling limit of \eqref{DSNK}.

For brevity, set
\begin{equation}
\label{zN}
z_{\zeta,N} := \zeta + \frac zN, \quad z_{\zeta,N}^2 = \zeta^2+\frac{z\zeta}J\bigg(1+\frac{z\overline\zeta}{4J}\bigg).
\end{equation}
Then we deduce from Lemma~\ref{lem:weird} that
\[
\lim_{J\to\infty}J^{-1-i}K_J^{(i)}\left(z_{\zeta,N}^2,w_{\zeta,N}^2\right) = 0
\]
uniformly on compact subsets of $\C\times\C$ (as in the proof of Theorem~\ref{thm:scaling}, we need to multiply $K_J^{(i)}$ by $\big(1-z_{\zeta,N}^2\big)\big(1-w_{\zeta,N}^2\big)$ in order to apply Lemma~\ref{lem:weird}, but clearly this does not change the limit). Hence,
\[
\lim_{J\to\infty}N^{-2}\mathbf K^{(1,1)}_{2J}\left(z_{\zeta,N},w_{\zeta,N}\right) = \lim_{N\to\infty}N^{-2}\mathbf K^{(2,2)}_{2J}\left(z_{\zeta,N},w_{\zeta,N}\right) = 0
\]
uniformly on compact subsets of $\C\times\C$, where the limit for $\mathbf K^{(2,2)}_N$ follows from the relation
\begin{equation}
\label{IS-DS}
\mathbf K^{(2,2)}_{2J}\left(z_{\zeta,N},w_{\zeta,N}\right) = \iota\big(z_{\zeta,N}\big)\iota\big(w_{\zeta,N}\big)\mathbf K^{(1,1)}_{2J}\left(\overline z_{\zeta,N},\overline w_{\zeta,N}\right),
\end{equation}
see \eqref{IDS1} and \eqref{iota}.

On the other hand, we have that
\begin{equation}
\label{S-DS}
\mathbf K^{(1,2)}_{2J}\left(z_{\zeta,N},w_{\zeta,N}\right) = \iota\big(w_{\zeta,N}\big)\mathbf K^{(1,1)}_{2J}\left(z_{\zeta,N},\overline w_{\zeta,N}\right).
\end{equation}
Thus, we deduce from Theorem~\ref{thm:scaling} that
\[
\lim_{J\to\infty}N^{-1-i}K_J^{(i)}\left(z_{\zeta,N}^2,\overline w_{\zeta,N}^2\right) = \frac{1}{3^{i-1}(1+i)\pi}\frac{(1-\zeta^2)^{1/2}}{(1-\overline\zeta^2)^{1/2}}E_i\left(z\overline\zeta+\overline w\zeta\right)
\]
uniformly on compact subsets of $\C\times\C$. Recall that the function $(1-z)^{1/2}$ was defined as holomorphic in the unit disk with the branch cut along the positive reals greater than 1. Thus, $\Arg\left((1-e^{\mathrm{i}t})^{1/2}\right)=(t-\pi)/4$, $t\in[0,2\pi]$. Hence,
\[
\overline z(1-z^2)^{1/2}(1-\overline z^2)^{-1/2} = -\iota(z)
\]
and therefore
\[
\lim_{J\to\infty}N^{-2}\mathbf K^{(1,2)}_{2J}\left(z_{\zeta,N},\overline w_{\zeta,N}\right) = \omega\big(z\overline\zeta\big)\omega\big(w\overline\zeta\big)\left[\frac1{2\pi}E_1\left(z\overline\zeta+\overline w\zeta\right) - \frac{\lambda}{3\pi}E_2\left(z\overline\zeta+\overline w\zeta\right)\right]
\]
uniformly on compact subsets of $\C\times\C$, from which the claim of the theorem follows by the very definition of $E_i$, see \eqref{Egamma}.

{\bf Case $N=2J+1$:} Given \eqref{IDS2}, we only need to show that
\[
\lim_{J\to\infty}N^{-2}\pi_{2J}\left(z_{\zeta,N}\right)\sum_{j=0}^{J-1}\frac{s_{2j}}{s_{2J}}\pi_{2j+1}\left(w_{\eta,N}\right) =0
\]
uniformly on compact subsets of $\C$, where $\eta=\zeta$ or
$\eta=\overline\zeta$.  By \eqref{pies}, we wish to take the limit of
\begin{equation}
\label{pain}
\frac1{4J^2}P_{J}^{1/2,-1/2}\left(z_{\zeta,N}^2\right)
\sum_{j=0}^{J-1}\frac{s_{2j}}{s_{2J}}\frac
{w_{\eta,N}}4\left[\left(1+\frac1s\right)P_j^{1/2,-3/2}\left(w_{\eta,N}^2\right)-
  \frac 3sP_j^{3/2,-3/2}\left(w_{\eta,N}^2\right)\right],
\end{equation}
and the claim now follows from Corollary~\ref{cor:asymptotic} and
Lemma~\ref{lem:sums}. 
\end{proof}

For the proof of Theorem~\ref{thm:7} we shall need the following two lemmas.

\begin{lemma}
\label{lem:asymptoticat1} 
If $\gamma=1+\alpha+\beta>-1$, it holds that
\begin{equation}
\label{sumPnat1}
\lim_{J\to\infty}\frac{\Gamma(1+\gamma)}{J^{1+\gamma}}\sum_{j=0}^{J-1}\frac{s_{2j}}{s_{2J}}P_j^{\alpha,\beta}\left(1+\frac zJ\right) = \sqrt{1-\lambda}\int_0^1\frac{x^\gamma M_{\alpha,\beta}(zx)}{\sqrt{1-\lambda x}}\upd x
\end{equation}
locally uniformly in $\C$.
\end{lemma}
\begin{proof}
When $\lambda<1$, it is a straightforward computation using the asymptotic behavior of the Gamma function to verify 
\[
\frac{s_{2j}}{s_{2J}} = \frac{\Gamma\left(\frac{s-2j-1}2\right)}{\Gamma\left(\frac{s-2j}2\right)}\frac{\Gamma\left(\frac{s-2J}2\right)}{\Gamma\left(\frac{s-2J-1}2\right)} = \big(1+o_J(1)\big)\sqrt{\frac{s-2J}{s-2j}} = \big(1+o_J(1)\big)\sqrt{\frac{1-\lambda}{1-jJ^{-1}}},
\]
where we used Lemma~\ref{lemma:2}. Then it follows from  Proposition~\ref{prop:asymptoticat1} that
\[
\frac{\Gamma(1+\gamma)}{J^{1+\gamma}}\sum_{j=0}^{J-1}\frac{s_{2j}}{s_{2J}}P_j^{\alpha,\beta}\left(1+\frac zJ\right) = \sum_{j=0}^{J-1}\frac{1+o_j(1)}J\left(\frac{j+1}J\right)^\gamma\sqrt{\frac{1-\lambda}{1-jJ^{-1}}}M_{\alpha,\beta}\left(\frac{jz}J\right).
\]
As the right-hand side of the equation above is essentially a Riemann sum for the right-hand side of \eqref{sumPnat1}, the claim follows.

When $\lambda=1$, it holds that $s_{2j}/s_{2J}=o_J(1)$. Hence, replacing the square root by $o_J(1)$ in the last equation, we see that the left hand side of \eqref{sumPnat1} converges to zero.
\end{proof}

\begin{lemma}
\label{lem:miracle}
When $\lambda<1$, it holds that
\[
M_{1/2,-1/2}(z) = \frac1{\sqrt{1-\lambda}}\int_0^1\frac{M(xz)-\lambda xM^\prime(xz)}{\sqrt{1-\lambda x}}\upd x.
\]
\end{lemma}
\begin{proof}
Using the series representation for $M=M_{1/2,-3/2}$, see \eqref{Mab},
and $M^\prime$ we get that 
\begin{align*}
&\int_0^1\frac{M(xz)-\lambda xM^\prime(xz)}{\sqrt{1-\lambda x}}\upd x
\\
& \hspace{2cm} =
\frac1{\Gamma(3/2)}\sum_{n=0}^\infty\frac{\Gamma(n+3/2)}{\Gamma(n+1)}\frac{z^n}{n!}\int_0^1\left(x^n-\lambda\frac{n+3/2}{n+1}x^{n+1}\right)\frac{\upd
  x}{\sqrt{1-\lambda x}}. 
\end{align*}
Integration by parts yields that
\[
\int_0^1\frac{\lambda x^{n+1}}{\sqrt{1-\lambda x}}\upd x = -2\sqrt{1-\lambda} + 2(n+1)\int_0^1\frac{\lambda x^n}{\sqrt{1-\lambda x}}\upd x - 2(n+1)\int_0^1\frac{\lambda x^{n+1}}{\sqrt{1-\lambda x}}\upd x
\]
and therefore
\[
\int_0^1\frac{M(xz)-\lambda xM^\prime(xz)}{\sqrt{1-\lambda x}}\upd x = \frac{\sqrt{1-\lambda}}{\Gamma(3/2)}\sum_{n=0}^\infty\frac{\Gamma(n+3/2)}{\Gamma(n+1)}\frac{z^n}{n!}\frac1{n+1} = \sqrt{1-\lambda}M_{1/2,-1/2}(z)
\]
by \eqref{Mab}.
\end{proof}

\begin{proof}[Proof of Theorem~\ref{thm:7}]
{\bf Case $N=2J$:} Let $z_{\xi,N}$ and $w_{\xi,N}$ be defined by \eqref{zN}. As $\xi^2=1$, it follows from Theorem~\ref{thm:scaling1} that
\[
\lim_{J\to\infty}N^{-1-i}K_J^{(i)}\left(z_{\xi,N}^2,w_{\xi,N}^2\right) = \frac1{1+i}\frac1{2^{1+i}} E_{1/2,-1/2,i-1/2,-3/2}(z\xi,w\xi)
\]
uniformly on compact subsets of $\C\times\C$, where $K_J^{(i)}$ were defined in \eqref{Ki}. Hence, we deduce from \eqref{DSNK} that
\[
\begin{array}{l}
\displaystyle\lim_{J\to\infty}N^{-2}\mathbf
K^{(1,1)}_{2J}\left(z_{\xi,N},w_{\xi,N}\right) =
\omega\left(z\xi\right)\omega\left(w\xi\right) \bigskip \\ 
\hspace{.7in}\displaystyle\times\frac\xi{16}\bigg[\left(E^{(1)}(z\xi,w\xi)-E^{(1)}(w\xi,z\xi)\right)-\lambda\left(E^{(2)}(z\xi,w\xi)-E^{(2)}(w\xi,z\xi)\right)\bigg]
\end{array}
\]
where $E^{(i)}:=E_{1/2,-1/2,i-1/2,-3/2}$. By the very definition, see \eqref{secondE}, it holds that
\begin{equation}
\label{Ei}
\begin{array}{l}
E^{(i)}(a_1,a_2)-E^{(i)}(a_2,a_1) \bigskip \\
\displaystyle = (1+i)\int_0^1 x^i\left[M_{1/2,-1/2}(a_1x)M_{i-1/2,-3/2}(a_2x)-M_{1/2,-1/2}(a_2x)M_{i-1/2,-3/2}(a_1x)\right]\upd x.
\end{array}
\end{equation}
The following relations can be readily checked: 
\begin{equation}
\label{Ms}
M(x)=M_{1/2,-3/2}(x), \quad M_{3/2,-3/2}(x)=\frac23M^\prime(x); \quad M_{1/2,-1/2}(x)=2\big(M^\prime(x)-M(x)\big).
\end{equation}
Plugging these relations into \eqref{Ei}, we obtain the desired limit for $\mathbf K^{(1,1)}_{2J}$.

{\bf Case: $N=2J+1$:} In this case we need to deal with additional limits of the form \eqref{pain} where $\zeta=\eta=\xi$. It follows from Proposition~\ref{prop:asymptoticat1} and Lemma~\ref{lem:asymptoticat1} that the limit of \eqref{pain} is equal to
\[
\sqrt{1-\lambda}\frac\xi{16}M_{1/2,-1/2}(z\xi)\int_0^1\frac{M_{1/2,-3/2}(xw\xi)-\frac32\lambda xM_{3/2,-3/2}(xw\xi)}{\sqrt{1-\lambda x}}\upd x.
\]
Thus, this limit is zero when $\lambda=1$ and is equal to
\begin{equation}
\label{sym}
\sqrt{1-\lambda}\frac\xi{16}M_{1/2,-1/2}(z\xi)\frac\xi{16}M_{1/2,-1/2}(w\xi)
\end{equation}
when $\lambda<1$ by \eqref{Ms} and Lemma~\ref{lem:miracle}. Since \eqref{sym} is symmetric with respect to $z$ and $w$, the additional terms in \eqref{IDS2} cancel each other out.
\end{proof}

\begin{proof}[Proof of Lemma~\ref{lem:compare}]
Set $F(b):=\frac14|M(z)|^2$, where $z=a+\mathrm{i}b$. It follows from
\cite[Eq. 13.7.2]{Olver} that $\lim_{|b|\to\infty} F(b)/G(b) = 1$,
where $G(b) := |z|e^{2a}/\pi$. As both functions tend to infinity as
$|b|\to\infty$, L'Hopital's rule yields that $ \lim_{|b|\to\infty}
F^\prime(b)|z|/b = e^{2a}/\pi$. Since $F^\prime(b) =
\big[M^\prime\big(z\big)M^\prime\big(\overline
z\big)-M\big(z\big)M^\prime\big(\overline z\big)\big]/4$, the claim of
the lemma follows. 
\end{proof}

For the proof of Theorem~\ref{thm:Scaling2} we shall need the following lemma.

\begin{lemma}
\label{lem:EO}
Let $\xi=\pm1$. For $y\in\R$, it holds that
\begin{equation}
\label{EOpi}
\left\{
\begin{array}{lll}
(\epsilon\widetilde\pi_{2n})(y) &=& \displaystyle -\xi - \int_\xi^y\widetilde\pi_{2n}(u)\upd u, \smallskip \\
(\epsilon\widetilde\pi_{2n+1})(y) &=& \displaystyle \frac1{4s} - \int_\xi^y\widetilde\pi_{2n+1}(u)\upd u.
\end{array}
\right.
\end{equation}
\end{lemma}
\begin{proof}
It is a straightforward calculation to get that
\[
\epsilon\left(x^{2k}\max\big\{1,|x|\big\}^{-s}\right)(y) = -\int_\xi^yx^{2k}\max\big\{1,|x|\big\}^{-s}\upd x - \frac\xi{2k+1}.
\]
Hence, using the representation from Theorem~\ref{thm:2}, we get that
\begin{eqnarray}
(\epsilon\widetilde\pi_{2n})(y) &=& -\int_\xi^y\widetilde\pi_{2n}(x)\upd x - \frac\xi\pi\sum_{k=0}^n\frac{\Gamma(k+1/2)}{\Gamma(k+1)}\frac{\Gamma(n-k+1/2)}{\Gamma(n-k+1)}  \nonumber \\
&=& -\int_\xi^y\widetilde\pi_{2n}(x)\upd x - \xi P_n^{-1/2,-1/2}(1). \nonumber
\end{eqnarray}
As $P_n^{-1/2,-1/2}(1)=1$ by \eqref{Pnat1}, the first claim in \eqref{EOpi} follows. Analogously, it holds that
\[
\epsilon\left(x^{2k+1}\max\big\{1,|x|\big\}^{-s}\right)(y) = -\int_\xi^yx^{2k+1}\max\big\{1,|x|\big\}^{-s}\upd x + \frac{1}{s-2k-2}
\]
and therefore
\begin{eqnarray}
(\epsilon\widetilde\pi_{2n+1})(y) &=& -\int_\xi^y\widetilde\pi_{2n+1}(u)\upd u - \frac1{4\pi s}\sum_{k=0}^n \frac{ \Gamma(k + 3/2) \Gamma(n - k - 1/2)}{\G{k+1}\G{n -k + 1} } \nonumber \\
&=& -\int_\xi^y\widetilde\pi_{2n+1}(u)\upd u  + \frac1{4s}P_n^{1/2,-3/2}(1), \nonumber
\end{eqnarray}
which finishes the proof of \eqref{EOpi} as $P_n^{1/2,-3/2}(1)=1$ by \eqref{Pnat1}.
\end{proof}

For the proof of Theorem~\ref{thm:Scaling2} we shall need the following relation
\begin{equation}
\label{Mab+1}
M_{\alpha,\beta+1}(z) = (1+\gamma)\int_0^1x^\gamma
M_{\alpha,\beta}(zx)\upd x; \qquad \gamma = 1 + \alpha + \beta.
\end{equation}

\begin{proof}[Proof of Theorem~\ref{thm:Scaling2}]
To prove the theorem we need to show that the scaling limits
\eqref{everything} of the matrix kernel \eqref{eq:21} are equal to
\eqref{three-kernels} with $A=A_\xi$ defined by \eqref{auxkappa}. As
in the previous two theorems, we use notation and expressions of
\eqref{IDS1} \& \eqref{IDS2} for the entries of $\mathbf K_N$. Notice
also that the fourth case in \eqref{three-kernels} is simply a
restatement of Corollary~\ref{cor:Scaling1}. As usual we shall use
\eqref{pies} \& \eqref{pes} throughout the analysis.

{\bf Case $N=2J$:} Recall that $\mathbf K^{(1,1)}_{2J}$ does not
contain the $\epsilon$ operator and therefore remains the same for all
cases.  Further, we get from Lemma~\ref{lem:EO} that 
\begin{eqnarray}
\mathbf K^{(1,2)}_{2J}(z,y) &=& 2\sum_{j=0}^{J-1}\left[\widetilde\pi_{2j}(z)\left(\frac1{4s}-\int_\xi^y\widetilde\pi_{2j+1}(x)\upd x\right)+\widetilde\pi_{2j+1}(z)\left(\xi+\int_\xi^y\widetilde\pi_{2j}(x)\upd x\right)\right] \nonumber \\
&=& -\int_\xi^y \mathbf K^{(1,1)}_{2J}(z,x)\upd x + \frac1{2s}\sum_{j=0}^{J-1}\widetilde\pi_{2j}(z) + 2\xi\sum_{j=0}^{J-1}\widetilde\pi_{2j+1}(z). \nonumber
\end{eqnarray}
With the notation \eqref{zN}, observe that
\begin{equation}
\label{comb1}
\lim_{J\to\infty}\frac1N\int_\xi^{y_{\xi,N}} \mathbf K^{(1,1)}_{2J}\left(z_{\xi,N},v\right)\upd v = \int_0^y\lim_{J\to\infty}\frac1{N^2}\mathbf K^{(1,1)}_{2J}\left(z_{\xi,N},v_{\xi,N}\right)\upd v = \int_0^y \varkappa_\xi(z,v) \upd v
\end{equation}
locally uniformly in $\C\times\R$ by Theorem~\ref{thm:7}. It also follows from \eqref{sumPnat1} applied with $s=\infty$ (in which case $s_{2j}/s_{2J}=1$ and $\lambda=0$) that
\[
\lim_{J\to\infty}\frac2{J^2}\sum_{j=0}^{J-1}P_j^{1/2,-1/2}\left(z_{\xi,N}^2\right) = M_{1/2,1/2}(z\xi) = 4\int_0^1u\left(M^\prime(z\xi u)-M(z\xi u)\right)\upd u,
\]
where we used \eqref{Mab+1} and \eqref{Ms}. Hence, it holds that
\begin{equation}
\label{comb2}
\lim_{J\to\infty}\frac1{2sN}\sum_{j=0}^{J-1}\widetilde\pi_{2j}\left(z_{\xi,N}\right) =\frac{\omega(z\xi)}4\int_0^1\lambda u\left(M^\prime(z\xi u)-M(z\xi u)\right)\upd u
\end{equation}
locally uniformly in $\C$. Once more, we deduce from \eqref{sumPnat1}
and \eqref{Mab+1} that, for $i \in \{0,1\}$, 
\[
\lim_{J\to\infty}\frac{2^i}{J^{1+i}}\sum_{j=0}^{J-1}P_j^{i+1/2,-3/2}\left(z_{\xi,N}^2\right) = M_{i+1/2,-1/2}(z\xi) = (1+i)\int_0^1 u^i M_{i+1/2,-3/2}(z\xi u)\upd u.
\]
Thus, we get that
\begin{equation}
\label{comb3}
\lim_{J\to\infty}\frac{2\xi}N\sum_{j=0}^{J-1}\widetilde\pi_{2j+1}\left(z_{\xi,N}\right) = \frac{\omega(z\xi)}4\int_0^1 \left(M(z\xi u)-\lambda uM^\prime(z\xi u)\right) \upd u
\end{equation}
locally uniformly in $\C$, where we again used \eqref{Ms}.  Combining \eqref{comb1}---\eqref{comb3}, we get that
\begin{eqnarray}
\lim_{J\to\infty} \frac1N\mathbf K^{(1,2)}_{2J}\big(z_{\xi,N},y_{\xi,N}\big) &=& -\int_0^y\varkappa_\xi(z,v)\upd v + \frac{\omega(z\xi)}4\int_0^1(1-\lambda u)M(z\xi u)\upd u \nonumber \\
\label{comb4} 
&=& -DA_\xi(z,y).
\end{eqnarray}
This finishes the proof of the second and the third cases in
\eqref{three-kernels} since
\[
\left\{
\begin{array}{lll}
\displaystyle \mathbf K^{(2,1)}_{2J}\left(z_{\xi,N},y_{\xi,N}\right) &=& \iota(z_{\xi,N})\mathbf K^{(1,1)}_{2J}\left(\overline z_{\xi,N},y_{\xi,N}\right)\smallskip \\
\displaystyle \mathbf K^{(2,2)}_{2J}\left(z_{\xi,N},y_{\xi,N}\right) &=& \iota(z_{\xi,N})\mathbf K^{(1,2)}_{2J}\left(\overline z_{\xi,N},y_{\xi,N}\right)
\end{array}.
\right.
\]

To prove the first equality in \eqref{three-kernels}, notice that
\eqref{comb4} provides us with the terms on the anti-diagonal of $\mathbf K_{2J}(x,y)$. Thus, we only need to compute the limit of the $\mathbf K^{(2,2)}_{2J}\big(x_{\xi,N},y_{\xi,N}\big)$ (observe the presence of $\frac12\sgn(y_{\xi,N}-x_{\xi,N})=\frac12\sgn(y-x)$ in \eqref{eq:21}). To this end,we get from Lemma~\ref{lem:EO} that
\begin{eqnarray}
\mathbf K^{(2,2)}_{2J}(x,y) &=& 2\sum_{j=0}^{J-1}\left[ - \left(\xi+\int_\xi^x\widetilde\pi_{2j}(v)\upd v\right)\left(\frac1{4s} - \int_\xi^y\widetilde\pi_{2j+1}(v)\upd v\right) \right. \nonumber \\
&& \hspace{1.35in} \left. + \left(\frac1{4s} - \int_\xi^x\widetilde\pi_{2j+1}(v)\upd v\right)\left(\xi+\int_\xi^y\widetilde\pi_{2j}(v)\upd v\right)\right] \nonumber \\
&=& \int_\xi^x\int_\xi^y \mathbf K^{(1,1)}_{2J}(u,v)\upd v\upd u +\left(\int_\xi^y-\int_\xi^x\right)\sum_{j=0}^{J-1}\left(\frac1{2s}\widetilde\pi_{2j}(v)+2\xi\widetilde\pi_{2j+1}(v)\right)\upd v. \nonumber
\end{eqnarray}
As before, we get that
\[
\int_\xi^{x_{\xi,N}}\int_\xi^{y_{\xi,N}} \mathbf K^{(1,1)}_{2J}(u,v)\upd v\upd u = N^{-2}\int_0^x\int_0^y \mathbf K^{(1,1)}_{2J}(u_{\xi,N},v_{\xi,N})\upd v\upd u
\]
and therefore this term approaches $\int_0^x\int_0^y\varkappa_\xi(u,v)\upd v\upd u$ locally uniformly in $\R\times\R$. Moreover, we have that
\[
\int_\xi^{x_{\xi,N}}\sum_{j=0}^{J-1}\left(\frac1{2s}\widetilde\pi_{2j}(v)+2\xi\widetilde\pi_{2j+1}(v)\right)\upd v = N^{-1}\int_0^x\sum_{j=0}^{J-1}\left(\frac1{2s}\widetilde\pi_{2j}(v_{\xi,N})+2\xi\widetilde\pi_{2j+1}(v_{\xi,N})\right)\upd v
\]
and therefore this term converges to
\[
\int_0^x\frac{\omega(v\xi)}4\int_0^1(1-\lambda u)M(v\xi u)\upd u \upd v.
\]
Since the second integral for this sum can be handled similarly, we see that $\mathbf K^{(2,2)}_{2J}\big(x_{\xi,N},y_{\xi,N}\big)$ approaches $A_\xi(x,y)$ as claimed.

{\bf Case $N=2J+1$:} As in the case of even $N$, we only need to consider the scaled limits of $\mathbf K^{(1,2)}_N$ and $\mathbf K^{(2,2)}_N$. It follows from \eqref{IDS2} that
\[
\mathbf K^{(1,2)}_{2J+1}(z,y) - \mathbf K^{(1,2)}_{2J}(z,y) = - 2\sum_{j=0}^{J-1}\frac{s_{2j}}{s_{2J}}\bigg[\widetilde\pi_{2J}(z)\epsilon\widetilde\pi_{2j+1}(y)-\epsilon\widetilde\pi_{2J}(y)\widetilde\pi_{2j+1}(z)\bigg] + \frac{\widetilde\pi_{2J}(z)}{s_{2J}}
\]
which is equal to
\[
-\int_\xi^y\left(\mathbf K^{(1,1)}_{2J+1}(z,v)-\mathbf K^{(1,1)}_{2J}(z,v)\right)\upd v - \sum_{j=0}^{J-1}\frac{s_{2j}}{s_{2J}}\left(\frac1{2s}\widetilde\pi_{2J}(z)+2\xi\widetilde\pi_{2j+1}(z)\right)+ \frac{\widetilde\pi_{2J}(z)}{s_{2J}}
\]
by Lemma~\ref{lem:EO}. We get that
\begin{equation}
\label{amlostthere1}
\begin{array}{l}
\displaystyle \frac1N\int_\xi^{y_{\xi,N}}\left(\mathbf K^{(1,1)}_{2J+1}(z_{\xi,N},v)-\mathbf K^{(1,1)}_{2J}(z_{\xi,N},v)\right)\upd v \smallskip \\
\displaystyle \hspace{1.7in} = \frac1{N^2}\int_0^y\left(\mathbf K^{(1,1)}_{2J+1}(z_{\xi,N},v_{\xi,N})-\mathbf K^{(1,1)}_{2J}(z_{\xi,N},v_{\xi,N})\right)\upd v,
\end{array}
\end{equation}
which converges locally uniformly to zero by Theorem~\ref{thm:7}. Further, we have that
\begin{equation}
\label{amlostthere2}
\lim_{J\to\infty}\frac1{2sN}\widetilde\pi_{2J}\left(z_{\xi,N}\right)\sum_{j=0}^{J-1}\frac{s_{2j}}{s_{2J}} = \frac{\omega\big(z\xi\big)}4\frac{\lambda\sqrt{1-\lambda}}{1+\sqrt{1-\lambda}} M_{1/2,-1/2}(z\xi)
\end{equation}
by \eqref{Pnat1} and \eqref{sumPnat1} applied at $z=0$ with any pair of parameters such that $\gamma=0$.  Moreover, we get from Lemmas \ref{lem:asymptoticat1} \& \ref{lem:miracle} that
\begin{equation}
\label{amlostthere3}
\lim_{J\to\infty}\frac{2\xi}N \sum_{j=0}^{J-1}\frac{s_{2j}}{s_{2J}}\widetilde\pi_{2j+1}\left(z_{\xi,N}\right) = \frac{\omega\big(z\xi\big)}4(1-\lambda)M_{1/2,-1/2}(z\xi).
\end{equation}
At last, we have that $\lim_{J\to\infty}s_{2J}^{-1}=\sqrt{1-\lambda}/2$, which, in combination with \eqref{Pnat1}, yields that
\begin{equation}
\label{amlostthere4}
\lim_{J\to\infty}\frac1{s_{2J}N}\widetilde \pi_{2J}\left(z_{\xi,N}\right) = \frac{\omega\big(z\xi\big)}4\sqrt{1-\lambda}M_{1/2,-1/2}(z\xi).\end{equation}
Combining the conclusion of \eqref{amlostthere1} with \eqref{amlostthere2}---\eqref{amlostthere4}, we get that
\begin{equation}
\label{oops}
\lim_{J\to\infty}N^{-1}\left(\mathbf K^{(1,2)}_{2J+1}\big(z_{\xi,N},y_{\xi,N}\big) - \mathbf K^{(1,2)}_{2J}\big(z_{\xi,N},y_{\xi,N}\big)\right) = 0
\end{equation}
locally uniformly in $\C\times\R$.

Finally, let us settle the case $x,y\in\R$. The only difference from above is in the $\mathbf K^{(2,2)}_{2J+1}$ component. It holds that $\mathbf K^{(2,2)}_{2J+1}(x,y) - \mathbf K^{(2,2)}_{2J}(x,y)$ is equal to
\[
- 2\sum_{j=0}^{J-1}\frac{s_{2j}}{s_{2J}}\bigg[\epsilon\widetilde\pi_{2J}(x)\epsilon\widetilde\pi_{2j+1}(y)-\epsilon\widetilde\pi_{2J}(y)\epsilon\widetilde\pi_{2j+1}(x)\bigg] + \frac{\epsilon\widetilde\pi_{2J}(x)-\epsilon\widetilde\pi_{2J}(y)}{s_{2J}},
\]
which itself is equal to
\[
\begin{array}{l}
\displaystyle \int_\xi^x\int_\xi^y\left(\mathbf K^{(1,1)}_{2J+1}(u,v)-\mathbf K^{(1,1)}_{2J}(u,v)\right)\upd v \upd u \smallskip \\
\displaystyle \hspace{1in} + \left(\int_\xi^y-\int_\xi^x\right)\left(-\sum_{j=0}^{J-1}\frac{s_{2j}}{s_{2J}}\left[\frac1{2s}\widetilde\pi_{2J}(v)+2\xi\widetilde\pi_{2j+1}(v) \right] + \frac{\widetilde \pi_{2J}(v)}{s_{2J}} \right)\upd v.
\end{array}
\]
The claim of the theorem now follows by appealing to \eqref{oops} and the conclusion of \eqref{amlostthere1}.
\end{proof}

\subsection{Proofs of Theorems~\ref{thm:Inside}---\ref{thm:Outside}}

For the proof of Theorem~\ref{thm:Inside} we shall need the following restatement of Lemma~\ref{lem:EO}.
\begin{lemma}
\label{lem:EOat0}
For $y\in(-1,1)$, it holds that
\begin{equation}
\label{EOpiat0}
\left\{
\begin{array}{lll}
(\epsilon\widetilde\pi_{2n})(y) &=& \displaystyle -\int_0^y\widetilde\pi_{2n}(u)\upd u, \smallskip \\
(\epsilon\widetilde\pi_{2n+1})(y) &=& \displaystyle -\frac14P_{n+1}^{-1/2,-3/2}(0) - \int_0^y\widetilde\pi_{2n+1}(u)\upd u.
\end{array}
\right.
\end{equation}
\end{lemma}
\begin{proof}
We start with formulae \eqref{EOpi}. Since $-\int_1^y=-\int_0^y+\int_0^1$, we just need to compute the latter integral in both cases. It follows from Theorem~\ref{thm:2} and \eqref{pes} that
\begin{equation}
\label{needthisone}
\int_0^1 \widetilde\pi_{2n}(u)\upd u = P_n^{-1/2,-1/2}(1) = 1
\end{equation}
by \eqref{Pnat1}. Analogously, we have that
\begin{eqnarray}
\int_0^1 \widetilde\pi_{2n+1}(u)\upd u &=& \frac14\left(P_n^{-1/2,-3/2}(1) - P_{n+1}^{-1/2,-3/2}(0) - \frac1sP_n^{1/2,-3/2}(1)\right) \nonumber \\
&=&  -\frac14 P_{n+1}^{-1/2,-3/2}(0) - \frac1{4s}, \nonumber
\end{eqnarray}
where $P_n^{-1/2,-3/2}(1)=0$ by \eqref{Pnat1}. Combining the above equalities with \eqref{EOpi} yields \eqref{EOpiat0}. 
\end{proof}

\begin{proof}[Proof of Theorem~\ref{thm:Inside}]
To prove the theorem we need to show that the matrix kernel \eqref{eq:21} converges to \eqref{three-kernels} with $A=A_\D$ defined by \eqref{ISNInside}. As before, we employ \eqref{IDS1} \& \eqref{IDS2}.

{\bf Case $N=2J$:} Recall that formula \eqref{DSNK} holds for all $z,w$. Thus, $\mathbf K^{(1,1)}_{2J}(z,w)$ is equal to the right-hand side of \eqref{DSNK} when $z,w\in\D$. Therefore, we need to analyze the behavior of $K_j^{(i)}$, see \eqref{Ki}, in $\D\times\D$. Limit \eqref{SumInside2} yields that
\[
\lim_{J\to\infty}K_J^{(i)}\big(z^2,w^2\big) = \frac1{2\pi}\int_\T\frac{\Lambda_{-1/2,-3/2}(\tau)|\upd\tau|}{\big(1-z^2\overline\tau\big)^{3/2}\big(1-w^2\tau\big)^{(2i+1)/2}}
\]
locally uniformly in $\D\times\D$. As $\Lambda_{-1/2,-3/2}(\tau)=\sqrt{-\tau}$, one can readily verify that the functions $\mathbf K^{(1,1)}_{2J}(z,w)$ converge to $DA_\D D(z,w)$. This proves the complex/complex case the theorem.

Let now $y\in(-1,1)$. Then it follows from Lemma~\ref{lem:EOat0} that
\begin{eqnarray}
\mathbf K^{(1,2)}_{2J}(z,y) &=& -\int_0^y \mathbf K^{(1,1)}_{2J}(z,v)\upd v - \frac12\sum_{j=0}^{J-1}\pi_{2j}(z)P_{j+1}^{-1/2,-3/2}(0) \nonumber \\
&=& -\int_0^y \mathbf K^{(1,1)}_{2J}(z,v)\upd v - \frac12\left(K_{J+1}^{1/2,-1/2,-1/2,-3/2}(z^2,0)-K_{J+1}^{1/2,-3/2,-1/2,-3/2}(z^2,0)\right) \nonumber
\end{eqnarray}
where we used \eqref{pies} and \eqref{nn-1}. The first term in the last sum converges to
\[
-\int_0^y DA_\D D(z,v)\upd v = DA_\D(z,0) - DA_\D(z,y)
\]
and the second term converges to
\[
-\frac1{4\pi}\int_\T\frac{\left(\Lambda_{-1/2,-3/2}(\tau)-\Lambda_{-3/2,-3/2}(\tau)\right)|\upd\tau|}{\big(1-z^2\overline\tau\big)^{3/2}} = -DA_\D(z,0),
\]
by \eqref{SumInside2} and the computation
\[
\Lambda_{-1/2,-3/2}(\tau)-\Lambda_{-3/2,-3/2}(\tau) = - \Lambda_{-3/2,-1/2}(\tau) = -\Lambda_{-1/2,-3/2}(\overline\tau).
\]
This yields the complex/real case of the theorem.

Let now $x,y\in(-1,1)$. It follows from Lemma~\ref{lem:EOat0} that
\[
\mathbf K^{(2,2)}_{2J}(x,y) = \int_0^x\int_0^y \mathbf K^{(1,1)}_{2J}\upd v\upd u + \left(\int_0^x-\int_0^y\right)\left(\frac12\sum_{j=0}^{J-1}\pi_{2j}(u)P_{j+1}^{-1/2,-3/2}(0)\right)\upd u.
\]
Immediately, we get that
\[
\lim_{N\to\infty}\mathbf K^{(2,2)}_{2J}(x,y) = \int_0^x\int_0^y DA_\D D(u,v) \upd v\upd u  + \left(\int_0^x-\int_0^y\right) DA_\D(u,0)\upd u = A_\D(x,y),
\]
where for the last equality we used the fact that $A_\D(x,y)$ is anti-symmetric and is zero at $(0,0)$. This finishes the proof of the real/real case.

{\bf Case $N=2J+1$:} It holds by \eqref{convergence1} and \eqref{pies} that
\[
\left\{
\begin{array}{lcl}
\pi_{2J}(z) &=& c_J(-1/2)\left((1-z)^{-3/2}+o_J(1)\right) \bigskip\\
\pi_{2j+1}(w) &=& \frac w4 c_j(-3/2)\left(\left(1+\frac1s\right)(1-w)^{-3/2} - \frac3s(1-w)^{-5/2} + o_j(1)\right)
\end{array}
\right.
\]
locally uniformly for $z,w\in\D$. As the numbers $s_{2j}/s_{2J}$ are bounded and $c_J(-1/2)\sim (J+1)^{-1/2}$, $c_j(-3/2)\sim(j+1)^{-3/2}$ by \eqref{ckalphaas}, we get that
\[
\pi_{2J}(z)\sum_{j=0}^{J-1}\frac{s_{2j}}{s_{2J}}\pi_{2j+1}(w) = o_J(1)
\]
locally uniformly in $\D\times\D$. Therefore, $\mathbf K^{(1,1)}_{2J+1}$ has the same limit as $\mathbf K^{(1,1)}_{2J}$ as $J\to\infty$. 

Further, by \eqref{EOpiat0}, we have that
\begin{align*}
& \left(\mathbf K^{(1,2)}_{2J+1}-\mathbf K^{(1,2)}_{2J}\right)(z,y) = \\
& \hspace{2cm} -\int_0^y \left(\mathbf K^{(1,1)}_{2J+1}-\mathbf
  K^{(1,1)}_{2J}\right)(z,v)\upd v +
\frac{\pi_{2J}(z)}2\sum_{j=0}^{J-1}\frac{s_{2j}}{s_{2J}}P_{j+1}^{-1/2,-3/2}(0)
+ \frac{\pi_{2J}(z)}{s_{2J}}.
\end{align*}
 As $P_{j+1}^{-1/2,-3/2}(0)=c_{j+1}(-3/2)$ and $s_{2J}>1$, this difference converges to zero locally uniformly in $\D\times\D$. 

Finally, observe that the difference $\left(\mathbf K^{(2,2)}_{2J+1}-\mathbf K^{(2,2)}_{2J}\right)(x,y)$ is equal to
\begin{align*}
& \int_0^x\int_0^y \left(\mathbf K^{(1,1)}_{2J+1}-\mathbf
  K^{(1,1)}_{2J}\right)(u,v)\upd v\upd u  \\
& \hspace{2cm} +
\left(\int_0^y-\int_0^x\right)\left(\frac{\pi_{2J}(v)}2\sum_{j=0}^{J-1}\frac{s_{2j}}{s_{2J}}P_{j+1}^{-1/2,-3/2}(0)
  + \frac{\pi_{2J}(v)}{s_{2J}}\right)\upd v, 
\end{align*}
which converges to zero locally uniformly in $\D\times\D$. This finishes the proof of the theorem.
\end{proof}

The following lemma is needed for the proof of Theorem~\ref{thm:DSN}.

\begin{lemma}
\label{lem:2j+1-out}
It holds that
\[
\lim_{J\to\infty} \frac{\sqrt J}{s-2J-1}\frac1{z^{2J}}\sum_{j=0}^{J-1}\frac{s_{2j}}{s_{2J}}\pi_{2j+1}(z) = \frac \lambda{4\sqrt\pi}\frac1{\sqrt{z^2-1}}
\]
locally uniformly in $\Om$.
\end{lemma}
\begin{proof}
Using notation \eqref{ckalpha}, we get from Theorem~\ref{thm:2} that
\begin{eqnarray}
\pi_{2j+1}(z) &=& \frac {z^{2j+1}}{4s}\sum_{k=0}^j\big(s-2(k+1)\big)c_k(1/2)c_{j-k}(-3/2)z^{-(j-k)} \nonumber \\
&=& \frac {z^{2j+1}}{4s}c_j(1/2)\sum_{i=0}^j\big(s-2j-2+2i\big)\frac{c_{j-i}(1/2)}{c_j(1/2)}\frac{c_i(-3/2)}{z^{2i}} \nonumber \\
&=& \frac {z^{2j+1}}{4s}c_j(1/2)\left(\big(s-1-2j\big)\sum_{i=0}^j\frac{c_{j-i}(1/2)}{c_j(1/2)}\frac{c_i(-3/2)}{z^{2i}}-\sum_{i=0}^j\frac{c_{j-i}(1/2)}{c_j(1/2)}\frac{c_i(-1/2)}{z^{2i}}\right). \nonumber
\end{eqnarray}
Thus, using the standard normal family argument, we get that
\begin{eqnarray}
\pi_{2j+1}(z) &=& z^{2j+1}\frac{c_j(1/2)}{4s}\left(\big(s-1-2j\big)\left(f(z)+o_j(1)\right)-\left(1/f(z)+o_j(1)\right)\right) \nonumber \\
\label{2j+1-as}
&=& \frac{z^{2j+1}}{f(z)}\frac{c_j(1/2)}{4s}\left((s-2-2j)(1+o_j(1)) -  z^{-2}(s-1-2j)\right)
\end{eqnarray}
locally uniformly in $\Om$, where $f(z)=\sqrt{1-1/z^2}$. Observe that 
\[
c_j(1/2)=\left(1-\frac{1}{2j}\right)c_{j-1}(1/2) 
\]
and that $(s-2-2j)s_{2j} = (s-3-2j)s_{2j+2}$ by Lemma~\ref{lemma:2}. Therefore,
 \[
\big(s_{2j}/s_{2J}\big)\pi_{2j+1}(z) = (1+o_j(1))B_{j+1}(z) - (1+o_j(1))B_j(z),
\]
where we set
\[
B_j(z):=\frac{z^{2j-1}}{f(z)}\frac{c_j(1/2)}{4s}\frac{s_{2j}}{s_{2J}}(s-1-2j).
\]
Hence,
\[
\sum_{j=0}^{J-1}\frac{s_{2j}}{s_{2J}}\pi_{2j+1}(z) = B_J(z) - B_0(z) + \sum_{j=1}^{J-1}o_{J-j}(1)B_{J-j}(z). 
\]
For each fixed $j\geq0$, we have that
\[
\lim_{J\to\infty}\frac{\sqrt J}{s-2J-1}\frac{B_{J-j}(z)}{z^{2J}} = \frac{1+2jc^{-1}}{z^{2j+1}f(z)}\frac{\lambda}{4\sqrt\pi}\lim_{J\to\infty}\frac{s_{2J-2j}}{s_{2J}},
\]
where $c:=\lim_{J\to\infty}(s-2J-1)$ ($c^{-1}=0$ when $c=\infty$). The latter limit exists and is finite by Lemma~\ref{lemma:2} (clearly, equal to 1 when $j=0$). Moreover,
\[
\lim_{J\to\infty}\frac{\sqrt J}{s-2J-1}\frac{B_0(z)}{z^{2J}} = 0
\]
locally uniformly in $\Om$. Thus, the claim of the lemma now follows from the standard normal family argument.
\end{proof}

\begin{proof}[Proof of Theorem~\ref{thm:DSN}] We rely heavily on
  \eqref{eq:21}, \eqref{IDS1} \& \eqref{IDS2}, which will be used
  without explicit mention. 

{\bf Case $N=2J$:} In order to prove \eqref{eq:19}, we shall need to
repeat the argument leading to the proof of \eqref{SumOutside}. 
As in the previous lemma, set $f(z):=\sqrt{1-1/z^2}$. We get from
\eqref{2j+1-as} and \eqref{convergence2} that 
\begin{equation}
\label{monstrous}
\begin{array}{l}
\displaystyle \frac2{s-2J}\frac{f(z)f(w)}{(zw)^{2J-1}}\sum_{j=0}^{J-1}\big[\pi_{2j}(z)\pi_{2j+1}(w) -\pi_{2j}(w)\pi_{2j+1}(z)\big] =\smallskip \\
\hspace{.5in} \displaystyle \sum_{j=0}^{J-1}\frac{c_{J-1-j}^2(1/2)}{2s}\left(1+\frac{2j}{s-2J}\right)\frac{w(1+o_{J-1-j}(1))-z(1+o_{J-1-j}(1))}{(zw)^{2j+1}} + \smallskip \\
\hspace{.5in} \displaystyle +\sum_{j=0}^{J-1}\frac{c_{J-1-j}^2(1/2)}{2s}\left(1+\frac{2j+1}{s-2J}\right)\frac{w(1+o_{J-1-j}(1))-z(1+o_{J-1-j}(1))}{(zw)^{2j+2}}. 
\end{array}
\end{equation}
Since $\lim_{J\to\infty}(s-2J)=c$, by \eqref{ckalphaas} we have that
\[
\lim_{J\to\infty} \frac{c_{J-1-j}^2(1/2)}{2s} = \lim_{J\to\infty} \frac Ns \frac{c_{J-1-j}^2(1/2)}{4J} = \frac\lambda\pi,
\]
we deduce by employing the normal family argument that the right-hand side of \eqref{monstrous} converges to
\begin{eqnarray}
\frac\lambda\pi (w-z)\left[\sum_{j=0}^\infty\frac{1+2jc^{-1}}{(zw)^{2j+1}}+\sum_{j=0}^\infty\frac{1+(2j+1)c^{-1}}{(zw)^{2j+2}}\right] &=& \frac\lambda\pi \frac{w-z}{zw}\sum_{j=0}^\infty\frac{1+jc^{-1}}{(zw)^j} \nonumber \\
&=& \frac\lambda\pi\frac{w-z}{zw-1}\left[1+\frac{c^{-1}}{zw-1}\right] \nonumber
\end{eqnarray}
locally uniformly in $\Om$, from which the claim of the theorem easily follows.

{\bf Case $N=2J+1$:}
We get from \eqref{convergence2} and Lemma~\ref{lem:2j+1-out} that
\[
\lim_{J\to\infty}\frac2{s-2J-1}\frac{\pi_{2J}(z)}{(zw)^{2J+1}}\sum_{j=0}^{J-1}\frac{s_{2j}}{s_{2J}}\pi_{2j+1}(w) = \frac \lambda\pi\frac1{w\sqrt{z^2-1}\sqrt{w^2-1}}
\]
locally uniformly in $\Om\times\Om$. Then, as $\lim_{J\to\infty}(s-2J)=c+1$ now, it holds that
\begin{eqnarray}
\lim_{J\to\infty}\frac{|zw|^s}{(zw)^{2J+1}}\frac{DS_{2J+1}(z,w)}{s-2J-1} &=& \lim_{J\to\infty}\frac{|zw|^s}{(zw)^{2J+1}}\frac{DS_{2J}(z,w)}{s-2J-1} + \frac\lambda\pi\frac1{zw}\frac{w-z}{\sqrt{z^2-1}\sqrt{w^2-1}} \nonumber \\
&=& \frac\lambda\pi\frac{w-z}{\sqrt{z^2-1}\sqrt{w^2-1}}\frac1{zw}\left[\frac{1+c^{-1}}{zw-1}+\frac{c^{-1}}{(zw-1)^2}+1\right] \nonumber
\end{eqnarray}
locally uniformly in $\Om\times\Om$. Since
\[
\frac{1+c^{-1}}{zw-1}+\frac{c^{-1}}{(zw-1)^2}+1 = \frac{zw}{zw-1} + \frac{c^{-1}}{zw-1} +\frac{c^{-1}}{(zw-1)^2} =  \frac{zw}{zw-1} + \frac{zwc^{-1}}{zw-1},
\]
the theorem follows.
\end{proof}

For the proof of Theorem~\ref{thm:Outside} we shall need the following lemma.

\begin{lemma}
\label{lem:EOatInfty}
For $|y|>1$, it holds that
\begin{equation}
\label{EOOutside}
\left\{
\begin{array}{lll}
(\epsilon\widetilde\pi_{2n})(y) &=& \displaystyle -\frac{s_{2n}}2\xi -\int_{\xi\infty}^y \widetilde\pi_{2n}(u) \upd u
\displaystyle  ,
  \smallskip \\ 
(\epsilon\widetilde\pi_{2n+1})(y) &=& \displaystyle -\int_{\xi\infty}^y \widetilde\pi_{2n+1}(u) \upd u,
\end{array}
\right.
\end{equation}
where $\xi=\pm1$ and the numbers $s_k$ were defined in \eqref{s_k}.
\end{lemma}
\begin{proof}
We start with equation \eqref{EOpi}. Assume $y>1$. Then
$-\int_1^y=-\int_\infty^y-\int_1^\infty$. It holds by
Theorem~\ref{thm:2}, \eqref{pes} and \eqref{Pnat1} that
\[
 - \int_1^\infty \widetilde\pi_{2n+1}(u) \upd u = -\frac1{4s}P_n^{1/2,-3/2}(1) = -\frac1{4s}.
\]
 To
finish the proof of the second equality in \eqref{EOOutside}, we only
need to observe that, since $\widetilde\pi_{2n+1}(u)$ is an odd function,
\[
\int_{\infty}^y \widetilde\pi_{2n+1}(u) \upd u = \int_{-\infty}^y
\widetilde\pi_{2n+1}(u) \upd u.
\]
Analogously, we get that
\[
-\int_1^\infty \widetilde\pi_{2n}(u) \upd u = -\left(\int_0^\infty-\int_0^1\right)\widetilde\pi_{2n}(u) \upd u = -\frac{s_{2n}}2+1,
\]
where we used the fact that $\widetilde\pi_{2n}$ is even, \eqref{s_k},
and \eqref{needthisone}. Now, since 
\[
-\int_{\infty}^y \widetilde\pi_{2n}(u) \upd u = s_{2n} -
\int_{-\infty}^y \widetilde\pi_{2n}(u) \upd u, 
\]
the first equality in \eqref{EOOutside} follows. The case $y<-1$ can
be handled similarly. 
\end{proof}

\begin{proof}[Proof of Theorem~\ref{thm:Outside}] To prove the theorem
  we need to show that the matrix kernel \eqref{eq:21} converges to
  \eqref{three-kernels} with $A=A_\Om$ defined by \eqref{A_Om}. Again,
  we shall utilize \eqref{IDS1} \& \eqref{IDS2} without mentioning it
  . Clearly, the case $u,v\in\Om\setminus\R$ follows immediately from
  Theorem~\ref{thm:DSN} (see also \eqref{DSNtoZero} and \eqref{K_N-hat}). 

{\bf Case $N=2J$:} According to Lemma~\ref{lem:EOatInfty}, it holds that
\begin{eqnarray}
\mathbf K^{(1,2)}_{2J}(z,y) &=& - \int_{\xi_y\infty}^y \mathbf K^{(1,1)}_{2J}(z,v)\upd v + \xi_y\sum_{j=0}^{J-1}s_{2j}\widetilde\pi_{2j+1}(z), \nonumber \\
\mathbf K^{(2,2)}_{2J}(x,y) &=& \int_{\xi_x\infty}^x\int_{\xi_y\infty}^y \mathbf K^{(1,1)}_{2J}(u,v)\upd v\upd u +\left(\xi_x\int_{\xi_y\infty}^y-\xi_y\int_{\xi_x\infty}^x\right)\sum_{j=0}^{J-1}s_{2j}\widetilde\pi_{2j+1}(u)\upd u, \nonumber
\end{eqnarray}
where $\xi_u=\sgn(u)$, $u\in\R\setminus\{0\}$. From
Lemma~\ref{lemma:2}, we know that 
\[
s_{2j} = 2\big(1+o_s(1)\big)\sqrt\frac{s+1}2
\frac{\Gamma\left(\frac{s-2J-1}{2}\right)}{\Gamma\left(\frac{s-2J}{2}\right)}
\frac{s_{2j}}{s_{2J}},
\]
and therefore when $y\in\R$ and
$\lim_{J\to\infty}(s-2J)=c<\infty$ Lemma~\ref{lem:2j+1-out} implies 
\[
\lim_{J\to\infty}\frac{|zy|^{2J}}{(zy)^{2J}}\sum_{j=0}^{J-1}s_{2j}\widetilde\pi_{2j+1}(z) 
=  \frac{c-1}{2\sqrt\pi}\frac{\Gamma\left(\frac
    {c-1}2\right)}{\Gamma\left(\frac c2\right)}\frac
1{|z|^c\sqrt{z^2-1}}. 
\]
 Furthermore, the above limit equals zero when
 $\lim_{J\to\infty}(s-2J)=\infty$. Since
 $|zy|^{2J}/(zy)^{2J}=|zv|^{2J}/(zv)^{2J}$ for $y,v\in\R$ and
 $|uv|^{2J}=(uv)^{2J}$ for $u,v\in\R$, the claim of the theorem
 follows. 

{\bf Case $N=2J+1$:} by Lemma~\ref{lem:EOatInfty}, it holds that
\begin{equation}
\label{S-2J+1-1}
\mathbf K^{(1,2)}_{2J+1}(z,y) = \mathbf K^{(1,2)}_{2J}(z,y) -
\xi_y\sum_{j=0}^{J-1}s_{2j}\widetilde\pi_{2j+1}(z) -
\int_{\xi_y\infty}^y\left(\mathbf K^{(1,1)}_{2J+1}-\mathbf
  K^{(1,1)}_{2J}\right)(z,v)\upd v   +
\frac{\widetilde\pi_{2J}(z)}{s_{2J}}. 
\end{equation}
Since
$|zy|^{2J+1}/(zy)^{2J+1}=\xi_y|z|^{2J+1}/z^{2J+1}=|zv|^{2J+1}/(zv)^{2J+1}$,
when $c=\lim_{J\to\infty}(s-2J-1)$, we deduce from the first part of
the proof that 
\begin{equation}
\label{S-2J+1-2}
\lim_{J\to\infty}\frac{|zy|^{2J+1}}{(zy)^{2J+1}}\left[\mathbf
  K^{(1,2)}_{2J}(z,y) -
  \xi_y\sum_{j=0}^{J-1}s_{2j}\widetilde\pi_{2j+1}(z) \right]=
-\int_{\xi_y\infty}^y \frac{B(z,v)}{zv}\upd v.
\end{equation}
Furthermore, it holds that
\[
\begin{array}{l}
\displaystyle \frac{|zy|^{2J+1}}{(zy)^{2J+1}}\int_{\xi_y\infty}^y\left(\mathbf K^{(1,1)}_{2J+1}-\mathbf K^{(1,1)}_{2J}\right)(z,v)\upd v = \int_{\xi_y\infty}^y\left(1-\frac1{zv}\right)\frac{|zv|^{2J+1}}{(zv)^{2J}}\mathbf K^{(1,1)}_{2J}(z,v)\upd v \bigskip \\
\hspace{4cm} \displaystyle +
\int_{\xi_y\infty}^y\left(\frac{|zv|^{2J+1}}{(zv)^{2J+1}}\mathbf
  K^{(1,1)}_{2J+1}(z,v)-\frac{|zv|^{2J+1}}{(zv)^{2J}}\mathbf
  K^{(1,1)}_{2J}(z,v)\right)\upd v,
\end{array}
\]
and therefore
\begin{equation}
\label{S-2J+1-3}
\lim_{J\to\infty}\frac{|zy|^{2J+1}}{(zy)^{2J+1}}\int_{\xi_y\infty}^y\left(\mathbf K^{(1,1)}_{2J+1}-\mathbf K^{(1,1)}_{2J}\right)(z,v)\upd v = \int_{\xi_y\infty}^y\left(1-\frac1{zv}\right)B(z,v)\upd v.
\end{equation}
It also true that
\begin{equation}
\label{S-2J+1-4}
\lim_{J\to\infty}\frac{|zy|^{2J+1}}{(zy)^{2J+1}}\frac{\widetilde\pi_{2J}(z)}{s_{2J}} = \frac{\xi_y}{|z|^c\sqrt{z^2-1}}\lim_{J\to\infty}\frac{c_J(1/2)}{s_{2J}} = \frac1{\sqrt\pi}\frac{\Gamma\left(\frac{c+1}2\right)}{\Gamma\left(\frac c2\right)}\frac{\xi_y}{|z|^c\sqrt{z^2-1}}.
\end{equation}
Substituting \eqref{S-2J+1-2}---\eqref{S-2J+1-4} into
\eqref{S-2J+1-1}, we obtain the desired result. Going through the
steps above, it is rather straightforward to show that the
corresponding limit is zero when $\lim_{J\to\infty}(s-2J)=\infty$.  

Finally, observe that by Lemma~\ref{lem:EOatInfty},
\[
\begin{array}{l}
\displaystyle \mathbf K^{(2,2)}_{2J+1}(x,y) = \mathbf K^{(2,2)}_{2J}(x,y) + \left(\xi_y\int_{\xi_x\infty}^x-\xi_x\int_{\xi_y\infty}^y\right)\sum_{j=0}^{J-1}s_{2j}\widetilde\pi_{2j+1}(u)\upd u \bigskip \\
\hspace{1cm} \displaystyle +
\int_{\xi_x\infty}^x\int_{\xi_y\infty}^y\left(\mathbf
  K^{(1,1)}_{2J+1}-\mathbf K^{(1,1)}_{2J}\right)(u,v)\upd v\upd u   +
\left(\int_{\xi_y\infty}^y -
  \int_{\xi_x\infty}^x\right)\frac{\widetilde\pi_{2J}(u)}{s_{2J}}\upd
u + \frac{\xi_y-\xi_x}2.
\end{array}
\]
The last case of the theorem is now proved by doing steps similar to \eqref{S-2J+1-2}---\eqref{S-2J+1-4} and observing that
\[
\xi_x\xi_y \left(\frac{\xi_y-\xi_x}2+\frac12\sgn(x-y)\right) = \frac12\sgn(x-y). \qedhere
\]
\end{proof}

\subsection{Proof of Theorem~\ref{thm:expected}}

For the proof of the theorem we shall need several auxiliary computations.

\begin{lemma}
\label{lem:EOagain}
It holds that
\[
\left\{
\begin{array}{lcll}
\displaystyle \epsilon\widetilde\pi_{2n}(y) &=& \displaystyle  -yP_n^{-1/2,-1/2}(y^2), & |y|\leq 1, \smallskip \\
\displaystyle \epsilon\widetilde\pi_{2n+1}(y) &=& \displaystyle  \frac{y^2}{4s}P_n^{1/2,-3/2}(y^2), & |y|\geq1.
\end{array}
\right.
\]
In particular, $\epsilon\widetilde\pi_{2n}(\pm1)=\mp1$ and $\epsilon\widetilde\pi_{2n+1}(\pm1)=\frac1{4s}$.
\end{lemma}
\begin{proof}
Evaluating the integrals appearing in the proof of Lemma~\ref{lem:EO}, we get that
\[
\left\{
\begin{array}{lcll}
\displaystyle \epsilon\left(x^{2k}\max\big\{1,|x|\big\}^{-s}\right)(y) &=& \displaystyle  -\frac{y^{2k+1}}{2k+1}, & |y|\leq 1, \smallskip \\
\displaystyle \epsilon\left(x^{2k+1}\max\big\{1,|x|\big\}^{-s}\right)(y) &=& \displaystyle  \frac{|y|^{2k+2-s}}{s-2k-2}, & |y|\geq1.
\end{array}
\right.
\]
Thus, the formula for $\epsilon\widetilde\pi_{2n}(y)$ follows from \eqref{pies} \& \eqref{pes} and the fact that $(2k+1)c_k(-1/2)=c_k(1/2)$, see \eqref{ckalpha} for the definition of these constants. Moreover, we can rewrite the formula for $\pi_{2n+1}$ in Theorem~\ref{thm:2} as
\[
\pi_{2n+1}(z) = \frac1{4s}\sum_{k=0}^n(s-2k-2)c_k(1/2)c_{n-k}(-3/2)x^{2k+1}
\]
from which formula for $\epsilon\widetilde\pi_{2n+1}(y)$ easily follows. The final claim of the lemma is an immediate consequence of Proposition~\ref{prop:asymptoticat1}.
\end{proof}

\begin{lemma}
\label{lem:whowouldguess}
For $x>-1$, it holds that
\[
\sum_{m=0}^M \frac{\Gamma(m+1/2)}{\Gamma(m+1)}\frac{\Gamma(M-m+1+x)}{\Gamma(M-m+5/2+x)} = \frac4{2M+3+2x}\frac{\Gamma(M+3/2)}{\Gamma(M+1)}\frac{\Gamma(1+x)}{\Gamma(3/2+x)}.
\]
\end{lemma}
\begin{proof}
The case $M=0$ is elementary. Thus, it only remains to complete the inductive step. The sum we are computing is equal to
\[
\sum_{m=0}^{M-1} \frac{\Gamma(m+1/2)}{\Gamma(m+1)}\frac{\Gamma(M-1-m+1+x+1)}{\Gamma(M-1-m+5/2+x+1)} + \frac{\Gamma(M+1/2)}{\Gamma(M+1)}\frac{\Gamma(1+x)}{\Gamma(5/2+x)}
\]
and, by the inductive hypothesis, to
\[
\frac4{2(M-1)+3+2(x+1)}\frac{\Gamma(M+1/2)}{\Gamma(M)}\frac{\Gamma(1+x)}{\Gamma(5/2+x)} + \frac{\Gamma(M+1/2)}{\Gamma(M+1)}\frac{\Gamma(1+x)}{\Gamma(5/2+x)}
\]
and hence to
\[
\frac4{2M+3+2x}\frac{\Gamma(M+1/2)}{\Gamma(M+1)}\frac{\Gamma(1+x)}{\Gamma(5/2+x)}\left[M(1+x)+\frac{2M+3+2x}4\right].
\]
This finishes the proof as the term in square brackets factors as $(M+1/2)(3/2+x)$.
\end{proof}
\begin{lemma}
\label{lem:thelast}
It holds for $x>0$ that
\[
\sum_{j=0}^{J-1}\frac{\Gamma(j+1/2)}{\Gamma(j+1)}\frac{\Gamma(j+x)}{\Gamma(j+3/2+x)} = \frac2x\frac{\Gamma(J+1/2)}{\Gamma(J)}\frac{\Gamma(J+x)}{\Gamma(J+1/2+x)}.
\]
\end{lemma}
\begin{proof}
The case $J=1$ is trivial. Thus, we need only to show the inductive step. The sum we are computing is equal to
\[
\sum_{j=0}^{J-2}\frac{\Gamma(j+1/2)}{\Gamma(j+1)}\frac{\Gamma(j+x)}{\Gamma(j+3/2+x)} + \frac{\Gamma(J-1/2)}{\Gamma(J)}\frac{\Gamma(J-1+x)}{\Gamma(J+1/2+x)},
\]
which, by the inductive hypothesis, can be written as
\[
\frac2x\frac{\Gamma(J-1/2)}{\Gamma(J-1)}\frac{\Gamma(J-1+x)}{\Gamma(J-1/2+x)} + \frac{\Gamma(J-1/2)}{\Gamma(J)}\frac{\Gamma(J-1+x)}{\Gamma(J+1/2+x)},
\]
which, in turn, is equal to
\[
\frac2x\frac{\Gamma(J-1/2)}{\Gamma(J)}\frac{\Gamma(J-1+x)}{\Gamma(J+1/2+x)}\left[(J-1)\left(J-\frac12+x\right)+\frac x2\right].
\]
As the term in square brackets factors as $(J-1/2)(J-1+x)$, the proof of the lemma is complete.
\end{proof}

\begin{proof}[Proof of Theorem~\ref{thm:expected}]

{\bf Case $N=2J$:} For a measurable set $A\subseteq\R$, denote by $N_A$ the number of real roots of a random polynomial of degree $N$ from the real Mahler ensemble that belong to $A$. It follows from \eqref{eq:16}, \eqref{MatrixKernels}, \eqref{eq:21}, and \eqref{IDS1} that
\begin{eqnarray}
E[N_A] &=& \int_A \Pf \big[\mathbf K_N(x,x)\big] \upd\mu_\R(x) = \int_A \Pf\left[\begin{array}{cc}0 & \varkappa_N\epsilon(x,x)\smallskip\\ \epsilon\varkappa_N(x,x) & 0 \end{array}\right] \upd\mu_\R(x) \nonumber \\
\label{EA}
&=& \sum_{n=0}^{J-1}\int_A2\big[\widetilde\pi_{2n}(x)\epsilon\widetilde\pi_{2n+1}(x) - \widetilde\pi_{2n+1}(x)\epsilon\widetilde\pi_{2n}(x)\big]\upd\mu_\R(x). 
\end{eqnarray}

Let $A=[-1,1]$, where we abbreviate $N_\mathsf{in}=N_{[-1,1]}$. Recall that for $x\in[-1,1]$, $\widetilde\pi_k(x)=\pi_k(x)$. Further, since $\epsilon$ operator for real arguments essentially amounts to anti-differentiation, see the paragraph after Theorem~\ref{thm:5}, we also have that $(\epsilon\widetilde\pi_k)^\prime(x)=-\widetilde\pi_k(x)$. Therefore,
\begin{eqnarray}
E[N_\mathsf{in}] &=& \sum_{n=0}^{J-1}\left[ -2\left.\epsilon\widetilde\pi_{2n+1}(x)\epsilon\widetilde\pi_{2n}(x)\right|_{-1}^1 - 4\int_{-1}^1\pi_{2n+1}(x)\epsilon\widetilde\pi_{2n}(x)\upd\mu_\R(x)\right] \nonumber \\
\label{Ein}
&=:& Js^{-1} + \sum_{n=0}^{J-1} I_n,
\end{eqnarray}
where we used the second conclusion of
Lemma~\ref{lem:EOagain}. Appealing to Lemma~\ref{lem:EOagain} once
more, as well as to Theorem~\ref{thm:2}, we obtain that $I_n$ equals
\[
\frac1{\pi s}\sum_{m=0}^n\frac{\Gamma(m+1/2)}{\Gamma(m+1)}\frac{\Gamma(n-m+1/2)}{\Gamma(n-m+1)}\left[-\frac2\pi\sum_{k=0}^n\frac{\Gamma(k+3/2)}{\Gamma(k+1)}\frac{\Gamma(n-k-1/2)}{\Gamma(n-k+1)}\frac{s-2k-2}{2k+2m+3}\right].
\]
Since $s-2k-2=-(2k+2m+3)+(s+2m+1)$ and applying Lemma~\ref{lemma:7} with $a=3/2$, $b=-1/2$, and  $x=-m-1$, we see that the term in square brackets from the above equation is equal to
\[
-2P_n^{1/2,-3/2}(1) - (s+2m+1)\frac{\big(-(m+1)\big)\cdots\big(-(m+1)-(n-1)\big)}{\big(-(m+3/2)\big)\cdots\big(-(m+3/2)-n\big)}.
\]
As $P_n^{1/2,-3/2}(1) = P_n^{-1/2,-1/2}(1)=1$ by Proposition~\ref{prop:asymptoticat1}, we further get that
\begin{equation}
\label{ilsdgh}
I_n = -\frac2s + \frac1{\pi s}\sum_{m=0}^n(s+2m+1)\frac{\Gamma(m+1/2)}{\Gamma(m+1)}\frac{\Gamma(n-m+1/2)}{\Gamma(n-m+1)}\frac{\Gamma(n+m+1)}{\Gamma(n+m+5/2)}\frac{\Gamma(m+3/2)}{\Gamma(m+1)}.
\end{equation}
Observe that
\[
\frac{\Gamma(m+1/2)}{\Gamma(m+1)}\frac{\Gamma(m+3/2)}{\Gamma(m+1)} = 1+ O\left(\frac1{m+1}\right)
\]
by \eqref{ckalphaas}. Thus, upon replacing $m$ by $n-m$, we get that
\[
I_n = -\frac2s + \frac1{\pi s}\sum_{m=0}^n\left(1+ O\left(\frac1{n-m+1}\right)\right)\big[s+2n+2-2m-1\big]\frac{\Gamma(m+1/2)}{\Gamma(m+1)}\frac{\Gamma(2n-m+1)}{\Gamma(2n-m+5/2)}. 
\]
It follows from Lemma~\ref{lem:whowouldguess} applied with $M=n$ and $x=M$ that
\[
\sum_{m=0}^n\frac{\Gamma(m+1/2)}{\Gamma(m+1)}\frac{\Gamma(2n-m+1)}{\Gamma(2n-m+5/2)} = \frac4{4n+3}
\]
and therefore
\[
\sum_{m=0}^n\left(1 + O\left(\frac1{n-m+1}\right)\right)\frac{\Gamma(m+1/2)}{\Gamma(m+1)}\frac{\Gamma(2n-m+1)}{\Gamma(2n-m+5/2)} = \frac4{4n+3} + O\left(\frac{\log n}{(n+1)^{3/2}}\right),
\]
where the term $O(\cdot)$ follows from the estimates
\[
\frac{\Gamma(m+1/2)}{\Gamma(m+1)}\frac{\Gamma(2n-m+1)}{\Gamma(2n-m+5/2)} \leq \frac{\Gamma(n+1)}{\Gamma(n+5/2)}.
\]
Thus, we deduce that
\[
\begin{array}{l}
\displaystyle I_n = \frac{s+2n+2}{\pi s}\left(\frac4{4n+3}+ O\left(\frac{\log n}{(n+1)^{3/2}}\right)\right) \smallskip \\
\hspace{4cm}\displaystyle  -\frac2s - \frac2\pi\sum_{m=0}^n\big(1+ o_{n-m}(1)\big)\frac{\Gamma(m+3/2)}{\Gamma(m+1)}\frac{\Gamma(2n-m+1)}{\Gamma(2n-m+5/2)}.
\end{array}
\]
Using the monotonicity of the second fraction in the sum above once more and since
\[
\sum_{m=0}^n\frac{\Gamma(m+3/2)}{\Gamma(m+1)} = \frac23\frac{\Gamma(n+5/2)}{\Gamma(n+1)},
\]
it holds that
\[
\frac23\frac{\Gamma(n+5/2)}{\Gamma(n+1)}\frac{\Gamma(2n+1)}{\Gamma(2n+5/2)} \leq \sum_{m=0}^n\frac{\Gamma(m+3/2)}{\Gamma(m+1)}\frac{\Gamma(2n-m+1)}{\Gamma(2n-m+5/2)} \leq \frac23.
\]
It now easily follows from \eqref{ckalphaas} that
\[
I_n = \frac1\pi\left(\frac1{n+3/4} + O\left(\frac{\log n}{(n+1)^{3/2}}\right)\right) + \frac{\big\{\text{a term, which is uniformly bounded with }n\big\}}s.
\]
By plugging the above expression into \eqref{Ein}, we obtain the first claim of the theorem.

Let now $A=\R\setminus(-1,1)$, in which case we abbreviate $N_\mathsf{out}=N_{\R\setminus(-1,1)}$. As the integrand of \eqref{EA} is an even function of $x$, we can write
\begin{eqnarray}
E[N_\mathsf{out}] &=& \sum_{n=0}^{J-1}\left[ 4\left.\epsilon\widetilde\pi_{2n+1}(x)\epsilon\widetilde\pi_{2n}(x)\right|_1^\infty + 8\int_1^\infty\epsilon\widetilde\pi_{2n+1}(x)\widetilde\pi_{2n}(x)\upd\mu_\R(x)\right] \nonumber \\
\label{Eout}
&=:& Js^{-1} + \sum_{n=0}^{J-1} I_{n,s},
\end{eqnarray}
where we used the second conclusion of Lemma~\ref{lem:EOagain} once
more. As before, by appealing to Lemma~\ref{lem:EOagain} and
Theorem~\ref{thm:2}, we deduce that $I_{n,s}$ equals
\[
\frac2{\pi s}\sum_{i=0}^n\frac{\Gamma(i+3/2)}{\Gamma(i+1)}\frac{\Gamma(n-i+1/2)}{\Gamma(n-i+1)}\left[-\frac1\pi\sum_{k=0}^n\frac{\Gamma(k+3/2)}{\Gamma(k+1)}\frac{\Gamma(n-k-1/2)}{\Gamma(n-k+1)}\frac1{s-i-k-3/2}\right].
\]
The term in square brackets can be summed up using Lemma~\ref{lemma:7} applied with $a=3/2$, $b=-1/2$, and $x=s-i-1$, to yield
\begin{equation}
\label{l;ds}
I_{n,s} =\frac2{\pi s}\sum_{i=0}^n\frac{\Gamma(i+3/2)}{\Gamma(i+1)}\frac{\Gamma(n-i+1/2)}{\Gamma(n-i+1)}\frac{\Gamma(s-i)}{\Gamma(s-i-1/2)}\frac{\Gamma(s-i-n-3/2)}{\Gamma(s-i-n)}.
\end{equation}
By plugging the above expression into \eqref{Eout} we get
\[
E[N_\mathsf{out}]  = Js^{-1} + \frac2{\pi 
  s}\sum_{i=0}^{J-1}
\frac{\Gamma(i+3/2)}{\Gamma(i+1)}
\frac{\Gamma(s-i)}{\Gamma(s-i-1/2)}\sum_{n=i}^{J-1}\frac{\Gamma(n-i+1/2)}{\Gamma(n-i+1)}\frac{\Gamma(s-i-n-3/2)}{\Gamma(s-i-n)}.   
\]
Since
\[
\sum_{n=i}^{J-1}\frac{\Gamma(n-i+1/2)}{\Gamma(n-i+1)}\frac{\Gamma(s-i-n-3/2)}{\Gamma(s-i-n)}
=
\sum_{m=0}^{J-1-i}\frac{\Gamma(m+1/2)}{\Gamma(m+1)}\frac{\Gamma(s-2i-m-3/2)}{\Gamma(s-2i-m)}, 
\]
we get from Lemma~\ref{lem:whowouldguess} applied with $M=J-1-i$ and $x=s-J-i-3/2$ that
\begin{align*}
 E[N_\mathsf{out}]  = Js^{-1} + \frac2{\pi
  s}\sum_{i=0}^{J-1}
\frac{\Gamma(i+3/2)\Gamma(s-i)\Gamma(s-J-i-1/2)\Gamma(J-i+1/2)}{\Gamma(i+1)
  \Gamma(s-i-1/2)\Gamma(s-J-i)\Gamma(J-i)} 
\frac2{s-1-2i} \smallskip \\
 = Js^{-1} + \frac2{\pi
  s}\sum_{m=0}^{J-1}
\frac{\Gamma(J-m+1/2) \Gamma(s-J+m+1) \Gamma(\Delta+m-1/2)
  \Gamma(m+3/2)}{\Gamma(J-m) \Gamma(s-J+m+1/2) \Gamma(\Delta+m) \Gamma(m+1)} 
\frac1{m+\Delta/2}, 
\end{align*}
where $\Delta:=s-N+1$. Using \eqref{ckalphaas}, we can rewrite the sum above as
\[
\frac2{\pi s}\sum_{m=0}^{J-1}\frac{\Gamma(J-m+1/2)}{\Gamma(J-m)}\frac{\Gamma(s-J+m+1)}{\Gamma(s-J+m+1/2)}\sqrt{\frac{m+1}{m+\Delta}}\left(1+O\left(\frac1{m+1}\right)\right)\frac1{m+\Delta/2}.
\]
Since
\[
\sqrt{\frac{m+1}{m+\Delta}}\frac{\Gamma(J-m+1/2)}{\Gamma(J-m)}\frac{\Gamma(s-J+m+1)}{\Gamma(s-J+m+1/2)} \leq \frac{\Gamma(J+1/2)}{\Gamma(J)}\frac{\Gamma(s)}{\Gamma(s-1/2)},
\]
we have that
\begin{align*}
& \frac2{\pi
  s}\sum_{m=0}^{J-1}\frac{\Gamma(J-m+1/2)}{\Gamma(J-m)}\frac{\Gamma(s-J+m+1)}{\Gamma(s-J+m+1/2)}\sqrt{\frac{m+1}{m+\Delta}}O\left(\frac1{m+1}\right)\frac1{m+\Delta/2}
\\
& \hspace{10cm} = \sqrt{Ns^{-1}}O_N(1),
\end{align*}
and therefore
\[
E[N_\mathsf{out}] = \sqrt{Ns^{-1}}O_N(1) + \frac2{\pi s}\sum_{m=0}^{J-1}\frac{\Gamma(J-m+1/2)}{\Gamma(J-m)}\frac{\Gamma(s-J+m+1)}{\Gamma(s-J+m+1/2)}\sqrt{\frac{m+1}{m+\Delta}}\frac1{m+\Delta/2}.
\]
Furthermore, as
\[
1-\sqrt{\frac{m+1}{m+\Delta}} = \frac{\Delta-1}{\sqrt{m+\Delta}(\sqrt{m+\Delta}+\sqrt{m+1})} \leq \frac{\Delta}{m+\Delta/2},
\]
and since, estimating as before,
\[
\frac2{\pi s}\sum_{m=0}^{J-1}\frac{\Gamma(J-m+1/2)}{\Gamma(J-m)}\frac{\Gamma(s-J+m+1)}{\Gamma(s-J+m+1/2)}\frac{\Delta}{(m+\Delta/2)^2} \leq \sqrt{Ns^{-1}}O_N(1)\int_0^\infty\frac{\Delta\upd x}{(x+\Delta/2)^2},
\]
where the integral is equal to 2, we have that
\[
E[N_\mathsf{out}] = \sqrt{Ns^{-1}}O_N(1) + \frac2{\pi s}\sum_{m=0}^{J-1}\frac{\Gamma(J-m+1/2)}{\Gamma(J-m)}\frac{\Gamma(s-J+m+1)}{\Gamma(s-J+m+1/2)}\frac1{m+\Delta/2}.
\]
Continuing on the path of estimates, observe that
\[
\frac{\Gamma(J-m+1/2)}{\Gamma(J-m)}\frac{\Gamma(s-J+m+1)}{\Gamma(s-J+m+1/2)} = \sqrt{(J-m)(s-J+m)}\left(1+O\left(\frac1{J-m}\right)\right)
\]
by \eqref{ckalphaas}. As $s-J=\Delta+J-1$ and respectively
\[
\frac{\sqrt{s-J+m}}{m+\Delta/2}\sqrt{(J-m)}O\left(\frac1{J-m}\right) \leq O\left(\sqrt{\frac J{J-m}}\right),
\]
we get that
\[
\frac2{\pi s}\sum_{m=0}^{J-1}\sqrt{(J-m)(s-J+m)}O\left(\frac1{J-m}\right)\frac1{m+\Delta/2} \leq Ns^{-1}O_N(1)
\]
and therefore
\[
E[N_\mathsf{out}] = \sqrt{Ns^{-1}}O_N(1) + \frac2{\pi s}\sum_{m=0}^{J-1}\sqrt{(J-m)(s-J+m)}\frac1{m+\Delta/2}.
\]
Now, it holds that
\[
\sqrt{J(s-J)}-\sqrt{(J-m)(s-J+m)} = \frac{m(m-1+\Delta)}{\sqrt{J(s-J)}+\sqrt{(J-m)(s-J+m)}} \leq \frac{m(2m+\Delta)}{\sqrt{J(s-J)}}
\]
and that
\[
\frac2{\pi s}\sum_{m=0}^{J-1}\frac{m(2m+\Delta)}{\sqrt{J(s-J)}}\frac1{m+\Delta/2} = \frac1\pi\frac{\sqrt{J}(J-1)}{s\sqrt{s-J}} \leq \frac1\pi Js^{-1},
\]
where we used the fact that $s-J>J$. Hence,
\[
E[N_\mathsf{out}] = \sqrt{Ns^{-1}}O_N(1) + \frac2\pi\frac{\sqrt{J(s-J)}}s\sum_{m=0}^{J-1}\frac1{m+\Delta/2}.
\]
Finally, it only remains to notice that
\[
\sum_{m=1}^J\frac1{m+\Delta/2} \leq \int_0^J\frac{\upd x}{x+\Delta/2} \leq \sum_{m=0}^{J-1}\frac1{m+\Delta/2},
\]
which yields that
\[
\sum_{m=0}^{J-1}\frac1{m+\Delta/2} = \log(N+\Delta) - \log\Delta + O_N(1) = -\log\big(1-Ns^{-1}\big) + O_N(1).
\]

{\bf Case $N=2J+1$:} It follows from \eqref{IDS2} that to prove the asymptotic formula for $E[N_\mathsf{in}]$, we need to show that
\begin{equation}
\label{justthisone}
\int_{-1}^1\left(- 2\sum_{n=0}^{J-1}\frac{s_{2n}}{s_{2J}}\big[\widetilde\pi_{2J}(x)\epsilon\widetilde\pi_{2n+1}(x)-\epsilon\widetilde\pi_{2J}(x)\widetilde\pi_{2n+1}(x)\big] + \frac{\widetilde\pi_{2J}(x)}{s_{2J}}\right)\upd\mu_\R(x) = O_N(1).
\end{equation}
One can explicitly compute exactly as in Lemma~\ref{lem:EOagain} that
\begin{equation}
\label{;a}
s_{2J}^{-1}\int_{-1}^1\widetilde\pi_{2J}(x)\upd\mu_\R(x) = 2s_{2J}^{-1}P_J^{-1/2,-1/2}(1) =2s_{2J}^{-1} = \frac{\Gamma\left(\frac {s+1}2\right)}{\Gamma\left(\frac{s+2}2\right)}\frac{\Gamma\left(\frac{s-N+1}2\right)}{\Gamma\left(\frac{s-N}2\right)} \leq 1
\end{equation}
by Proposition~\ref{prop:asymptoticat1}, the definition $s_{2J}$, see \eqref{s_k}, and \eqref{ckalphaas}. To estimate the remaining part of the integral observe that
\[
2\int_{-1}^1\left(\widetilde\pi_{2J}(x)\epsilon\widetilde\pi_{2n+1}(x)-\epsilon\widetilde\pi_{2J}(x)\widetilde\pi_{2n+1}(x)\right)\upd\mu_\R(x) = \frac1s - 4\int_{-1}^1\epsilon\widetilde\pi_{2J}(x)\widetilde\pi_{2n+1}(x)\upd\mu_\R(x)
\]
exactly as in \eqref{Ein}. Moreover, as in \eqref{ilsdgh}, we have that the above quantity is equal to
\begin{equation}
\label{poiu}
-\frac1s + \frac1{\pi s}\sum_{m=0}^J(s+2m+1)\frac{\Gamma(m+1/2)}{\Gamma(m+1)}\frac{\Gamma(J-m+1/2)}{\Gamma(J-m+1)}\frac{\Gamma(n+m+1)}{\Gamma(n+m+5/2)}\frac{\Gamma(m+3/2)}{\Gamma(m+1)}.
\end{equation}
Observe that
\[
\sum_{m=0}^J\frac{\Gamma(J-m+1/2)}{\Gamma(J-m+1)}\frac{\Gamma(n+m+1)}{\Gamma(n+m+5/2)} = \frac4{2(J+n)+3}\frac{\Gamma(J+3/2)}{\Gamma(J+1)}\frac{\Gamma(n+1)}{\Gamma(n+3/2)} \leq \frac{O_N(1)}{\sqrt{J(n+1)}}
\]
by performing the substitution $m\mapsto J-m$ and applying Lemma~\ref{lem:whowouldguess} with $M=J$ and $x=n$, as well as by using \eqref{ckalphaas}. Since
\[
\frac{s+2m+1}{\pi s}\frac{\Gamma(m+1/2)}{\Gamma(m+1)}\frac{\Gamma(m+3/2)}{\Gamma(m+1)} =O(1),
\]
the sum in \eqref{poiu} is bounded by a constant times $1/\sqrt{J(n+1)}$. As the numbers $s_{2n}$ increase with $n$ and hence $s_{2n}/s_{2J}\leq1$ for $n\leq J$, and since $\sum_{n=0}^{J-1}1/\sqrt{J(n+1)}=O_N(1)$, we see that \eqref{justthisone} indeed takes place.

To prove the asymptotic formula for $E[N_\mathsf{out}]$, we need to show that
\begin{align}
&\int_{\R\setminus(-1,1)}\left(-
  2\sum_{n=0}^{J-1}\frac{s_{2n}}{s_{2J}}\big[\widetilde\pi_{2J}(x)\epsilon\widetilde\pi_{2n+1}(x)-\epsilon\widetilde\pi_{2J}(x)\widetilde\pi_{2n+1}(x)\big]
  + \frac{\widetilde\pi_{2J}(x)}{s_{2J}}\right)\upd\mu_\R(x) \nonumber
\\  
\label{andthisone}
&\hspace{10cm} =
\sqrt{Ns^{-1}}O_N(1). 
\end{align}
We immediately deduce from the definition of $s_{2J}$ and \eqref{;a} that
\[
s_{2J}^{-1}\int_{\R\setminus(-1,1)}\widetilde\pi_{2J}(x)\upd\mu_\R(x) = 1 - \frac{\Gamma\left(\frac{s+1}2\right)}{\Gamma\left(\frac{s+2}2\right)}\frac{\Gamma\left(\frac{s-N+1}2\right)}{\Gamma\left(\frac{s-N}2\right)} = Ns^{-1}O_N(1) ,
\]
where, as usual, we used \eqref{ckalphaas}. Further, we have as in
\eqref{Eout} that 
\begin{align*}
&2\int_{\R\setminus(-1,1)}\left(\widetilde\pi_{2J}(x)\epsilon\widetilde\pi_{2n+1}(x)-\epsilon\widetilde\pi_{2J}(x)\widetilde\pi_{2n+1}(x)\right)\upd\mu_\R(x) \\
& \hspace{4cm} = \frac1s +
8\int_1^\infty\epsilon\widetilde\pi_{2n+1}(x)\widetilde\pi_{2J}(x)\upd\mu_\R(x). 
\end{align*}
The same computation as in \eqref{l;ds} tells us that we need to estimate the quantity
\begin{equation}
\label{;qefsd}
\frac2{\pi s}\sum_{n=0}^{J-1}\frac{s_{2n}}{s_{2J}}\sum_{i=0}^J\frac{\Gamma(i+3/2)}{\Gamma(i+1)}\frac{\Gamma(J-i+1/2)}{\Gamma(J-i+1)}\frac{\Gamma(s-i)}{\Gamma(s-i-1/2)}\frac{\Gamma(s-i-n-3/2)}{\Gamma(s-i-n)},
\end{equation}
where we dispensed with the term $\sum_{n=0}^{J-1}\frac{s_{2n}}{s_{2J}}\frac1s$ as it is bounded above by $Js^{-1}$. Since $s>2J$, we have that
\begin{eqnarray}
\sum_{n=0}^{J-1}\frac{\Gamma\left(\frac{s-1}2-n\right)}{\Gamma\left(\frac s2-n\right)}\frac{\Gamma(s-i-n-3/2)}{\Gamma(s-i-n)} &\leq& \sum_{n=0}^{J-1}\frac{\Gamma(J-n-1/2)}{\Gamma(J-n)}\frac{\Gamma(s-i-n-3/2)}{\Gamma(s-i-n)} \nonumber \\
&=& \sum_{j=0}^{J-1}\frac{\Gamma(j+1/2)}{\Gamma(j+1)}\frac{\Gamma(j+s-J-i-1/2)}{\Gamma(j+s-J-i+1)} \nonumber \\
&=& \frac2{s-J-i-1/2}\frac{\Gamma(J+1/2)}{\Gamma(J)}\frac{\Gamma(s-i-1/2)}{\Gamma(s-i)} \nonumber
\end{eqnarray}
by Lemma~\ref{lem:thelast} applied with $x=s-J-i-1/2$. Hence, \eqref{;qefsd} is bounded above by
\[
\frac2{\pi s}\frac{\Gamma(J+1/2)}{\Gamma(J)}\frac{\Gamma\left(\frac{s-N+1}2\right)}{\Gamma\left(\frac{s-N}2\right)}\sum_{i=0}^J\frac{\Gamma(i+3/2)}{\Gamma(i+1)}\frac{\Gamma(J-i+1/2)}{\Gamma(J-i+1)}\frac2{s-J-i-1/2}
\]
and respectively, upon replacing $s$ by $2J$, it is bounded above by
\[
\frac4{\pi s}\frac{\Gamma(J+1/2)}{\Gamma(J)}\frac{\Gamma\left(\frac{s-N+1}2\right)}{\Gamma\left(\frac{s-N}2\right)}\sum_{i=0}^J\frac{\Gamma(i+3/2)}{\Gamma(i+1)}\frac{\Gamma(J-i-1/2)}{\Gamma(J-i+1)} \ll \frac{\sqrt{N(s-N)}}s,
\]
where we used \eqref{ckalphaas} and the fact that the sum on the left-hand side of the above inequality is nothing else but $P_J^{1/2,-3/2}(1)$, which is equal to 1 according to Proposition~\ref{prop:asymptoticat1}. This finishes the proof of \eqref{andthisone} and therefore of the theorem.
\end{proof}

\section{Acknowledgments}

This paper was many years in the making, and the first author (CDS)
talked with a number of people during the course of its development.
Firstly, we thank Alexei Borodin and Percy Deift for enlightening
early conversations (this was before the derivation of the
skew-orthogonal polynomials and kernel asymptotics for Ginibre's real
ensemble, and before I knew anything of random matrices, matrix
kernels and the ilk; Alexei and Percy were very patient in guiding me
down this path).  We also thank Brian Rider who has visited this
problem with us on and off over the last half dozen years.  Since this
problem is ultimately a spin-off of the first author's thesis, albeit
one very far from the original intentions of said thesis, we would be
remiss if we did not thank Jeff Vaaler for many years of helpful and
encouraging conversations revolving around Mahler's measure and its
various star bodies. 

\appendix

\section{Random Normal Matrices}
\label{sec:rand-norm-matr}

Normal matrices are square matrices which commute with their adjoints.
That is, a matrix $\mathbf Z \in \C^{N \times N}$ is normal if
$\mathbf{Z Z}^{\ast} = \mathbf{Z}^{\ast} \mathbf{Z}$.  By the spectral
theorem, if $\mathbf Z$ is normal, there exist a unitary $N \times N$
matrix $\mathbf U$ and a diagonal matrix $\bs \Lambda$ such that
\begin{equation}
\label{eq:10}
\mathbf Z = \mathbf{U}^{\ast} \bs \Lambda \mathbf{U}.
\end{equation}  
Given a normal matrix $\mathbf X \in \R^{N \times N}$, there exists
an orthogonal matrix $\mathbf O$ and a block diagonal matrix $\bs
\Gamma$ such that
\begin{equation}
\label{eq:11}
\mathbf X = \mathbf O^{\transpose}  \bs \Gamma \mathbf O,
\qquad \mbox{where} \qquad 
\bs \Gamma = \begin{bmatrix}
\alpha_1 \\ & \ddots \\ & & \alpha_L \\ & & & \mathbf B_1 \\ & & & & \ddots \\
& & & & & \mathbf B_M 
\end{bmatrix},
\end{equation}
and the $\alpha_{\ell} \in \R$ and the $\mathbf B_m \in \R^{2 \times 2}$ are
of the form 
\[
\mathbf B_m = \begin{bmatrix}
x_m & y_m \\
-y_m & x_m
\end{bmatrix}.
\]
Clearly, in this situation, $L + 2M = N$.  

We will denote the set of complex normal
matrices by $\mathcal N_N(\C)$ and the set of real normal matrices by
$\mathcal N_N(\R)$.   $\mathcal N_N(\C)$ and $\mathcal N_N(\R)$ are
naturally embedded in $\C^{N \times N}$ and $\R^{N \times N}$
respectively.  The canonical metrics on $\C^{N \times N}$ and $\R^{N
  \times N}$ induce metrics on $\mathcal N_N(\C)$ and $\mathcal
N_N(\R)$, and from these induced metrics we arrive at natural volume
forms in these sets.  These volume forms in turn induce measures on $\mathcal
N_N(\C)$ and $\mathcal N_N(\R)$ which we will denote by
$\theta_{\C}$ and $\theta_{\R}$ respectively.  

Equations (\ref{eq:10}) and (\ref{eq:11}) yield {\em spectral
  parametrizations} of $\mathcal N_N(\C)$ and $\mathcal N_N(\R)$---the
coordinates of which we refer to as {\em spectral variables}.  Among
the spectral variables are those which represent the eigenvalues of
normal matrices.  The remaining variables are derived from the
eigenvectors.  In the case of $\mathcal N_N(\C)$ we may produce a
canonical measure $\xi_{N}$ on the sets of eigenvalues (as identified
with $\C^N$) by integrating the pull back of $\theta_{\C}$ under the
spectral parametrization with respect to the eigenvalue coordinates
over the entire unitary group.  This removes any dependency of the
measure $\theta_{\C}$ on $\mathbf U$, and what we find is that $\xi_N$
encodes the local behavior of sets of eigenvalues of matrices in
$\mathcal N_N(\C)$.  As we shall see, $\xi_N$ is absolutely continuous
with respect to Lebesgue measure on $\C^N$ and its Radon-Nikodym
derivative is the familiar-looking Vandermonde term which demonstrates
how the eigenvalues of random normal matrices tend to repel each other.

We may likewise produce a canonical measure on the set of
eigenvalues of real normal matrices as identified with 
\[
\bigcup_{(L, M) \atop L + 2M = N} \R^L \times \C^M.
\]
For a particular pair $(L,M)$ such that $L + 2M = N$, we will call
$\R^L \times \C^M$ a {\em sector} of the space of eigenvalues.  In the
case of real normal ensembles, the canonical measure on the set of
eigenvalues induces measures $\xi_{L,M}$ on each of the sectors $\R^L
\times \C^M$, and we shall see that $\xi_{L,M}$ is absolutely
continuous with respect to Lebesgue measure on $\R^L \times \C^M$.  As
in the case of $\mathcal N_N(\C)$, the Radon-Nikodym derivative of
$\xi_{L,M}$ with respect to Lebesgue measure on $\R^L \times \C^M$
demonstrates repulsion among the eigenvalues of random real normal
matrices.    

\subsection{The Spectral Parametrization and the Induced Measure on
  Eigenvalues }

\subsubsection{Complex Normal Matrices}
\label{sec:compl-norm-matr}

The joint density of eigenvalues of complex normal matrices
\cite{MR1643533} is well known, but we recall the derivation here as
it motivates the discussion of real normal matrices.  See also
\cite{MR652932} for an exposition on calculations of this flavor.

$\mathcal N_N(\C)$ inherits the Hermitian metric from $\C^{N \times
  N}$  given by $\Tr( d \mathbf Z \, d \mathbf Z^{\ast} )$ where

\[
\mathbf Z = \left[ z_{m,n} \right]_{m,n=1}^N \qquad \mbox{and} \qquad d\mathbf Z = \left[ dz_{m,n} \right]_{m,n=1}^N.
\]
First we write
\[
\Tr(d \mathbf Z \, d \mathbf Z^{\ast}) = \Tr( \mathbf U^{\ast} \mathbf
U \, d \mathbf Z \, \mathbf U^{\ast} \mathbf U \, d\mathbf Z^{\ast})
= \Tr(  \mathbf U \, d \mathbf Z \, \mathbf U^{\ast} \mathbf U \,
d\mathbf Z^{\ast} \mathbf U^{\ast}).
\]
Using the change of variables (\ref{eq:10}), we have 
\[
\mathbf U \, d\mathbf Z \, \mathbf U^{\ast}=  d\bs \Lambda
 + \mathbf U \, d\mathbf U^{\ast} \, \bs \Lambda  + \bs \Lambda
 \, d\mathbf U \, \mathbf U^{\ast}.
\]
Since $\mathbf U \mathbf U^{\ast} = \mathbf I$, 
\[
d\mathbf S := \mathbf U \, d\mathbf U^{\ast} = - d\mathbf U \, \mathbf
U^{\ast},
\]
and hence
\[
\mathbf U \, d\mathbf Z \, \mathbf U^{\ast}=  d\bs \Lambda
 + d\mathbf S \, \bs \Lambda  - \bs \Lambda \, d\mathbf S = d\bs
 \Lambda + \left[ d\mathbf S, \bs \Lambda \right],
\]
where the brackets in the latter expression represent the commutator.
Clearly then,
\[
\mathbf U \, d \mathbf Z \, \mathbf U^{\ast} \mathbf U \,
d\mathbf Z^{\ast} \mathbf U^{\ast} = d\bs \Lambda \, d\bs
\Lambda^{\ast} + \left[ d\mathbf S, \bs \Lambda\right] d\bs \Lambda^{\ast} +
  d\bs \Lambda \left[ d\mathbf S, \bs \Lambda^{\ast}\right] + \left[
      d\mathbf S, \bs \Lambda\right] \left[ d\mathbf S, \bs \Lambda^{\ast}\right].
\]
It is easily seen that $\Tr \left[ d\mathbf S, \bs \Lambda \right] =
0$ and therefore $\Tr \big( \left[ d\mathbf S, \bs \Lambda \right] \,
  d\bs \Lambda^{\ast} \big) = 0$.  By similar reasoning, $\Tr\big(   d\bs \Lambda
  \left[ d\mathbf S, \bs \Lambda^{\ast}\right] \big) = 0$, and hence
\begin{equation}
\label{eq:12}
\Tr(d\mathbf Z \, d\mathbf Z^{\ast}) = \Tr (d\bs \Lambda \, d\bs
\Lambda^{\ast} ) + \Tr\big(  \left[ d\mathbf S, \bs \Lambda\right] \left[
    d\mathbf S, \bs \Lambda^{\ast}\right] \big). 
\end{equation}
Setting $d \mathbf S = \left[ ds_{m,n} \right]_{m,n=1}^N$ and $\bs
\Lambda = \left[ \delta_{m,n} \lambda_m \right]_{m,n=1}^N$, then 
\[
[d\mathbf S, \bs \Lambda] = \big[ d s_{m,n} (\lambda_n - \lambda_m)
\big]_{m,n=1}^N \qquad \mbox{and} \qquad [d\mathbf S, \bs \Lambda^{\ast}] =
\big[ d s_{m,n} (\overline{\lambda}_n - \overline{\lambda}_m) \big]_{m,n=1}^N,
\]
and
\begin{align*}
 \Tr \left( [d\mathbf S, \bs \Lambda][d\mathbf S, \bs \Lambda^{\ast}]
  \right)
& = \sum_{m=1}^N \sum_{n=1}^N ds_{m,n} (\lambda_n - \lambda_m) \,
  ds_{n,m} (\overline{\lambda}_m - \overline{\lambda}_n) \\
& = -\sum_{m<n} ds_{m,n} \,  ds_{n,m} |\overline{\lambda}_m -
\overline{\lambda}_n|^2 - \sum_{m > n} ds_{m,n}  \,
  ds_{n,m} |\overline{\lambda}_m - \overline{\lambda}_n|^2.
\end{align*}
Using the fact that $d\mathbf S$ is antihermitian, that is $ds_{n,m} =
- d\overline{s}_{m,n}$, we find
\begin{align*}
 \Tr \left( [d\mathbf S, \bs \Lambda][d\mathbf S, \bs \Lambda^{\ast}]
 \right) = 2 \sum_{m<n} | ds_{m,n} |^2 |\overline{\lambda}_m -
\overline{\lambda}_n|^2,
\end{align*}
and hence, by (\ref{eq:12}),
\[
 \Tr \left(d \mathbf Z \, d\mathbf Z^{\ast} \right) = \sum_{n=1}^N | d
 \lambda_n |^2 + 2 \sum_{m<n} | ds_{m,n} 
 |^2 |{\lambda}_m - {\lambda}_n|^2.
\]
Finally, we define the vector of `spectral variables' 
\[
d \mathbf v = (d\lambda_1, \ldots, d\lambda_N, ds_{1,2}, \ldots,
ds_{1,N}, ds_{2,3}, \ldots, ds_{2,N}, \ldots, ds_{N-1,N})
\]
and thus,
\[
\Tr(d \mathbf Z \, d\mathbf Z^{\ast}) = d\overline{\mathbf
  v}^{\transpose} \, \mathbf G \, d{\mathbf v}, 
\]
where $\mathbf G$ is the Hermitian metric tensor
\[
\begin{bmatrix}
\, \mathbf I \\
& 2 | \lambda_1 - \lambda_2 |^2 \\
& & \ddots \\
& & & 2 | \lambda_1 - \lambda_N |^2 \\
& & & & 2 | \lambda_2 - \lambda_3 |^2 \\
& & & & & \ddots \\
& & & & & & 2 | \lambda_2 - \lambda_N|^2 \\
& & & & & & & \ddots \\
& & & & & & & & 2 | \lambda_{N-1} - \lambda_N|^2
\end{bmatrix}
\]
and $\mathbf I$ is the $N \times N$ identity matrix.  The metric
tensor induces a volume form on $\mathcal N_N(\C)$ as parametrized by
the spectral variables as given by
\[
d\omega = | \det \mathbf G | \bigg\{ \bigwedge_{n=1}^N d\lambda_n \bigg\} \wedge
\bigg\{ \bigwedge_{m < n} ds_{m,n} \bigg\},
\]
and it is easily seen that 
\[
\det \mathbf G = 2^{N(N-1)/2} \bigg\{ \prod_{m<n} | \lambda_m - \lambda_n
|^2 \bigg\}.
\]
Integrating out the spectral variables which correspond to the entries
of $d \mathbf S$---and therefore only on the eigenvectors of $\mathbf
Z$---we are left with a form dependent only on the eigenvalues.  That
is, there exists a constant $C_N$, depending only on $N$, so that the
induced volume form on eigenvalues is given by
\begin{equation}
\label{eq:13}
d\omega_{\mathrm{eigs}} = C_N \bigg\{ \prod_{m < n} | \lambda_m -
\lambda_n |^2 \bigg\} d\lambda_1 \wedge \cdots \wedge d\lambda_N.
\end{equation}

\subsubsection{Real Normal Matrices}

Here we fix $L$ and $M$ so that $L + 2M = N$ and we suppose that
$\mathbf X, \mathbf O$ and $\bs \Gamma$ are given as in (\ref{eq:11}).
Next we set $\beta_m = x_m + \mathrm{i} y_m; m=1,2,\ldots,
M$ and define the $N \times N$ matrices
\[
\bs \Lambda = \begin{bmatrix}
\alpha_1 \\
& \ddots \\
& & \alpha_L \\
& & & \beta_1 \\
& & & & \overline{\beta}_1 \\
& & & & & \ddots \\
& & & & & & \beta_M \\
& & & & & & & \overline{\beta_M}
\end{bmatrix},
\] 
and
\[
\mathbf Y = \begin{bmatrix}
1 \\
& \ddots \\
& & 1 \\
& & & \mathbf C \\
& & & & \ddots \\
& & & & & \mathbf C
\end{bmatrix} \qquad \mbox{where} \qquad \mathbf{C} = \frac{1}{\sqrt
  2}  
\begin{bmatrix} 
1 & -\mathrm{i} \\
-\mathrm{i} & 1
\end{bmatrix};
\]
the upper left block of $\mathbf Y$ is the $L \times L$ identity
matrix, where the lower right block is block diagonal consisting of
$M$ non-zero blocks.  It is easily seen that $\mathbf Y$ is unitary,
and $\mathbf Y \bs \Gamma \mathbf Y^{\ast} = \bs \Lambda$.  That is,
if we define $\mathbf U = \mathbf O \mathbf Y$, 
\[
\mathbf X = \mathbf U \bs \Lambda \mathbf U^{\ast}.
\]
It follows from arguments in Section~\ref{sec:compl-norm-matr} that
\begin{align*}
\Tr(d \mathbf X \, d\mathbf X^{\ast}) &= \sum_{\ell=1}^L d
\alpha_{\ell}^2 + 2 \sum_{m=1}^M | d\beta_m |^2 + 2 \sum_{j < k}^L |d
s_{j,k}|^2 | \alpha_j - \alpha_k |^2 \\ &+ 2 \sum_{\ell=1}^L \sum_{m=1}^M
ds_{\ell, L + 2m-1} \left| \alpha_{\ell} - \beta_m\right|^2 + ds_{\ell, L +
    2m} \left| \alpha_{\ell} - \overline \beta_m \right|^2 \\
& + 2\sum_{m < n}^M ds_{L + 2m-1, L + 2n-1} | \beta_m - \beta_n |^2 +
ds_{L + 2m-1, L + 2n} | \beta_m - \overline{\beta}_n |^2 \\
& + 2\sum_{m < n}^M ds_{L + 2m, L + 2n} | \overline \beta_m -
\overline \beta_n |^2 + ds_{L + 2m, L + 2n-1} | \overline \beta_m -
{\beta}_n |^2 \\
& + 2 \sum_{m=1}^M |\beta_m - \overline \beta_m |^2.
\end{align*}

We observe that, 
\[
d\mathbf S = \mathbf U \, d\mathbf U^{\ast} = \mathbf O \mathbf Y \, d(
\mathbf O \mathbf Y)^{\ast} = \mathbf O \mathbf Y \mathbf Y^{\ast} \, d
\mathbf O^{\transpose} = \mathbf O \, d\mathbf O^{\transpose},
\]
and consequently $d \mathbf S$ is independent of $\mathbf Y$.  We also
note that $|d \beta_m |^2 = dx_m^2 + d y_m^2$.

Like in the case of complex normal matrices, we introduce spectral
variables
\[
d\alpha_1, \ldots, d\alpha_L, dx_1, dy_1, \ldots,
dx_M, dy_M, ds_{1,2}, \ldots, ds_{1,N},
ds_{2,3}, \ldots, ds_{2,N}, \ldots, ds_{N-1,N}.
\]
It is easy to compute the Riemannian metric with respect to these
variables, and the volume form on $\mathcal N_N(\R)$ is given by
\[
\omega = \sqrt{ | \det \mathbf G | } \bigg\{ \bigwedge_{\ell=1}^L
d\alpha_{\ell} \bigg\} \wedge \bigg\{ \bigwedge_{m=1}^M dx_m
\wedge dy_m \bigg\}\wedge \bigg\{ \bigwedge_{m < n}
ds_{m,n} \bigg\},
\]
where
\begin{align*}
| \det \mathbf G | &= 2^{N(N-1)/4} 2^M \bigg\{ \prod_{j < k}^L | \alpha_j -
\alpha_k | \bigg\} \bigg\{ \prod_{\ell=1}^L \prod_{m=1}^M |
\alpha_{\ell} - \beta_m |^2 \bigg\} \\
& \hspace{5cm} \times \bigg\{ \prod_{m < n}^M | \beta_m
- \beta_n |^2 \bigg\} \bigg\{ \prod_{m=1}^M 2 | \ip{\beta_m} | \bigg\}.
\end{align*}
We conclude by integrating over the entries of $d \mathbf S$ that
there is a constant $c_N$ depending only on $N$, such that the induced
volume form in eigenvalues is given by
\begin{align}
\omega_{\mathrm{eigs}} &= c_N 2^M \bigg\{ \prod_{j < k}^L | \alpha_j -
\alpha_k | \bigg\} \bigg\{ \prod_{\ell=1}^L \prod_{m=1}^M |
\alpha_{\ell} - \beta_m |^2 \bigg\} \label{eq:14} \\
& \hspace{1cm} \times \bigg\{ \prod_{m < n}^M | \beta_m
- \beta_n |^2 \bigg\} \bigg\{ \prod_{m=1}^M 2 | \ip{\beta_m} |
\bigg\} \bigg\{ \bigwedge_{\ell=1}^L d\alpha_{\ell} \bigg\} \wedge
\bigg\{ \bigwedge_{m=1}^M dx_m \wedge dy_m \bigg\}. \nonumber
\end{align}

\subsection{The Induced Measure on Eigenvalues}

Given $\bs \lambda \in \C^N$, we define the Vandermonde matrix and
determinant by
\[
\Delta(\bs \lambda) = \det \mathbf V(\bs \lambda) \qquad \mbox{where} 
\qquad 
\mathbf V(\bs \lambda) = \left[ \lambda_n^{m-1} \right]_{m,n=1}^N.
\]
More generally, given a family of monic polynomials $\mathbf p = (p_1,
p_2, \ldots, p_N)$ with $\deg p_n = n-1$ we define the Vandermonde
matrix for the family $\mathbf p$ by
\[
\mathbf V^{\mathbf p}(\bs \lambda) = \left[
p_m(\lambda_n)
\right]_{m,n=1}^N.
\]
We will call such a family of polynomials a {\em complete} set of
monic polynomials.  It is easily seen that 
\[
\det \mathbf V^{\mathbf p}(\bs \lambda) = \det \mathbf V^{\mathbf
  p}(\bs \lambda) = \prod_{m < n} (\lambda_n - \lambda_m) = \Delta(\bs
\lambda).  
\]
It follows that, in the case of $\mathcal N_N(\C)$, the measure on
eigenvalues $\xi_N$ induced by (\ref{eq:13}) is given by
\[
\upd \xi_N( \bs \lambda) = C_N | \Delta(\bs \lambda) |^2 \, \upd\mu_{\C}^N(\bs \lambda),
\]
where $\mu_{\C}^N$ is Lebesgue measure on $\C^N$. 

Similarly, for $\mathcal N_N(\R)$, and the sector of eigenvalues
represented by $\R^L \times \C^M$:  Given $\bs \beta = (\beta_1, \ldots,
\beta_M) \in \C^M$ and $\bs \alpha = (\alpha_1,
\ldots, \alpha_L) \in \R^L$ then the measure on eigenvalues
$\xi_{L,M}$ induced by (\ref{eq:14}) is given by
\[
\upd\xi_{L,M}(\bs \alpha, \bs \beta) = c_N 2^M | \Delta(\bs \alpha, \bs
\beta) | \, \upd\mu_{\R}^L(\bs \alpha) \, \upd\mu_{\C}^M(\bs \beta),
\]
where $\Delta(\bs \upalpha, \bs \upbeta)$ is the
Vandermonde determinant in the variables $\alpha_1, \ldots, \alpha_L$,
$\beta_1, \overline{\beta}_1, \ldots, \beta_M, \overline{\beta}_M$.  

\bibliography{bibliography}

\begin{thebibliography}{10}

\bibitem{oxford}
G.~Akemann, J.~Baik, and P.~Di~Francesco.
\newblock {\em The Oxford handbook of random matrix theory}.
\newblock Oxford University Press, Oxford, 2011.

\bibitem{MR1470340}
Yuri Bilu.
\newblock Limit distribution of small points on algebraic tori.
\newblock {\em Duke Math. J.}, 89(3):465--476, 1997.

\bibitem{MR1576817}
A.~Bloch and G.~P{\'o}lya.
\newblock On the {R}oots of {C}ertain {A}lgebraic {E}quations.
\newblock {\em Proc. London Math. Soc.}, S2-33(1):102.

\bibitem{borodin-2007}
Alexei Borodin and Christopher~D. Sinclair.
\newblock Correlation functions of ensembles of asymmetric real matrices, 2007.
\newblock Accepted for publication in {\it Comm. Math. Phys.}

\bibitem{borodin-2008}
Alexei Borodin and Christopher~D. Sinclair.
\newblock The {G}inibre ensemble of real random matrices and its scaling
  limits.
\newblock {\em Comm. Math. Phys.}, 291(1):177--224, 2009.

\bibitem{MR1643533}
Ling-Lie Chau and Oleg Zaboronsky.
\newblock On the structure of correlation functions in the normal matrix model.
\newblock {\em Comm. Math. Phys.}, 196(1):203--247, 1998.

\bibitem{MR1868596}
Shey-Jey Chern and Jeffrey~D. Vaaler.
\newblock The distribution of values of {M}ahler's measure.
\newblock {\em J. Reine Angew. Math.}, 540:1--47, 2001.

\bibitem{MR1915821}
Amir Dembo, Bjorn Poonen, Qi-Man Shao, and Ofer Zeitouni.
\newblock Random polynomials having few or no real zeros.
\newblock {\em J. Amer. Math. Soc.}, 15(4):857--892 (electronic), 2002.

\bibitem{MR0143558}
Freeman~J. Dyson.
\newblock Statistical theory of the energy levels of complex systems. {III}.
\newblock {\em J. Mathematical Phys.}, 3:166--175, 1962.

\bibitem{MR0278668}
Freeman~J. Dyson.
\newblock Correlations between eigenvalues of a random matrix.
\newblock {\em Comm. Math. Phys.}, 19:235--250, 1970.

\bibitem{MR1231689}
Alan Edelman, Eric Kostlan, and Michael Shub.
\newblock How many eigenvalues of a random matrix are real?
\newblock {\em J. Amer. Math. Soc.}, 7(1):247--267, 1994.

\bibitem{MR0033372}
P.~Erd{\H{o}}s and P.~Tur{\'a}n.
\newblock On the distribution of roots of polynomials.
\newblock {\em Ann. of Math. (2)}, 51:105--119, 1950.

\bibitem{MR0073870}
Paul Erd{\H{o}}s and A.~C. Offord.
\newblock On the number of real roots of a random algebraic equation.
\newblock {\em Proc. London Math. Soc. (3)}, 6:139--160, 1956.

\bibitem{Olver}
F.W.J~Olver et~al., editor.
\newblock {\em NIST handbook of mathematical functions}.
\newblock Cambridge University Press, Cambridge, 2010.

\bibitem{mays-forrester}
Peter Forrester and Anthony Mays.
\newblock A method to calculate correlation functions for {$\beta=1$} random
  matrices of odd size.
\newblock {\em J. Statist. Phys.}, 134(3):443--462, 2009.

\bibitem{Gaudin1961447}
Michel Gaudin.
\newblock Sur la loi limite de l'espacement des valeurs propres d'une matrice
  ale{\^a}´atoire.
\newblock {\em Nuclear Physics}, 25(0):447 -- 458, 1961.

\bibitem{MR0084888}
J.~M. Hammersley.
\newblock The zeros of a random polynomial.
\newblock In {\em Proceedings of the {T}hird {B}erkeley {S}ymposium on
  {M}athematical {S}tatistics and {P}robability, 1954--1955, vol. {II}}, pages
  89--111, Berkeley and Los Angeles, 1956. University of California Press.

\bibitem{MR2552864}
J.~Ben Hough, Manjunath Krishnapur, Yuval Peres, and B{\'a}lint Vir{\'a}g.
\newblock {\em Zeros of {G}aussian analytic functions and determinantal point
  processes}, volume~51 of {\em University Lecture Series}.
\newblock American Mathematical Society, Providence, RI, 2009.

\bibitem{MR2422348}
C.~P. Hughes and A.~Nikeghbali.
\newblock The zeros of random polynomials cluster uniformly near the unit
  circle.
\newblock {\em Compos. Math.}, 144(3):734--746, 2008.

\bibitem{MR1390040}
Ildar Ibragimov and Ofer Zeitouni.
\newblock On roots of random polynomials.
\newblock {\em Trans. Amer. Math. Soc.}, 349(6):2427--2441, 1997.

\bibitem{MR0007812}
M.~Kac.
\newblock On the average number of real roots of a random algebraic equation.
\newblock {\em Bull. Amer. Math. Soc.}, 49:314--320, 1943.

\bibitem{MR0030713}
M.~Kac.
\newblock On the average number of real roots of a random algebraic equation.
  {II}.
\newblock {\em Proc. London Math. Soc. (2)}, 50:390--408, 1949.

\bibitem{kron1957}
L.~Kronecker.
\newblock Zwei s\"atze \"uber gleichungen mit ganzzahligen coefficienten.
\newblock {\em Journal f{\"u}r die reine und angewandte Mathematik},
  53:173--175, 1857.

\bibitem{MR2932626}
A.~B.~J. Kuijlaars.
\newblock Universality.
\newblock In {\em The {O}xford handbook of random matrix theory}, pages
  103--134. Oxford Univ. Press, Oxford, 2011.

\bibitem{MR812558}
Michel Langevin.
\newblock M\'ethode de {F}ekete-{S}zeg{\H o} et probl\`eme de {L}ehmer.
\newblock {\em C. R. Acad. Sci. Paris S\'er. I Math.}, 301(10):463--466, 1985.

\bibitem{MR1503118}
D.~H. Lehmer.
\newblock Factorization of certain cyclotomic functions.
\newblock {\em Ann. of Math. (2)}, 34(3):461--479, 1933.

\bibitem{MR1574980}
J.~E. Littlewood and A.~C. Offord.
\newblock On the {N}umber of {R}eal {R}oots of a {R}andom {A}lgebraic
  {E}quation.
\newblock {\em J. London Math. Soc.}, S1-13(4):288.

\bibitem{PSP:2031108}
J.~E. Littlewood and A.~C. Offord.
\newblock On the number of real roots of a random algebraic equation. ii.
\newblock {\em Mathematical Proceedings of the Cambridge Philosophical
  Society}, 35:133--148, 3 1939.

\bibitem{Lub09}
D.S. Lubinsky.
\newblock Universality limits for random matrices and de {B}ranges spaces of
  entire functions.
\newblock {\em J. Funct. Anal.}, 256:3688---3729, 2009.

\bibitem{Lub12}
D.S. Lubinsky.
\newblock Bulk universality holds in measure for compactly supported measures.
\newblock {\em J. Analyse Math\'ematique}, 116:219---253, 2012.

\bibitem{MR0112645}
M.~L. Mehta.
\newblock On the statistical properties of the level-spacings in nuclear
  spectra.
\newblock {\em Nuclear Phys.}, 18:395--419, 1960.

\bibitem{MR0277221}
M.~L. Mehta.
\newblock A note on correlations between eigenvalues of a random matrix.
\newblock {\em Comm. Math. Phys.}, 20:245--250, 1971.

\bibitem{MR0112895}
M.~L. Mehta and M.~Gaudin.
\newblock On the density of eigenvalues of a random matrix.
\newblock {\em Nuclear Phys.}, 18:420--427, 1960.

\bibitem{MR2129906}
Madan~Lal Mehta.
\newblock {\em Random matrices}, volume 142 of {\em Pure and Applied
  Mathematics (Amsterdam)}.
\newblock Elsevier/Academic Press, Amsterdam, third edition, 2004.

\bibitem{MR652932}
Robb~J. Muirhead.
\newblock {\em Aspects of multivariate statistical theory}.
\newblock John Wiley \& Sons Inc., New York, 1982.
\newblock Wiley Series in Probability and Mathematical Statistics.

\bibitem{MR1334766}
Thomas Ransford.
\newblock {\em Potential theory in the complex plane}, volume~28 of {\em London
  Mathematical Society Student Texts}.
\newblock Cambridge University Press, Cambridge, 1995.

\bibitem{MR0360515}
A.~Schinzel.
\newblock On the product of the conjugates outside the unit circle of an
  algebraic number.
\newblock {\em Acta Arith.}, 24:385--399, 1973.
\newblock Collection of articles dedicated to Carl Ludwig Siegel on the
  occasion of his seventy-fifth birthday. IV.

\bibitem{MR1308023}
Larry~A. Shepp and Robert~J. Vanderbei.
\newblock The complex zeros of random polynomials.
\newblock {\em Trans. Amer. Math. Soc.}, 347(11):4365--4384, 1995.

\bibitem{sinclair-2007}
Christopher~D Sinclair.
\newblock Averages over {G}inibre's ensemble of random real matrices.
\newblock {\em Int. Math. Res. Not.}, 2007:1--15, 2007.

\bibitem{sinclair-2005}
Christopher~D. Sinclair.
\newblock The range of multiplicative functions on $\mathbb{C}[x],
  \mathbb{R}[x]$ and $\mathbb{Z}[x]$.
\newblock {\em Proc. London Math. Soc.}, 96(3):697--737, 2008.

\bibitem{sinclair-2008}
Christopher~D. Sinclair.
\newblock Correlation functions for $\beta$=1 ensembles of matrices of odd
  size.
\newblock {\em J. Stat. Phys.}, 136(1):17--33, 2009.

\bibitem{Sinclair2012682}
Christopher~D. Sinclair and Maxim~L. Yattselev.
\newblock Universality for ensembles of matrices with potential theoretic
  weights on domains with smooth boundary.
\newblock {\em Journal of Approximation Theory}, 164(5):682 -- 708, 2012.

\bibitem{MR0289451}
C.~J. Smyth.
\newblock On the product of the conjugates outside the unit circle of an
  algebraic integer.
\newblock {\em Bull. London Math. Soc.}, 3:169--175, 1971.

\bibitem{MR648108}
Peter Walters.
\newblock {\em An introduction to ergodic theory}, volume~79 of {\em Graduate
  Texts in Mathematics}.
\newblock Springer-Verlag, New York, 1982.

\bibitem{uY2}
M.~Yattselev.
\newblock Large deviations and linear statistics for potential theoretic
  ensembles associated with regular closed sets.
\newblock \emph{accepted for publication in Probab. Theory Relat. Fields.}
  \url{http://arxiv.org/abs/ 1207.0718}.

\bibitem{Yuzvinski:1965lr}
S.~Yuzvinski.
\newblock Metric properties of endomorphisms of compact groups.
\newblock {\em Izv. Acad. Nauk SSSR, Ser. Mat.}, 29:1295--1328, 1965.

\end{thebibliography}

\begin{center}
\noindent\rule{4cm}{.5pt}
\vspace{.25cm}

\noindent {\sc \small Christopher D.~Sinclair}\\
{\small Department of Mathematics, University of Oregon, Eugene OR 97403} \\
email: \href{mailto:csinclai@uoregon.edu}{\tt csinclai@uoregon.edu}
\vspace{.25cm}

\noindent {\sc \small Maxim L.~Yattselev}\\
{\small Department of Mathematical Sciences, Indiana University	 --- Purdue University Indianapolis, 402~North Blackford Street, Indianapolis, IN 46202} \\
email: \href{mailto:maxim@math.iupui.edu}{\tt maxim@math.iupui.edu} 
\end{center}

\end{document}